\documentclass[journal,twocoloumn]{IEEEtran}
\usepackage{dblfloatfix}
\usepackage{graphicx}
\usepackage{cite}
\usepackage{amsmath}
\usepackage{amssymb}
\usepackage{amsthm}
\usepackage{float}
\usepackage{subcaption}
\usepackage{placeins}
\usepackage{algorithmic}
\usepackage{algorithm}
\usepackage{breakcites}
\usepackage[hyphens]{url}

\theoremstyle{plain}
\newtheorem{theorem}{Theorem}
\theoremstyle{definition}
\newtheorem{assumption}{Assumption}
\newtheorem{proposition}{Proposition}
\theoremstyle{definition}
\newtheorem{remark}{Remark}

\makeatletter
\def\@copyrightspace{\relax}
\makeatother

\begin{document}

\title{Terminal-Angle-Constrained Guidance based on Sliding Mode Control for UAV Soft Landing on Ground Vehicles}

\author{Sashank~Modali,~Satadal~Ghosh,~Sujit~P.B.
\thanks{Sashank Modali and Satadal Ghosh are with the Department of Aerospace Engineering at Indian Institute of Technology Madras, Chennai, 600036 India (e-mail: ae16b031@smail.iitm.ac.in; satadal@iitm.ac.in).}
\thanks{Sujit P.B. is with the Department of Electrical Engineering and Computer Science at  Indian Institute of Science Education and Research Bhopal, Bhopal, 462066 India (e-mail: sujit@iiserb.ac.in).}}

\maketitle

\begin{abstract}
In this paper the problem of guidance formulation for autonomous soft landing of unmanned aerial vehicles on stationary, moving, or accelerating / maneuvering ground vehicles at desired approach angles in both azimuth and elevation is considered. Nonlinear engagement kinematics have been used. While integrated nonlinear controllers have been developed in the literature for this purpose, in practical implementations the controller inputs often need modification of the existing autopilot structure, which is challenging. In order to avoid that a higher-level guidance algorithm is designed in this paper leveraging sliding mode control-based approach. In the presented guidance formulation, target-state-dependent singularity can be avoided in the guidance command. The effectiveness of the presented guidance law is verified with numerical simulation studies. However, since the algorithm in its basic form is found to demand high guidance command at large distances from maneuvering ground targets, a two-phase guidance is presented next to avoid this problem and validated with numerical simulations. Finally, the efficacy of the modified guidance algorithm is validated by Software-In-The-Loop simulations for a realistic testbed. 

\end{abstract}

\begin{IEEEkeywords}
Autonomous Landing , Guidance , Sliding mode control , Stationary, Moving or Accelerating target, Approach angle
\end{IEEEkeywords}

\section{Introduction}
Unmanned aerial vehicles (UAVs) have become essential for civilian applications \cite{geng2013mission,waharte2010supporting} due to simplicity in their operations and ease of availability in the market. Often they are used in critical applications like search and rescue \cite{searchandrescue}, product delivery 
operations \cite{uavdelivery}, etc. These vehicles often need to land accurately in constrained environments while performing these operations. Consequently, the current practice of manual intervention during landing may not be feasible in such circumstances. Further, due to environmental restrictions or threats in current applications of UAVs, especially in civilian domain, terminal/approach angle constrained landing is desired, in which it is of interest to approach the target following a pre-specified direction \cite{allaspectstationaryiac}. For example, in applications like the autonomous delivery of commodities to a moving truck, or shifting of items in the shop floor of a manufacturing unit, achieving a desired terminal angle is crucial for mission success. In such cases, executing the landing is challenging even with manual intervention. Therefore, there is a need to develop autonomous landing solutions for UAVs for accurate and soft landing on stationary as well as on moving (accelerating or non-accelerating) platforms at desired approach angles.

There are two components in landing -- detection of landing site (a.k.a. target) and tracking of the target \cite{gautam2014survey}. For detection of landing site, different kinds of sensing and data processing mechanisms have been presented in the literature. Transformation-invariant Hu-moments \cite{hu-og} applied on on-board camera-fed target images were leveraged in \cite{saripalli2002vision}, \cite{chandra2019hu} for this purpose, while scale-invariant feature transform of target image was used for feature matching with landing zone in \cite{cesetti2009vision}. Altimeter-feature matching and edge detection in target image were fused in \cite{campoy1} for target detection, while target image and LIDAR information were fused in \cite{arora2013infrastructure} for the same purpose. On the other hand, GPS-based identification of landing site was utilized in \cite{gps2,gps3,skulstad2015autonomous}. However, detection of landing site is beyond the scope of this paper. It is assumed in this paper that highly accurate data about the states of the target and the UAV itself are available to the UAV.

Once the target information is acquired, the next task is to design a tracking guidance and control mechanism for the aerial vehicle so that it can land on the target. To this end, guidance and control modules have been developed in the literature in both integrated \cite{ahmed2008backstepping, vlantis2015quadrotor, ghommam2017autonomous, leevisualservoing} and independent ways. An integrated backstepping controller for landing developed for flapping rotor-blade dynamics of rotary wing vehicles was presented in \cite{ahmed2008backstepping}. A nonlinear Model Predictive Control (MPC)-based approach was presented in \cite{vlantis2015quadrotor} for landing on a rover moving on an inclined platform. An adaptive tracking control scheme was developed in \cite{ghommam2017autonomous} using backstepping and dynamic surface control for quad-rotor landing. In \cite{leevisualservoing}, visual servoing was realized based on adaptive sliding mode control for the purpose of landing.

Though integrated guidance and control blocks show good performance, they often require modification of existing autopilot systems, which in practice could be challenging. Moreover, they often restrict the application of the design to a particular type of vehicle. Therefore, as an alternative, an on-board computer is planned to be utilized that determines suitable reference commands, known as guidance commands, which are sent to the autopilot. Since guidance commands are higher level commands, they only provide a reference to the lower level controllers, which in effect tries to achieve the desired guidance command. In the literature, UAV landing guidance has been formulated in several ways. Leveraging Hu-moments, proportional controller-based landing guidance was presented in \cite{saripalli2002vision}, \cite{chandra2019hu}. Glide-slope to stationary landing point from a critical altitude was studied in \cite{barber2007autonomous}, while a time-to-go-based polynomial guidance law was presented in \cite{bangguidance} for the landing of rotary and fixed-wing aerial vehicles. A spiral landing trajectory was generated in \cite{yoon2009spiral} by a pseudo-pursuit guidance law. In \cite{kim2013fully}, a landing guidance algorithm was devised using Proportional Navigation (PN) \cite{PN-og,PN-og2} in the longitudinal plane dynamics and $L_1$ guidance \cite{l1} in the lateral plane dynamics. PN and proportional-derivative controllers were coupled to develop a landing guidance scheme in \cite{borowczyk2017autonomous}, while a pure pursuit-based guidance scheme was framed in \cite{gautam2017autonomous} for quad-rotor landing on moving target.

In applications related to UAV landing, most of the existing literature has considered stationary landing platforms, while moving platforms have been considered in some of the literature \cite{campoy1, vlantis2015quadrotor, leevisualservoing, borowczyk2017autonomous, gautam2017autonomous}, but most of the considered problems involved a restricted class of targets such as non-maneuvering targets or targets modelled with linearized kinematics. To the best of the authors' knowledge, problems related to soft landing on a comprehensive class of targets that even includes maneuvering and accelerating targets as well have not yet been addressed in the literature. Thus, the problem of interest in this paper is to formulate a unified approach to landing guidance for a UAV on all types of ground targets, i.e., stationary, nonmaneuvering, maneuvering with constant speed or accelerating ground targets while considering nonlinear engagement kinematics for landing.

Additionally, in the literature related to UAV landing, the terminal direction objective has mainly been addressed in directed-runway landing \cite{gps3}, or net-recovery landing \cite{bangguidance, kim2013fully, yoon2009spiral}. This problem has been more explored in guidance literature, where several approaches have been studied for terminal/approach/impact angle control such as optimal control \cite{optimal-tac}, sliding mode control \cite{kumarNSM, kumar2012sliding}, and PN \cite{biasedpng, biasedpng2, compositehighspeedIAC, allaspectstationaryiac}. Finally, most of the literature related to rotary UAV landing except a few \cite{gautam2017autonomous, ghommam2017autonomous, slanted, Paris2019DynamicLO}, have been on vertical landing only, and this is not ideal in terms of time taken for touchdown and control effort. 
Moreover, the consideration of terminal angle (both azimuth and elevation angles) control as an objective proves to be helpful in field-of-view-constrained scenarios, and thus, also helps in eliminating any specific need for vertical landing. To this end, a novel guidance scheme inspired by the sliding mode philosophy is presented in this paper for landing a UAV on stationary, moving and accelerating / maneuvering targets at desired approach angles (both azimuth and elevation angles). Suitable consideration of multiple sliding variables also leads to avoiding target-state-dependent singularity in the guidance command, which is otherwise usually present in Sliding mode control-based guidance literature. When implemented in scenarios with a large initial range from the target, the proposed guidance scheme leads to high desired guidance commands, especially in the cases with highly maneuvering targets. To obviate this problem, the guidance scheme is further modified to be split into two phases. In the first phase, the desired azimuth angle is set as a constant, while in the second phase, which is initiated when the range in the horizontal plane is smaller than a pre-fixed threshold, the desired azimuth angle is selected relative to the target's heading angle.

The rest of the paper is organized as follows. First, the landing problem is described in Section \ref{sec:pf}. Then, the landing guidance law is designed and a detailed discussion on the synthesized guidance algorithm is presented in Section \ref{sec:guidancelaw}. The effectiveness of the proposed guidance law is demonstrated using simulation studies over different kinds of ground target platforms, and the results are presented in Section \ref{sec:results}. To avoid high guidance commands for maneuvering targets at large initial ranges, a two-phase guidance scheme is next presented in Section \ref{sec:twophase} along with numerical point mass simulations and their results. Then, to depict the algorithm's efficacy in realistic rotary vehicle test-bed, Software-In-The-Loop simulations are conducted, and the results are presented in Section \ref{sec:SITL}. Finally, conclusions and possible future works are discussed in Section \ref{sec:conclusions}.

\section{Problem Formulation}\label{sec:pf}
\subsection{Goal for Landing Guidance}
In this paper, the problem considered is on autonomous soft landing, in which a UAV is considered to initiate its movement from an arbitrary initial position in three dimensional (3-D) space not very far from an Unmanned Ground Vehicle (UGV or target) moving on the horizontal (x-y) plane in an inertial xyz-frame. A guidance algorithm is to be developed for the UAV's motion to soft-land (land smoothly) on the UGV while achieving desired approach angles with respect to both the xy-plane and the z-axis.

An illustrative engagement scenario is shown in Fig. \ref{fig:geometry}. Here, $R, R_{xy}, R_z$ represent the distance between the target and the UAV, its component when projected on the xy-plane and the vertical range between the UAV and the target, respectively. The UAV's speed, flight path angle, and heading angle in the xy-plane are denoted as $V_{p}$, $\gamma$, and $\alpha_{p}$, respectively. Besides these, $\theta$ denotes the angle between the line of sight (LOS) from the pursuer to the target and the xy-plane, and represents the elevation angle of the pursuer relative to the target. 
Further, $\psi$ denotes the angle between the orthogonal projection of the LOS onto the xy-plane and the reference x-axis. It is referred to as the azimuth angle in this paper. Here, $\gamma$ and $\theta$ are defined such that they lie in the interval, [$-\frac{\pi}{2}$, $\frac{\pi}{2}$] rad. In general, the (-$\pi$, $\pi$] rad convention of angles is followed in this paper. The UGV's speed and heading angle are denoted by $V_{t}$ and $\alpha_{t}$, respectively. The kinematics of the system is decomposed into in inertial xy-plane and vertical direction (inertial z-axis) for analytical convenience.

The objective of the guidance law is to achieve the followings:
\begin{align} \label{eq:goal}
    &\lim_{t \to \infty} R_{xy} = 0; \:\:
    \lim_{t \to \infty} R_{z} = 0;\:\:
    \lim_{t \to \infty} \dot{R}_{xy} = 0; \:\:
    \lim_{t \to \infty} \dot{R_{z}} = 0;\nonumber \\
    &\lim_{t \to \infty} (\psi - \alpha_{t}) =  \zeta_{des}; \:\:
    \lim_{t \to \infty} \theta = \theta_{des}
\end{align}

Note that the first four objectives in \eqref{eq:goal} are required for soft landing. The other two objectives dictate desired approach angles, which are defined here as desired LOS angles, one being the angle between the orthogonal projection of the LOS on the xy-plane and the heading angle of the target ($\psi$ - $\alpha_t$), and the other being the elevation angle ($\theta$). The desired values for these angles are represented by $\zeta_{des} \ (=\psi_{des}-\alpha_t)$ and $\theta_{des}$, respectively.

\begin{figure}[h!]
\begin{center}
\includegraphics[width=0.4\textwidth]{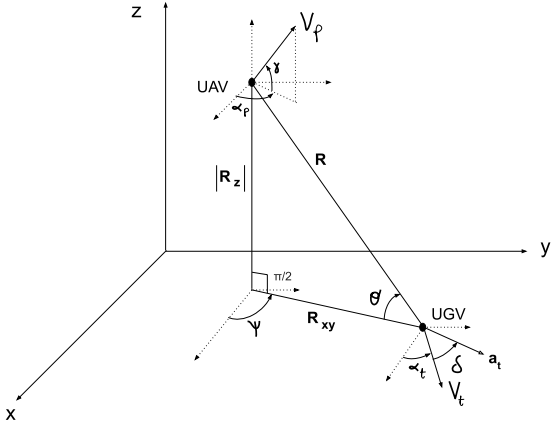}
\caption{UAV-target engagement geometry for landing} 
\label{fig:geometry}
\end{center}
\end{figure}

\subsection{Equations of Motion}

The 3-D engagement kinematics of the UAV and the ground target (UGV) are represented below in a control-affine form in Eqs. \eqref{eq:Rxydot} - \eqref{eq:input_vector}, where $\dot{V_{p}}$,  $\dot{\alpha_{p}}$, and $\dot{\gamma}$ are the guidance control inputs. 
\begin{align}
\dot{R}_{xy}=V_{t}\cos(\alpha_{t}-\psi) -  V_{p}\cos(\gamma)\cos(\alpha_{p} - \psi) \label{eq:Rxydot} &\\
\dot{R_{z}}=-V_{p}\sin(\gamma) \label{eq:Rzdot}&\\
\dot{\psi}=\frac{1}{R_{xy}}(V_{t}\sin(\alpha_{t}-\psi) - V_{p}\cos(\gamma)\sin(\alpha_{p} - \psi)) \label{eq:psidot}&\\
\dot{V_{t}}=a_t \cos(\delta)
\label{eq:vtdot}&\\
\dot{\alpha_{t}}=a_t \sin(\delta)
\label{eq:alphatdot}&\\
\left[\dot{V_{p}}\quad\dot{\alpha_{p}}\quad\dot{\gamma}\right]^T=U
\label{eq:input_vector}&
\end{align}

In addition, the kinematics of variables $\theta$ $(= \tan^{-1}(-R_z/R_{xy}))$ and $R$ $(= \sqrt{R_{xy}^2 + R_z^2})$ depend on Eqs. \eqref{eq:Rxydot} and \eqref{eq:Rzdot}, and are derived as follows :- 

\begin{align}
dot{R} = \frac{1}{R}(R_{xy}\dot{R}_{xy} + R_{z}\dot{R_{z}})
\label{eq:Rdot}&\\
\dot{\theta}=\frac{1}{R}(V_{p} \ \sin(\gamma) \ \cos(\theta) + \dot{R}_{xy} \ \sin(\theta)) \label{eq:thetadot}&
\end{align}

Here, $a_t$ denotes the target's acceleration, and $\delta$ denotes the angle between the target's thrust vector and it's heading vector at any instant, as shown in Fig. \ref{fig:geometry}.

\begin{assumption}
In this paper, it is assumed that $a_t$ and $\delta$ are piece-wise continuously differentiable in time. It is also assumed that both the components of acceleration ($a_t \cos(\delta)$, $a_{t} \sin(\delta)$), and the angular acceleration of the target ($\ddot{\alpha_t}$) at every instant are known with good accuracy.
\end{assumption}

\section{Guidance Law Design}\label{sec:guidancelaw}

\subsection{Background}\label{subsec:synthesis}

Note that for successful touchdown on the UGV, at which $R \to 0$, the collision course condition needs to be satisfied by the UAV with the UGV's motion. Moreover, in order to obtain the desired terminal velocity of the UAV, which comprises its speed ($\lim_{R \to 0} V_p$), heading angle ($\lim_{R \to 0} \alpha_p$), and flight path angle ($\lim_{R \to 0} \gamma$), for soft landing on the UGV, the soft landing objective should be considered as well. Now, the condition for collision course between the UAV and the UGV is given as,
\begin{align} \label{eq:colcourse}
    V_{t} \ \sin(\alpha_{t} - \psi) &= V_{p} \ \cos(\gamma) \ \sin(\alpha_{p} - \psi)
\end{align}
Next, for soft landing from Eq. \eqref{eq:goal} the followings need to be satisfied: $\lim_{R \to 0}\dot{R}_{xy}= 0 $ and $\lim_{R \to 0}\dot{R_{z}} = 0$. From Eqs. \eqref{eq:Rxydot} and \eqref{eq:Rzdot},
\begin{align}
 \label{eq:colcourse2}
\dot{R}_{xy} &= 0 \implies V_{t} \ \cos(\alpha_{t} - \psi) = V_{p} \ \cos(\gamma) \ \cos(\alpha_{p} - \psi)& \\
 \label{eq:colcourse3}
\dot{R}_{z} &= 0 \implies V_{p} \ \sin(\gamma) = 0&
\end{align}

If $V_t = 0$ when Eqs. \eqref{eq:colcourse}, \eqref{eq:colcourse2}, and \eqref{eq:colcourse3} are satisfied,

\begin{align}
\left.
\begin{array}{ll}
V_{p} \cos(\gamma) \cos(\alpha_p - \psi) = 0 \\
V_{p} \cos(\gamma) \sin(\alpha_p - \psi) = 0 \\
V_{p} \sin(\gamma) = 0
\end{array}
\right \}&\implies V_{p} = 0 \\
\intertext{
Else if $V_t \neq 0$,
}
\left.
\begin{array}{rr}
V_t \cos(\alpha_t - \psi) = V_{p} \cos(\gamma) \cos(\alpha_p - \psi) \\
V_t \sin(\alpha_t - \psi) = V_{p} \cos(\gamma) \sin(\alpha_p - \psi) \\
V_{p} \sin(\gamma) = 0
\end{array}
\right \}&\implies 
\begin{array}{ll}
     \gamma = 0\\
     V_{p} = V_{t}\\
     \alpha_p = \alpha_t
\end{array}
\end{align}
 
Thus, following Eq. \eqref{eq:goal} for soft landing and solving Eqs. \eqref{eq:Rzdot}, \eqref{eq:colcourse}, and \eqref{eq:colcourse2}, the desired terminal speed, heading angle, and flight path angle of the UAV are obtained as follows.
 \begin{align}
    \label{eq:desiredheadspeed}
     &{\lim_{R \to 0} (V_{p} - V_{t})=0} &\nonumber \\
     &{\lim_{R \to 0} (\alpha_{p} - \alpha_{t}) = 0 \ , \ \lim_{R \to 0} \gamma = 0} \text{, } &\text{ if } \lim_{R \to 0} V_t \neq 0
\end{align}

Besides these, from Eq. \eqref{eq:goal}, the terminal angle constraints are posed as,

\begin{equation}
    \label{eq:desiredapproach}
    \lim_{R \to 0} (\psi - \alpha_t) = \zeta_{des} \quad , \quad \lim_{R \to 0} \theta = \theta_{des}
\end{equation}

Note that as $R \to 0$, that is at touchdown, Eq. \eqref{eq:goal} and the combination of Eqs. \eqref{eq:desiredheadspeed} and \eqref{eq:desiredapproach} provide equivalent conditions for terminal angle-constrained soft landing of the UAV on the UGV.

\subsection{Synthesis of Guidance Command}\label{subsec:synthesis}

In the design of a sliding surface for control-affine systems, the number of sliding variables is taken to be equal to the number of control inputs (3 in this case), in order to obtain a linear system of equations with a unique solution. These sliding variables are selected such that desired terminal conditions and objectives of the guidance law are satisfied when the system is in sliding mode. Thus, the terminal requirements as mentioned in Eq. \eqref{eq:goal} form the basic consideration behind the formulation of sliding variables for deriving a suitable guidance law. Therefore, three sliding variables are considered with $R_{xy}$, $R_z + \tan(\theta_{des}) R_{xy}$, and $\psi - \alpha_{t} - \zeta_{des}$, respectively. Since second order derivatives of these quantities contain the control inputs $U$ to the UAV, we choose first order exponential decay dynamics for variables $R_{xy}$, $R_z + \tan(\theta_{des}) R_{xy}$, and $\psi - \alpha_{t} - \zeta_{des}$ on the sliding mode. This ensures that following the sliding mode dynamics these variables and their time derivatives converge to zero asymptotically. Thus, the sliding variables are considered as below.
\begin{small}
\begin{equation} \label{eq:Sdefn}
{S = 
\begin{bmatrix}
\dot{R}_{xy}+k_{a}R_{xy} \\ 
\dot{R_{z}}+\tan(\theta_{des})\dot{R}_{xy} + k_{b}(R_{z} + \tan(\theta_{des})R_{xy})  \\ 
(\dot{\psi} - \dot{\alpha_{t}}) + k_{c}(\psi - (\alpha_{t} + \zeta_{des}))
\end{bmatrix}
 = 
\begin{bmatrix}
(S_{1}) \\ 
(S_{2}) \\
(S_{3})
\end{bmatrix}
}
\end{equation}
\end{small}
where, $k_{a}$ , $k_{b}$ and $k_{c}$ are tuning parameters in the designed guidance law. Now, the guidance is applied such that the dynamics of sliding variables satisfies the following.
\begin{equation} \label{eq:Sdotdefn}
{\dot{S} = -
\begin{bmatrix}
k_1 & 0 & 0\\0 & k_2 & 0 \\0 & 0 & k_3
\end{bmatrix}
\begin{bmatrix}
(S_{1})^{n/m} \\ 
(S_{2})^{n/m} \\
(S_{3})^{n/m}
\end{bmatrix}
}
\end{equation}
where, $k_{1}$, $k_{2}$, $k_{3}$, $m$ and $n$ are other tuning parameters in the designed guidance law. Among all these tuning parameters, $m$ and $n$ should be odd and co-prime integers such that $0<n<m$, to ensure finite time convergence of the system dynamics to sliding mode. A discussion on selection of other tuning parameters has been presented in Section \ref{subsec:ParameterSelection}. 

\begin{theorem}\label{th:stability}
The guidance algorithm that ascertains Eq. \eqref{eq:Sdotdefn} for reaching the sliding surface $S=0$, where $S$ denotes the vector of sliding variables given in Eq. \eqref{eq:Sdefn} enables a UAV to successfully soft-land on a ground target at desired approach angles (both azimuth and elevation angles) asymptotically.
\end{theorem}

\begin{proof}
Recall from Eq. \eqref{eq:Sdotdefn} that the guidance command inputs have been designed in Section \ref{subsec:synthesis} such that the sliding variables follow: $\dot{S_{1}}=-k_1 \ (S_{1})^{n/m}$, $\dot{S_{2}}=-k_2 \ (S_{2})^{n/m}$ and $\dot{S_{3}}=-k_3 \ (S_{3})^{n/m}$.

Now, from Eq. \eqref{eq:Sdotdefn}, for $i=1,2,3$,
\begin{align}\label{eq:straj}
    {(S_{i}(t))^{\frac{(m-n)}{m}} = (S_{i}(0))^{\frac{(m-n)}{m}} -\frac{(m-n)}{m} \ k_{i} \ t} \nonumber \\
    \implies S_i(t) = 0 \ \ \forall \ \ t \geq \frac{m}{m-n}(S_{i}(0))^{\frac{(m-n)}{m}}
\end{align}
Clearly, the dynamics of the chosen sliding variables are finite time convergent, that is the sliding mode dynamics can be enforced in finite time. When the system is in Sliding mode ($S=0$), the following can be noted from Eqs. \eqref{eq:thetadot} and \eqref{eq:Sdefn},
\begin{align}
\label{eq:slidingmodedynamics}
    \dot{R_{xy}} &=& -k_a R_{xy}&  \nonumber \\
    \dot{R_{z}} + \tan(\theta_{des}) \dot{R_{xy}} &=& - k_b( R_z + \tan(\theta_{des}) R_{xy})& \nonumber \\
    \dot{\theta} &=& (k_b-k_a) \cos^2(\theta) (\tan(\theta_{des})-\tan(\theta))& \nonumber \\
    \dot{\psi} - \dot{\alpha_t} &=& k_c (\psi - \alpha_t - \zeta_{des} )&
\end{align}
Here, exponential decay of $R_{xy}$ implies that $R_{xy} \to 0$ and $\dot{R_{xy}} \to 0$ as $t \to \infty$. Then, exponential decay of $R_z + \tan(\theta_{des})R_{xy}$ along with $R_{xy} \to 0$ and $\dot{R_{xy}} \to 0$ implies that $R_{z} \to 0$ and $\dot{R_{z}} \to 0$ as $t \to \infty$. 

For the convergence of $\theta$, consider the Lyapunov function, V = $1/2 (\theta-\theta_{des})^2$. $\dot{V} = \dot{\theta} (\theta-\theta_{des})$. Assuming $k_b > k_a$, $sgn(\dot{V}) = sgn((\theta-\theta_{des})(\tan(\theta_{des}-\theta)) (\leq 0)$, where  sgn(.) denotes the signum function. Following LaSalle's invariance principle \cite{Khalil:1173048}, $\theta \to \theta_{des}$ as $t \to \infty$, which means that the desired elevation angle $(\theta_{des})$ is also achieved asymptotically. 

Finally, $\psi - (\alpha_t+\zeta_{des})$ decays exponentially, which implies that the desired azimuth angle is achieved asymptotically. 

Thus, it can be seen that the desired terminal requirements from \eqref{eq:goal} are satisfied when the chosen dynamics in \eqref{eq:Sdotdefn} are enforced on the Sliding variables defined in \eqref{eq:Sdefn}.
\end{proof}

\begin{remark}
Unlike the signum function of sliding variables, usually considered in conventional sliding mode-based guidance design, the form of sliding mode dynamics considered in Eq. \eqref{eq:Sdefn} helps to reduce chattering, and allows for a smooth finite time convergence of the sliding variables.
\end{remark}

\subsection{Guidance Command Inputs}\label{subsec:guidanceInputs}

Eqs. \eqref{eq:Rxydot} - \eqref{eq:input_vector}, \eqref{eq:Sdefn} and \eqref{eq:Sdotdefn}, when expanded lead to a system of three equations, expressed as $AU=B$, where $U\in\mathbb{R}_{3\times 1}$ is as given in Eq. \eqref{eq:input_vector}, and $A\in \mathbb{R}_{3\times 3}$ and $B\in\mathbb{R}_{3\times 1}$ are given by,
\begin{footnotesize}
\begin{equation}\label{eq:matrixformA}
A =
\begin{bmatrix}
 1 & 0 & 0 \\
 \tan(\theta_{des}) & 1 & 0 \\
 0 & 0 & 1
\end{bmatrix}
A_{p}
\end{equation}
\begin{multline} \label{eq:matrixformAp}
\text{where, } A_{p} =  \\ 
\begin{bmatrix}
-\cos(\alpha_{p} - \psi)\cos(\gamma) & V_{p}\sin(\alpha_{p}-\psi)\cos(\gamma) & V_{p}\cos(\alpha_{p}-\psi)\sin(\gamma)\\
 & & \\
-\sin(\gamma) & 0 & -V_{p}\cos(\gamma) \\ 
 & & \\
-\sin(\alpha_{p}-\psi)\cos(\gamma) & - V_{p}\cos(\alpha_{p}-\psi)\cos(\gamma) & V_{p}\sin(\alpha_{p}-\psi)\sin(\gamma)
\end{bmatrix}
\\\hspace{4em}
\end{multline}
\begin{multline}  \label{eq:matrixformB}
\text{and} \ \  B = \\
\begin{bmatrix}
 - k_{1} \ (S_{1})^{n/m} + V_{p} \ \sin(\alpha_{p} - \psi) \ \cos(\gamma)\dot{\psi} - \\ \dot{V_{t}} \ \cos(\alpha_{t} - \psi) + V_{t} \ \sin(\alpha_{t} - \psi)(\dot{\alpha_{t}} - \dot{\psi}) - \\  k_a \ (V_{t} \ \cos(\alpha_{t} - \psi) - V_{p} \ \cos(\gamma) \ \cos(\alpha_{p} - \psi)) \\
  \\ 
 \tan(\theta_{des})(V_{p} \ \sin(\alpha_{p} - \psi) \ \cos(\gamma)\dot{\psi} - \\ \dot{V_{t}} \ \cos(\alpha_{t} - \psi) + V_{t} \ \sin(\alpha_{t} - \psi)(\dot{\alpha_{t}} - \dot{\psi}) - \\  k_b \ (V_{t} \ \cos(\alpha_{t} - \psi) - V_{p} \ \cos(\gamma) \ \cos(\alpha_{p} - \psi))) - \\ k_{2} \ (S_{2})^{n/m} + k_{b} \ V_{p} \ \sin(\gamma) \\ 
  \\ 
 -R_{xy} \ k_{3} \ (S_{3})^{n/m} - k_{c} \ R_{xy} \ (\dot{\psi} - \dot{\alpha_{t}}) + \\ \ddot{\alpha_{t}} \ R_{xy} + \dot{R}_{xy}\dot{\psi} - V_{p} \ \cos(\alpha_{p} - \psi) \ \cos(\gamma) \ \dot{\psi} - \\ V_{t} \ \cos(\alpha_{t} - \psi) \ (\dot{\alpha_{t}} - \dot{\psi}) - \dot{V_{t}} \ \sin(\alpha_{t} - \psi)
\end{bmatrix}
\end{multline}
\end{footnotesize}

The guidance command inputs are obtained as a solution to the system of equations, $AU = B$. Note that $\det(A)= (V_{p})^{2} \cos(\gamma)$. When $V_p \neq 0$ and $\cos(\gamma) \neq 0$, the solution to the system of equations $AU=B$ is given by, 
\begin{align}
\label{eq:input}
U &= A^{-1} B    
\end{align}
Here, the determinant of matrix $A$, as shown above, doesn't depend on the target's states (heading and position) directly as $\det(A)=0$ if and only if $\cos(\gamma)=0$, which happens when flight path angle is $\pm\pi/2$, or $V_p=0$. Thus, the proposed guidance scheme is free from target-state-dependent singularity. 
Note that although choosing the specific sliding variables and their dynamics as given in Eqs. \eqref{eq:Sdefn} and \eqref{eq:Sdotdefn}, respectively, could avoid target-state-dependent singularity, it might encounter UAV-state-dependent singularity in the following situations: Landing on a stationary ground target, Vertical takeoff, and Target moving toward the UAV.

Thus, Eq. \eqref{eq:input} leads to a non-singular guidance command in most practical scenarios except a few mentioned above. These situations could generally be handled by suitable manipulation of the inputs. With all these considerations, the overall guidance algorithm is formulated in Algorithm \ref{alg:inputsummary}.

\begin{algorithm}[!h]
\caption{Landing Guidance Command Summary}
\label{alg:inputsummary}
\begin{algorithmic}
\STATE $[\dot{V_p}' \quad \dot{\alpha_p}' \quad \dot{\gamma}']= A^{-1}B$

\IF {$V_{p} < M_{1}$ \textbf{and} $\dot{V_{p}}' < 0 $}
    \STATE $\dot{V_{p}}' = 0$
\ENDIF
\IF {$\cos(\gamma) < M_{2}$ \textbf{and} $\dot{\gamma}'\gamma >0$}
    \STATE $\dot{\gamma}' = 0$
\ENDIF
\STATE $\dot{V_{p}}' = \max([ \ N_1, \  \dot{V_{p}}' \ ])$

\STATE $\dot{\alpha_p}' = \max([ \ N_2, \  \dot{\alpha_p} \ ])$

\STATE $\dot{\gamma}' \  \ = \max([ \ N_{3}, \  \dot{\gamma}' \ ])$

\STATE $U = [\dot{V_{p}}' \quad \dot{\alpha_p}' \quad \dot{\gamma}']^T$
\end{algorithmic}
\end{algorithm}
\FloatBarrier
Values of pre-specified constants are selected based on simulations. The selection should be done such that $M_{1}$ and $M_{2}$ are sufficiently small. Here, the values of $N_{1}$, $N_{2}$ and $N_3$ represent upper bounds for the control inputs ($\dot{V_p},\dot{\alpha_t},\dot{\gamma}$) derived from the guidance scheme. 

\subsection{Selection of Parameters and Discussion on Designed Guidance Law}\label{subsec:ParameterSelection}

Discussions on the landing guidance algorithm in this section would be based on the premise that the dynamics of the sliding variables are achievable.
Consider the criteria for the tuning of the guidance parameters, namely, $k_{1}$, $k_{2}$, $k_{3}$, $k_{a}$, $k_{b}$, $m$ and $n$. Expanding Eq. \eqref{eq:input}, we obtain the guidance commands as below.
\begin{align} \label{eq:expand1}
&\dot{V_{p}} = \nonumber \\
& \quad k_{1} \ (S_{1})^{n/m} \ \cos(\gamma) \ \cos(\alpha_{p}-\psi) + k_{2} \ (S_{2})^{n/m} \sin(\gamma)\nonumber \\
&+ k_{3}R_{xy} \ \cos(\gamma)sin(\alpha_{p}  -\psi)(S_{3})^{n/m} - k_{b} \ V_{p} \ (\sin(\gamma))^{2}\nonumber \\
&+ k_{c} \ R_{xy} \ (\dot{\psi} - \dot{\alpha_{t}}) \ \cos(\gamma) \ \sin(\alpha_{p} - \psi) \nonumber \\ 
&+ \cos(\gamma) (\ k_{a} \ \dot{R}_{xy} \ \cos(\alpha_{p} - \psi) - \ddot{\alpha_{t}} \ R_{xy} \ \sin(\alpha_{p} - \psi)) \nonumber \\
&-\dot{R_{xy}} \ \dot{\psi} \ \cos(\gamma) \ \sin(\alpha_{p}-\psi) + \dot{V_{t}} \ \cos(\gamma) \ \cos(\alpha_{p} - \alpha_{t}) \nonumber \\
&+ V_{t} \ \cos(\gamma) \ (\dot{\alpha_{t}} - \dot{\psi}) \ \sin(\alpha_{p} - \alpha_{t}) \nonumber \\
&- \sin(\gamma)(k_b - k_a)(\tan(\theta_{des})\dot{R}_{xy} + \tan(\theta_{des})k_1 (S_1)^{n/m})
\end{align}
\begin{align} \label{eq:expand2}
    \dot{\alpha_{p}} &= \nonumber \\
    &\frac{1}{V_{p} \ \cos(\gamma)} \ [ -k_{1} \ (S_{1})^{n/m} \ \sin(\alpha_{p} - \psi) + V_{p} \ \cos(\gamma) \ (\dot{\psi})\nonumber \\
    &+ k_{3} \ (S_{3})^{n/m} \ R_{xy} \ \cos(\alpha_{p} - \psi) - \ddot{\alpha_{t}} \ R_{xy} \ \cos(\alpha_{p}-\psi) \nonumber \\ 
    &+ k_{c} \ R_{xy} \ (\dot{\psi} - \dot{\alpha_{t}}) \ \cos(\alpha_{p}-\psi) - \dot{R}_{xy} \ \dot{\psi} \ \cos(\alpha_{p}-\psi) \nonumber \\
    &+ k_{a} \ V_{p} \ \cos(\gamma) \ \cos(\alpha_{t}-\psi) \ \sin(\alpha_{p}-\psi) \nonumber \\ 
    &- k_{a} \ V_{t} \ \cos(\alpha_{t}-\psi) \ \sin(\alpha_{p}-\psi) \nonumber \\ 
    &+ \dot{V_{t}} \ \sin(\alpha_{t} - \alpha_{p}) + V_{t}\cos(\alpha_{t}-\alpha_{p}) \ (\dot{\alpha_{t}}-\dot{\psi}) ] 
\end{align}

\begin{align} \label{eq:expand3}
    \dot{\gamma} &=  \nonumber \\
    \quad &\frac{1}{V_p}[-k_1 (S_1)^{n/m} (\cos(\alpha_p - \psi)\sin(\gamma) + \tan(\theta_{des}) \cos(\gamma)) \nonumber \\ 
    &- R_{xy} k_3 (S_3)^{n/m} \sin(\alpha_p - \psi) \sin(\gamma) \nonumber \\
    &-\dot{V}_t \sin(\gamma) \cos(\alpha_p-\alpha_t) 
    +k_2 (S_2)^{n/m}\cos(\gamma) \nonumber \\
    &+\dot{R}_{xy}(-k_a - k_c(\dot{\psi} - \dot{\alpha_t}) \sin(\alpha_p-\psi)\sin(\gamma) \nonumber \\
    &-(k_b-k_a) \ \tan(\theta_{des}) \  \cos(\gamma) + \dot{\psi} \sin(\alpha_p-\psi)\sin(\gamma)) \nonumber \\
    &- k_b V_p \sin(\gamma)\cos(\gamma) - V_t \sin(\gamma) (\dot{\alpha_t} - \dot{\psi}) \sin(\alpha_p - \alpha_t)] 
\end{align}
\FloatBarrier

\begin{figure*}[b!]
\begin{center}
\includegraphics[width= 0.6\textwidth]{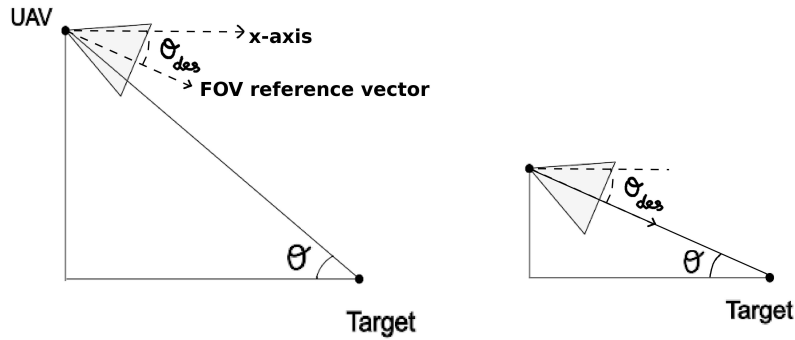}    
\vspace{-1\baselineskip}
\caption{FOV-constrained landing} 
\label{fig:geometry2}
\end{center}
\end{figure*}

\subsubsection{Selection Criteria for parameters $k_1$, $k_2$, and $k_3$}
As can be seen from Eqs. \eqref{eq:expand1} - \eqref{eq:expand3}, the expressions of guidance command inputs $\dot{V_{p}}$, $\dot{\alpha_{p}}$ and $\dot{\gamma}$ are not easily tractable. However, considering the dynamics of the sliding variables and associated guidance objectives, we arrive at the following criteria.

From Eqs. \eqref{eq:Rxydot} and \eqref{eq:psidot}, to avoid shooting up of $\dot{\psi}$, it is desired that $S_1$ converges to zero before $R{xy}$ approaches 0. Now, from Eq. \eqref{eq:straj} the time for sliding variable $S_1$ to reach the sliding surface $S_1=0$, denoted as $t_{reach_1}$, can be obtained as,
\begin{align}
    t_{reach_1} &= \frac{m}{(m-n)} \ \frac{(S_{1}(0))^{((m-n)/m)}}{k_{1}}
    \label{eq:tsliding} \\
\end{align}

\begin{proposition}
Given any initial condition of the target and the UAV, if the chosen dynamics considered in Eq. \eqref{eq:Sdotdefn} are enforced on the sliding variables defined Eq. \ref{eq:Sdefn}, an upper bound for $|\dot{R}_{xy}|$ can be derived as follows:
\begin{equation}\label{eq:upperbound}
    |\dot{R}_{xy}(t)| \leq (|\dot{R}_{xy0}| +  k_{a}R_{xy0})
\end{equation}
\end{proposition}

\begin{proof}
Differentiating both sides of first row of Eq. \eqref{eq:Sdefn},
\begin{align}
\label{eq:Rxyddot}
    \ddot{R}_{xy} &= -k_{a}\dot{R}_{xy} - k_{1}(S_{1})^{n/m}
\end{align}

It should be noted that once the system is on the sliding surface, both $|\dot{R}_{xy}(t)|$ and $R_{xy}(t)$ decrease until they become zero. Now, the subsequent discussion is split into two cases. 

Case 1:
First, Consider the case in which the guidance is initiated with $\dot{R}_{xy0} < 0$. This can be further divided into 2 sub-cases: $|\dot{R}_{xy0}| > k_{a}R_{xy0}$ and $|\dot{R}_{xy0}| < k_{a}R_{xy0}$.

In the first sub-case, $S_{10} < 0$. Then, from Eq. \eqref{eq:Sdotdefn}, $S_{1}(t)<0$ and $\dot{S_1}(t)>0$ throughout the reaching phase, in which the system is reaching the sliding surface. Since the sliding variable $S_1$ evolves smoothly by Eq. \eqref{eq:Sdotdefn}, it can be inferred from Eq. \eqref{eq:Rxyddot}, $\ddot{R}_{xy}>0$ for all values of $\dot{R}_{xy}<-k_a{R}_{xy}$, that is for the entire reaching phase. Thus,
\begin{equation}
    |\dot{R}_{xy}(t)| \leq |\dot{R}_{xy0}| \leq (|\dot{R}_{xy0}| +  k_{a}R_{xy0})
\end{equation}

In the second sub-case, $S_{10} > 0$. Therefore, following Eq. \eqref{eq:Sdotdefn}, $\dot{S_{1}(t)} < 0$ and $S_{1}(t) > 0$ throughout the reaching phase. Now, let at some time-instant $t_1$ during the reaching phase, $\dot{R}_{xy}(t_1)$ = 0 implying from Eq. \eqref{eq:Rxyddot} that $\ddot{R}_{xy}(t_1) \leq 0$. However, since $\dot{R}_{xy0} < 0$, $\dot{R}_{xy}(t)$ cannot reach zero without attaining positive $\ddot{R}_{xy}(t)$. Thus, it leads to a contradiction, which implies that for $\dot{R}_{xy0} < 0$, $\dot{R}_{xy}$ is always negative and doesn't reach zero throughout the reaching phase. Hence, $k_{a} R_{xy}(t) < k_{a} R_{xy0}$. Since $S_{1}(t) > 0$, $|\dot{R}_{xy}(t)| < k_a R_{xy}(t) < k_{a} R_{xy0}$.  Thus,
\begin{equation}
    |\dot{R}_{xy}(t)| \leq |k_{a} R_{xy0}| \leq (|\dot{R}_{xy0}| + k_{a}R_{xy0})
\end{equation}

Case 2:
Now, consider the case where $\dot{R}_{xy0} \geq 0$. Here, $S_{10} > 0$. Following Eq. \eqref{eq:Sdotdefn}, $\dot{S_{1}(t)} < 0$ and $S_{1}(t) > 0$ throughout the reaching phase. From Eq. \eqref{eq:Rxyddot}, it can be seen that $\ddot{R}_{xy}<0$, until $k_a\dot{R}_{xy}(t) \leq -k_1 (S_{1}(t))^{n/m}$. On the sliding surface $S_1=0$, note that $\dot{R}_{xy} = - k_a R_{xy}<0$. Since from Eq. \eqref{eq:Rxyddot}, $\dot{R}_{xy}$ is continuous in time $t$, at some time $=t_1$ (say) $\dot{R}_{xy}$ would cross zero, that is $\dot{R}_{xy}(t_{1})=0$.
Since $\dot{S_{1}(t)} < 0$ in the reaching phase, $S_1(t_1)<S_{10}$, that is,
\begin{equation}
k_a R_{xy}(t_{1})<\dot{R}_{xy0} + k_a R_{xy0}
\label{eq:t_1}
\end{equation}

Now, at any instant $t=t_{1} + h$, where $h>0$, $\dot{R}_{xy}(t_1 + h)<0$, which implies $k_a R_{xy}(t_1 + h) < k_a R_{xy}(t_1)$. For all $h>0$, this falls under Case 1b described above. Following similar logic,
$$|\dot{R}_{xy}(t)| \leq |\dot{R}_{xy}(t_1+h)| + k_{a}R_{xy}(t_1+h), \ \ \forall \ \ t> t_1 + h$$ 

Since $\dot{R}_{xy}$ is continuous in time (from \eqref{eq:Rxydot}), and as $\dot{R}_{xy}(t_1)=0$, at $t=t_{1}$ as $h \to 0$, we obtain, 
\begin{equation}
    \label{eq:t_1_2}
    |\dot{R}_{xy}(t)| < k_{a}R_{xy}(t_1) \ \ \forall \ \ t>t_1
\end{equation}

Finally, combining Eqs. \eqref{eq:t_1} and \eqref{eq:t_1_2}, we obtain,
\begin{equation}\label{eq:RxydotInequality}
|\dot{R}_{xy}(t)| \leq (|\dot{R}_{xy0}| + k_a R_{xy0})
\end{equation}

Thus, $|\dot{R}_{xy}|$ is bounded as given in Eq. \eqref{eq:upperbound} for any initial condition of engagement between the target and the UAV.

\end{proof}

The upper bound for $|\dot{R}_{xy}|$ as given in Eq. \eqref{eq:upperbound} is now expressed as a lower bound on the time taken for the horizontal range between the UAV and the ground target to reach zero, denoted by $t_{R_1}$. This is given by,
\begin{equation}
\label{eq:tR1}
    t_{R_{1}} \geq \frac{R_{xy0}}{( \ |\dot{R}_{xy0}| \ + \ K_{a} \ R_{xy0})}
\end{equation}

With the consideration of $t_{R_{1}}$ being greater than $t_{reach_1}$,  a sufficient condition for selection of $k_1$ is derived as below:
\begin{equation}
k_{1} \geq \ \frac{m}{(m-n)} \frac{(S_{10})^{((m-n)/m)}}{R_{xy0}} ( |\dot{R}_{xy0}| +  k_{a}  R_{xy0}  )
\end{equation}

Next, consider selection criteria of $k_{2}$ and $k_{3}$. For the sake of simplicity, it is desired that all sliding variables ($S_1,S_2$ and $S_3$) converge at the same time. Thus, the parameters are chosen as follows:
\begin{equation}
    \begin{split}
    \frac{k_{2}}{k_{1}} = (\frac{S_{20}}{S_{10}})^{((m-n)/m)}\: \: , \: \:
    \frac{k_{3}}{k_{1}} = (\frac{S_{30}}{S_{10}})^{((m-n)/m)}    
    \end{split}
\end{equation}

\subsubsection{Selection Criteria for parameters $m$ and $n$}

Consider the selection criteria of $m$ and $n$. As proposed in Section \ref{subsec:synthesis}, they are chosen to be odd and co-prime integers such that $0<n<m$, for ensuring finite time convergence of sliding variables. When $n/m$ is selected close to 0, guidance command would shoot up in close vicinity of the sliding surface, and this leads to chattering in $S$. However, as $n/m \to 1$, $t_{reach1} \to \infty$, which means that the sliding variables don't converge in finite time. Thus, there exists a trade-off between the time taken for the convergence of sliding variables and the magnitude of chattering in the sliding variables. Consequently, the parameters $m$ and $n$ are tuned based on simulations, such that $n/m$ is close to 1 and the time taken for the convergence for sliding variables is satisfactory.

\begin{figure*}[b!]
\centering
\subcaptionbox{Stationary Target\label{fig:stationarytraj}}{\includegraphics[width=0.2\textwidth,height=9.5cm,keepaspectratio]{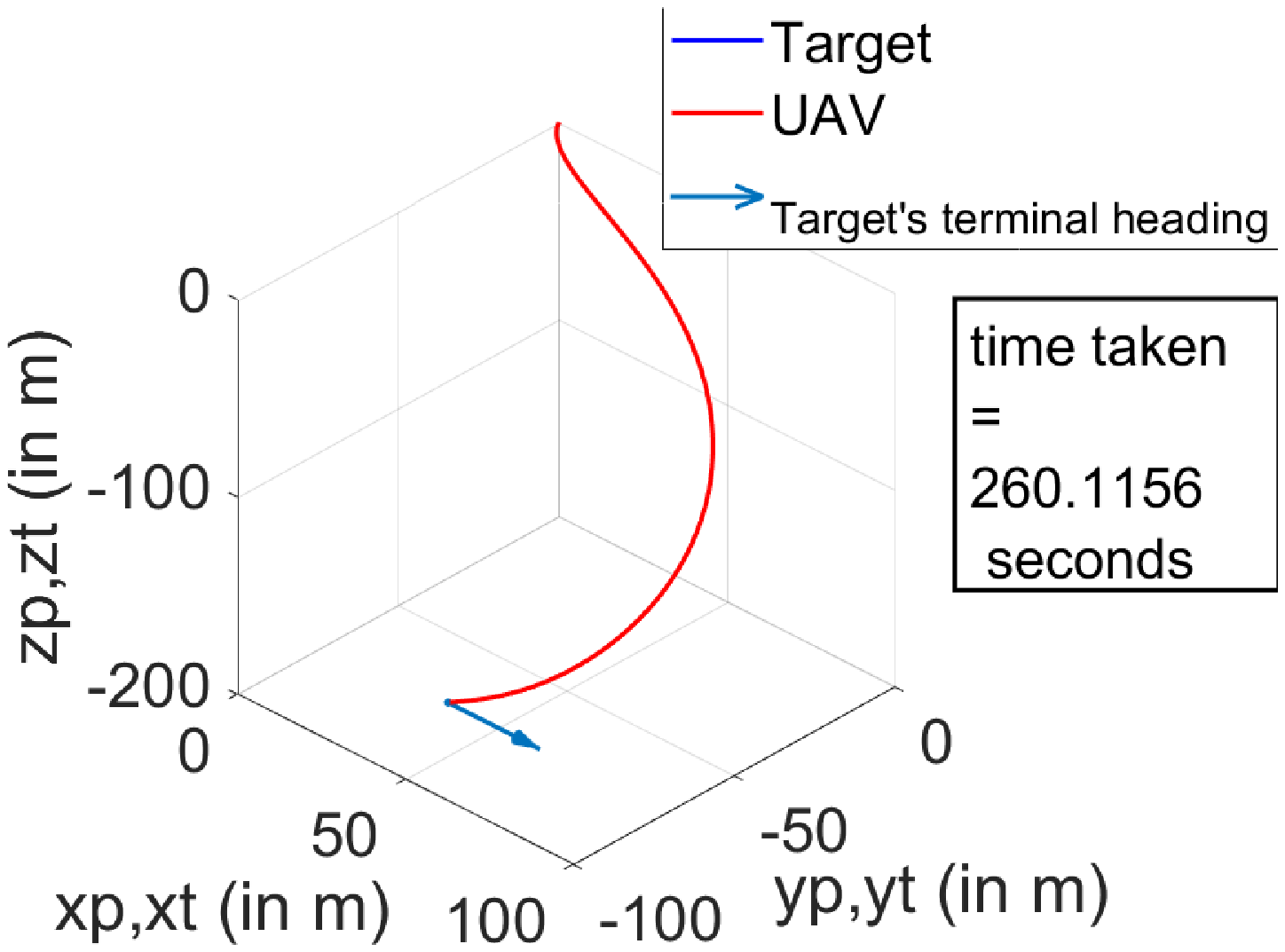}}%
\hfill 
\subcaptionbox{Non-Maneuvering Target\label{fig:slinetraj}}{\includegraphics[width=0.2\textwidth,height=9.5cm,keepaspectratio]{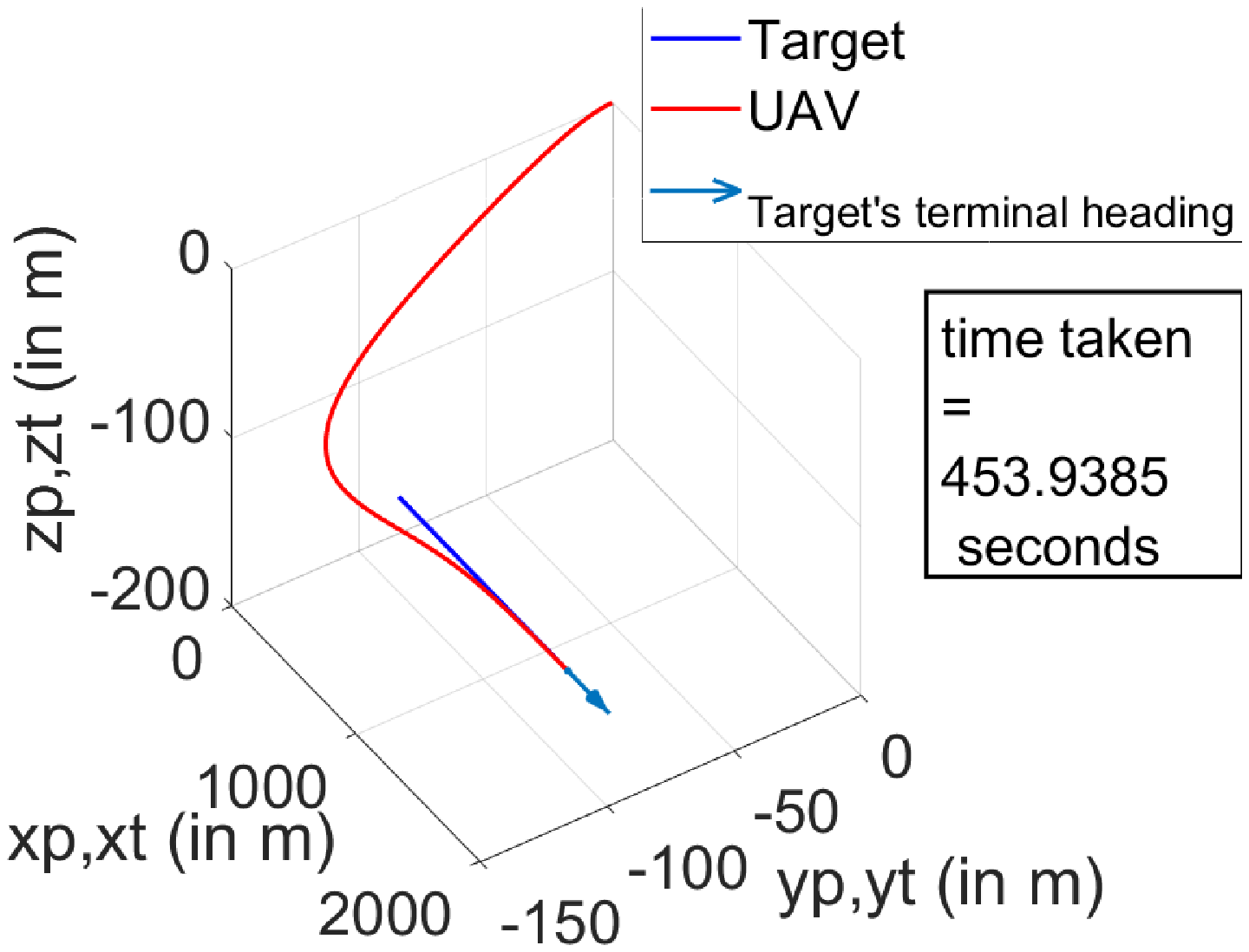}}%
\hfill 
\subcaptionbox{Constant Maneuvering Target\label{fig:circulartraj}}{\includegraphics[width=0.2\textwidth,height=9.5cm,keepaspectratio]{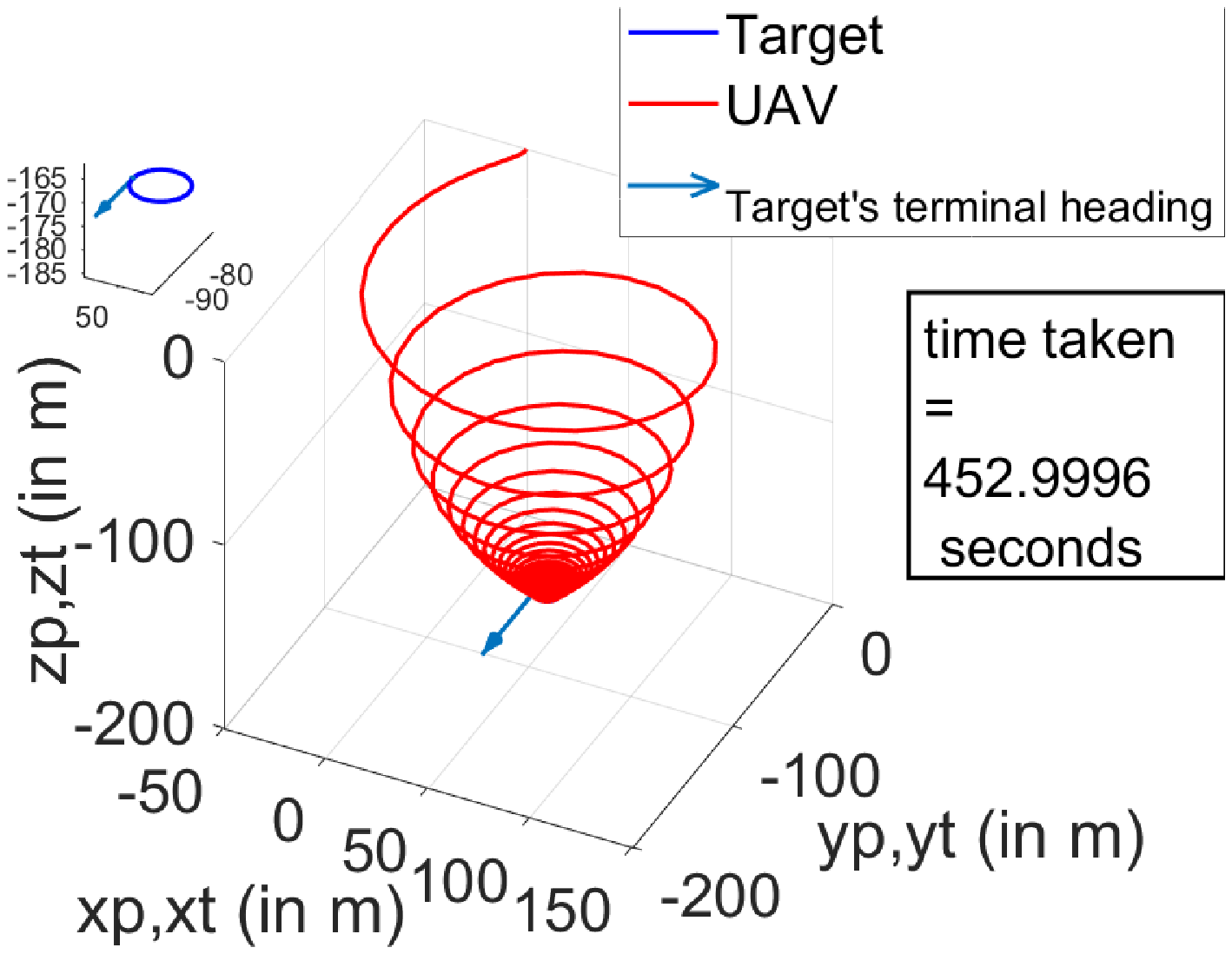}}%
\hfill 
\subcaptionbox{Sinusoidally Maneuvering Target\label{fig:sinusoidaltraj}}{\includegraphics[width=0.2\textwidth,height=9.5cm,keepaspectratio]{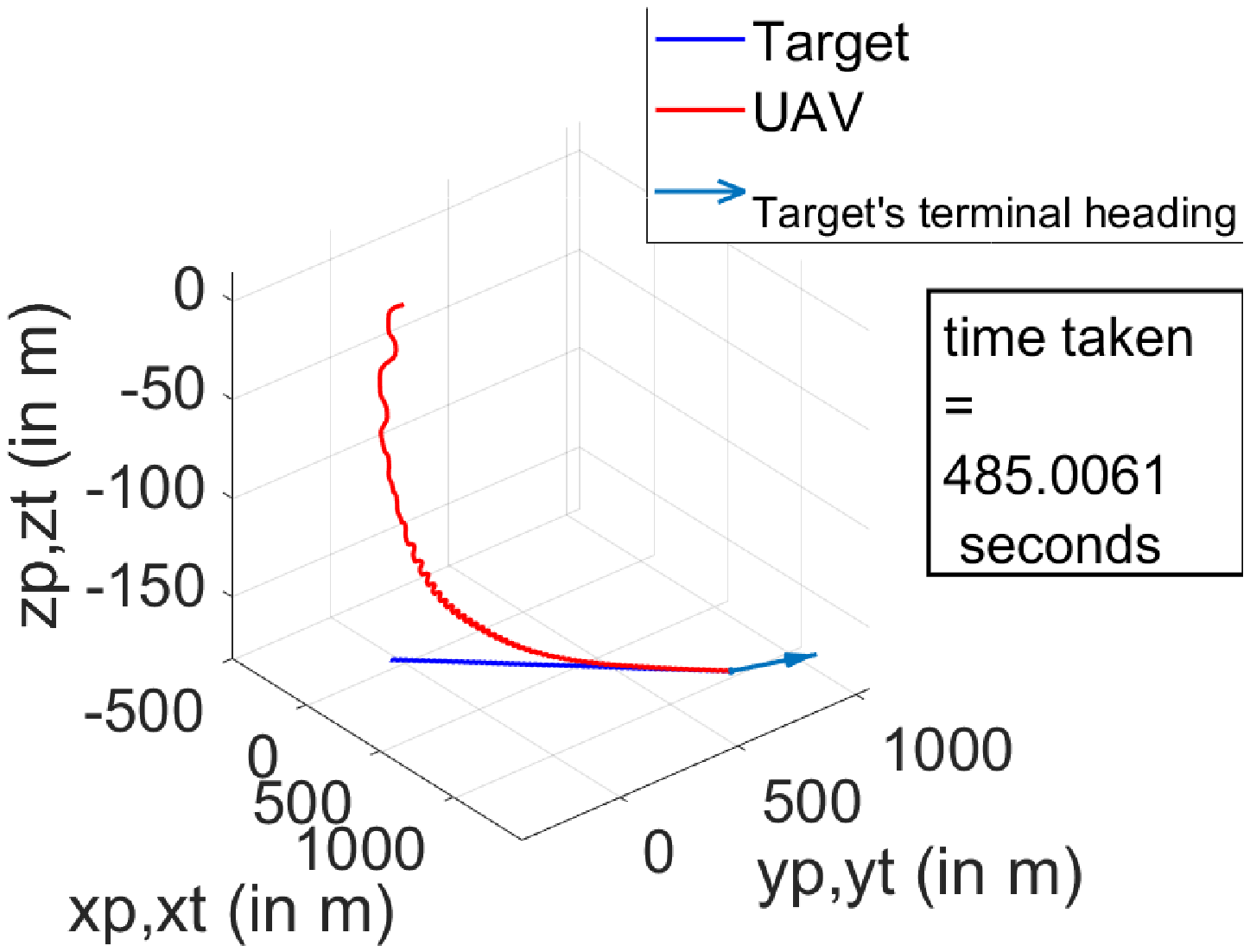}}%
\caption{Trajectory plots for UAV and target}
\label{fig:trajs}
\end{figure*}

\begin{figure*}[b!]
\centering
\subcaptionbox{Stationary Target\label{fig:stationarydist}}{\includegraphics[width=.2\textwidth,height=9.5cm,keepaspectratio,trim={0.5cm 0.0cm 0.0cm .08cm}]{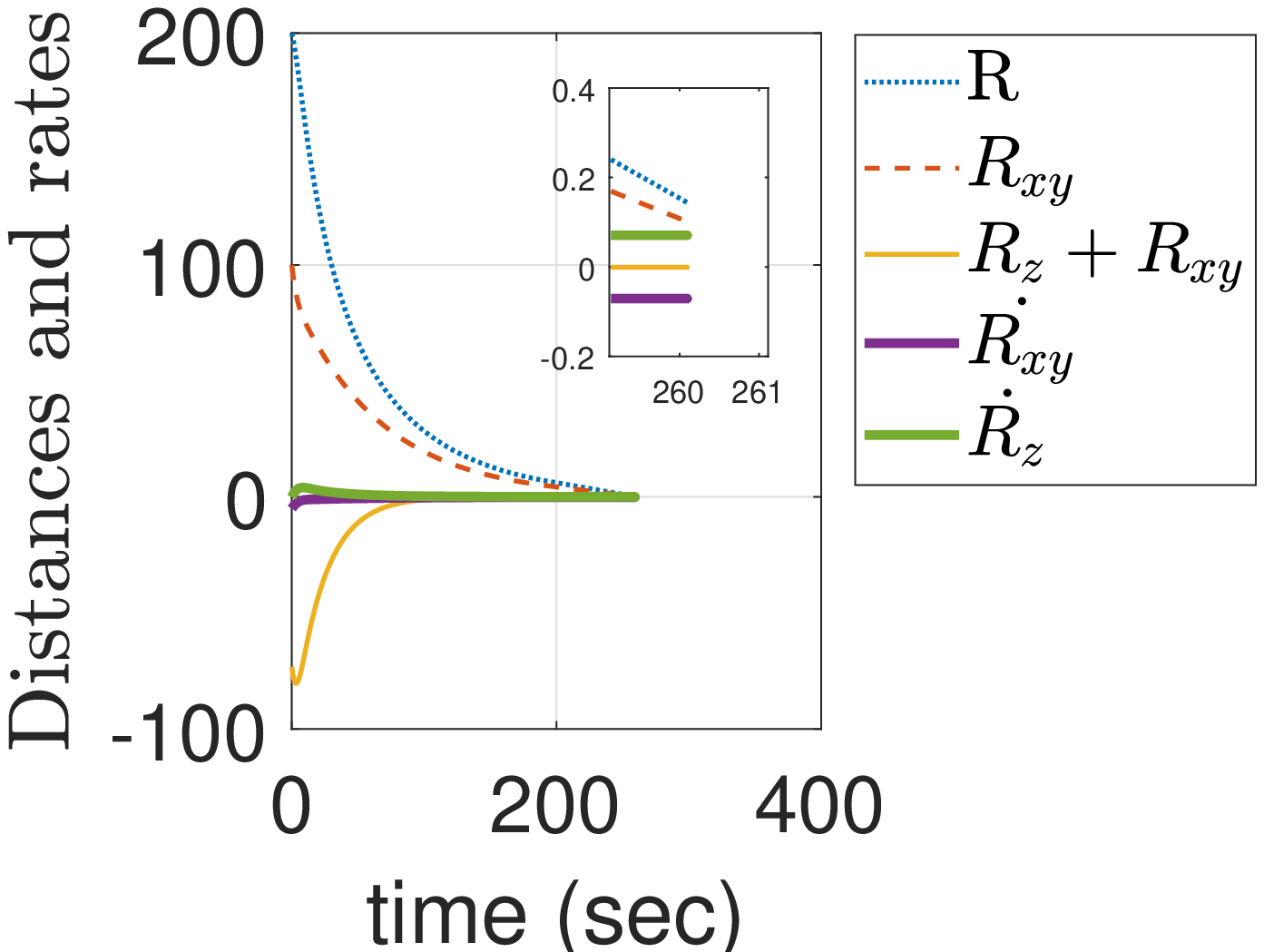}}%
\hfill 
\subcaptionbox{Non-Maneuvering Target\label{fig:slinedist}}{\includegraphics[width=.2\textwidth,height=9.5cm,keepaspectratio,trim={0.5cm 0.0cm 0.0cm .08cm}]{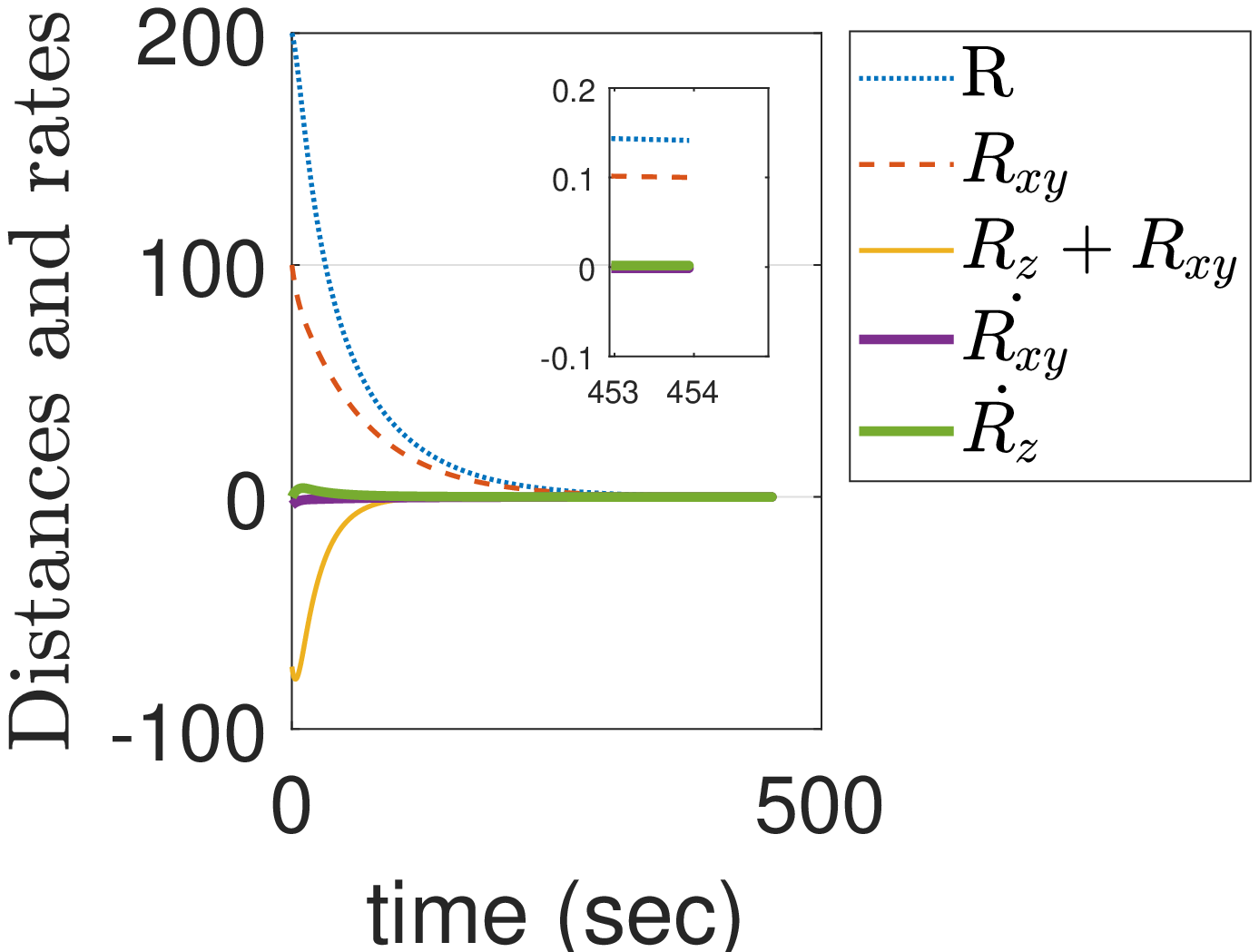}}%
\hfill 
\subcaptionbox{Constant Maneuvering Target\label{fig:circulardist}}{\includegraphics[width=.2\textwidth,height=9.5cm,keepaspectratio,trim={0.5cm 0.0cm 0.0cm .08cm}]{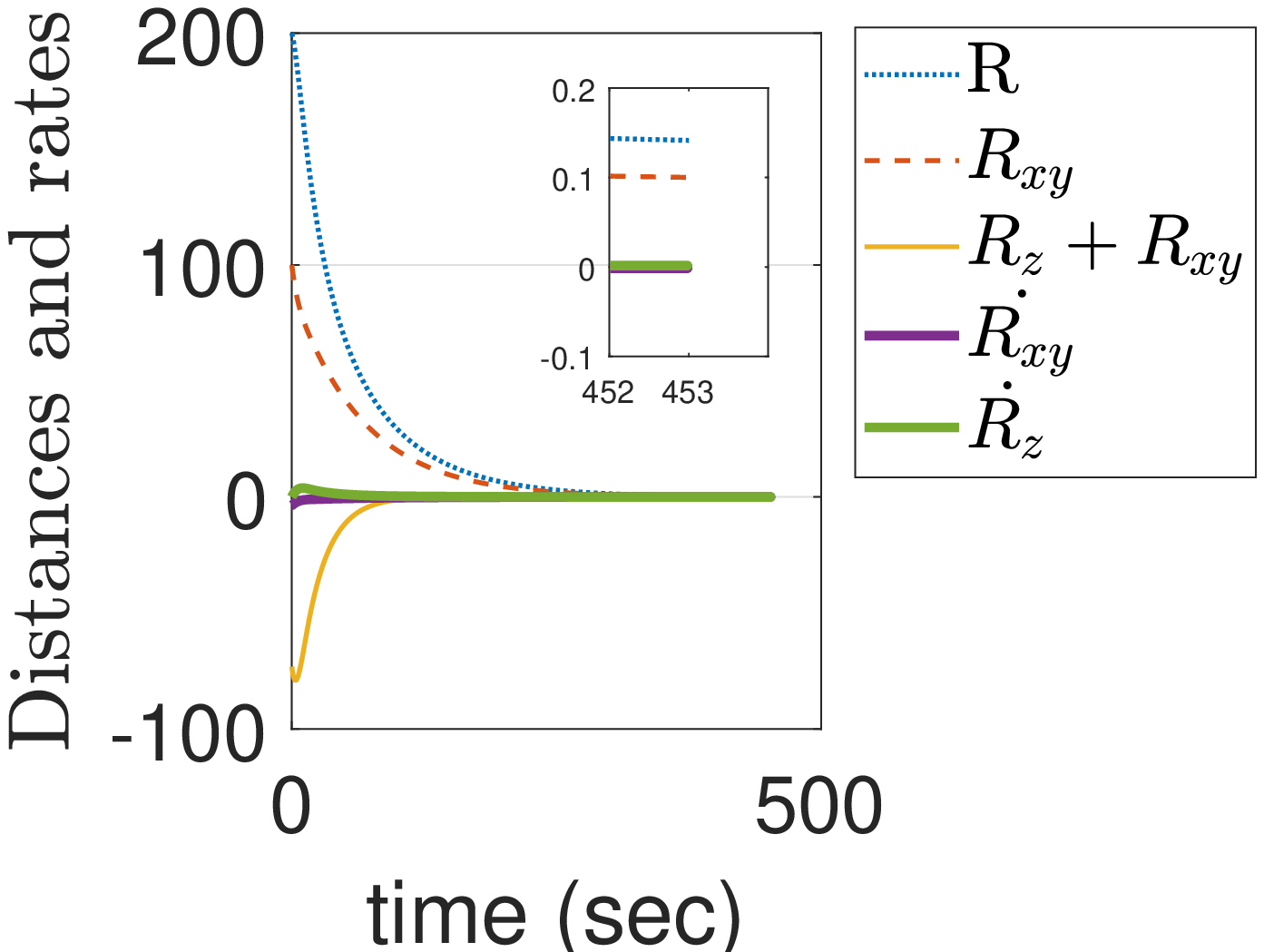}}%
\hfill 
\subcaptionbox{Sinusoidally Maneuvering Target\label{fig:sinusoidaldist}}{\includegraphics[width=.2\textwidth,height=9.5cm,keepaspectratio,trim={0.5cm 0.0cm 0.0cm .08cm}]{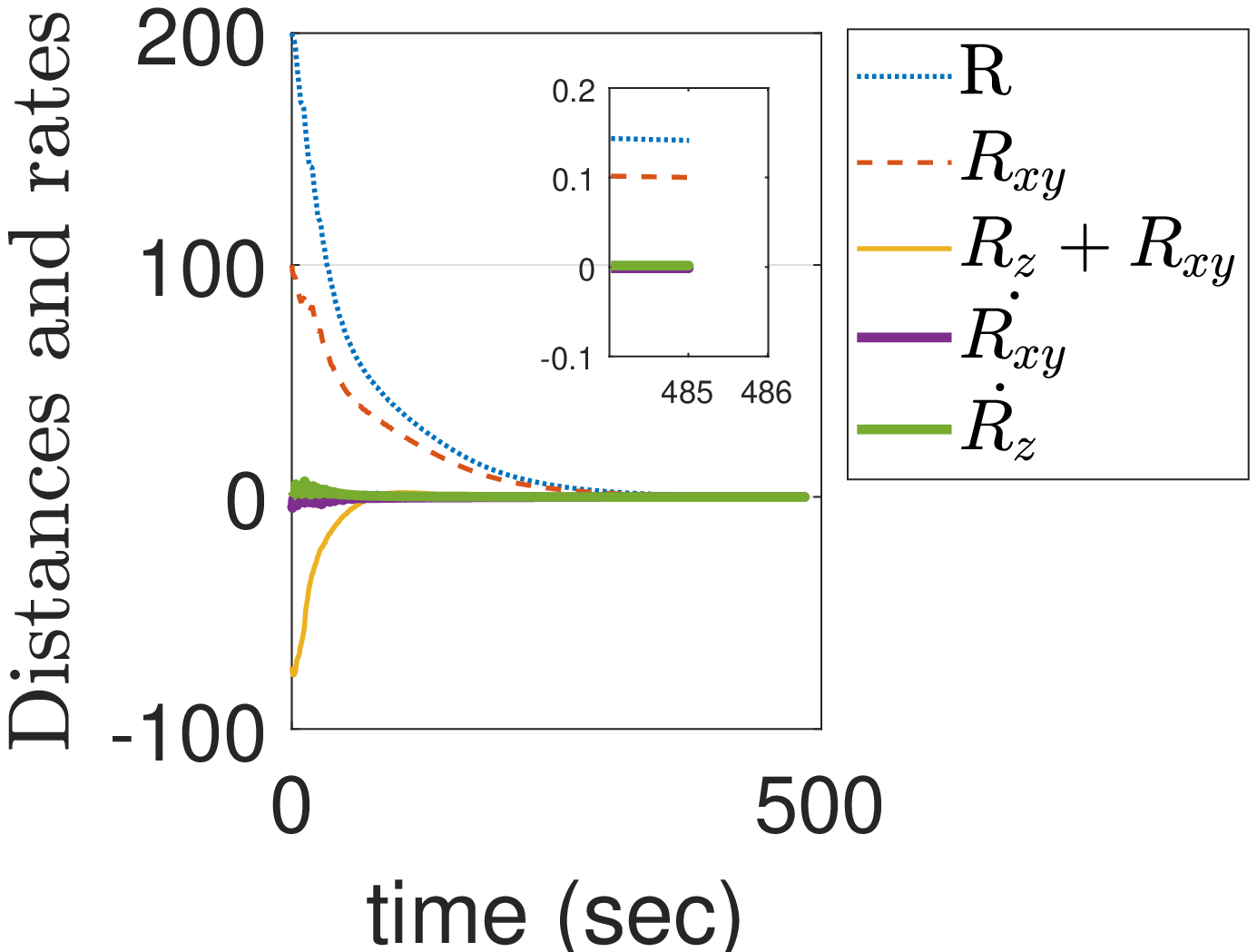}}%
\caption{Distance from target and it's projections on the xy-plane and along the z-axis}
\label{fig:dists}
\end{figure*}

\subsubsection{Selection Criteria for parameters $k_a$, $k_b$, and $k_c$}
Consider the parameter $k_a$, which represents the exponential decay constant for $R_{xy}$, when the system is in sliding mode. From Eq. \eqref{eq:slidingmodedynamics} for $R_{xy}$, a higher value for $k_a$ would result in higher desired speeds for the UAV. Thus, this parameter could be tuned based on simulations conducted for different cases. 

Finally, consider the selection criteria of $k_{b}$ and $k_{c}$. Since the desired azimuth and elevation angles are to be achieved before the completion of landing, it is desired that $\psi$ approaches $\alpha_t+\zeta_{des}$ and $\theta$ approaches $\theta_{des}$ faster than $R_{xy}$ and $R_z$ approach zero. This leads to 
the following selection criterion :
\begin{align}
    k_{b} > k_{a} ,  k_{c} > k_{a}
\end{align}
Here, $k_a$, $k_b$ and $k_c$ are to be chosen judiciously such that the guidance command inputs are within acceptable limits.

\subsubsection{FOV-constrained landing}
In Section \ref{subsec:synthesis}, it was proved in Theorem \ref{th:stability} that given finite time convergence of sliding variables and $k_b > k_a$, $\theta$ approaches a desired value asymptotically. This could be helpful in Field-of-View (FOV)-constrained landing, where the target is to be maintained in the UAV's FOV. In such a scenario, $\theta_{des}$ is set as the angle between the orthogonal projection of the centre-reference vector of the FOV on the xz-plane and the reference x-axis as shown in Fig. \ref{fig:geometry2}.

\section{Simulation Results}\label{sec:results}
MATLAB simulations are presented in this section assuming point mass models of the UAV and the UGV. As per the developed guidance algorithm (refer to Algorithm \ref{alg:inputsummary}), three guidance command inputs ($\dot{V_p}, \dot{\alpha_p}, \dot{\gamma}$) are generated. 
Simulation results are presented for four different kinds of targets - stationary, moving but non-maneuvering (straight-line trajectory), constant maneuvering (circular trajectory) and sinusoidally maneuvering trajectory of the ground target. In all the cases considered, the initial speed of the UAV is taken as 5 m/s. The UAV starts, at a distance of 200 m from the target, with $R_{xy0}=100 $ m and $R_{z0}=100 \sqrt{3} $ m, such that $\theta_0 = \pi/3$. The initial heading angle ($\alpha_{p0})$ and the flight path angle ($\gamma_0$) of the UAV are taken as $-\pi/3$ and $0$ rad, respectively. The initial LOS angle projected to the (xy)-plane ($\psi_{0}$) is taken as $-\pi/3$ rad. In the cases of moving target, the target speed is fixed at 3 m/s(=$V_t$). The parameters $m,n$ in Eq. \eqref{eq:Sdotdefn} are fixed as $m=5, n=3$. The parameters $N_1 , N_2$ and $N_3$, which dictate the maximum magnitudes for $\dot{V_p}$, $\dot{\alpha_p}$ and $\dot{\gamma}$ as mentioned in Algorithm \ref{alg:inputsummary}, are set as 10 $m/s^2$, $\pi/2$ rad/s and $\pi/2$ rad/s, respectively. Also, the parameters $M_1$ and $M_2$ as mentioned in Algorithm \ref{alg:inputsummary} are set as 0.1 m/s and 0.15, respectively. 
Rest of the guidance parameters used in the simulations along with the heading rate of the target and the desired approach angles in all cases are presented in Table \ref{tab:1initials}. The corresponding simulation results are shown in Figs. \ref{fig:trajs}-\ref{fig:slideswithtime}.

\begin{tiny}
\begin{center}
\begin{table}[h!]
\centering
\setlength\tabcolsep{4.5pt} %
\scalebox{0.8}[0.8]{
\begin{tabular}{ |p{0.07\textwidth}|p{0.04\textwidth}|p{0.045\textwidth}|p{0.025\textwidth}|p{0.05\textwidth}|p{0.035\textwidth}|p{0.035\textwidth}|p{0.035\textwidth}|p{0.035\textwidth}|p{0.035\textwidth}|}
\hline
Target type & $\dot{\alpha_{t}}(t)$ (rad/s) & $\zeta_{des} =$ \scalebox{.7}[1.0]{($\psi_{des} - \alpha_t)$} (rad) & $\theta_{des}$ (rad) & $k_a$ $=$ $1.5/R_{xy0}$ & $k_b$ $=$ $3k_a$ & $k_c$ $=$ $2k_a$ & $k_1$ & $k_2$ & $k_3$\\  
\hline
Stationary & 0 & $\pi$ & $\pi/4$ & 0.0150 & 0.0450 & 0.0300 & 0.1395 & 0.1784 & 0.0442\\
\hline
Non-maneuvering & 0 & $\pi/2$ & $\pi/4$ & 0.0150 & 0.0450 & 0.0300 & 0.0914 & 0.1297 & 0.0323\\ 
\hline
Constant-maneuvering & $\pi/6$ & $\pi/2$ & $\pi/4$ & 0.0150 & 0.0450 & 0.0300 & 0.0914 & 0.1297 & 0.0641\\
\hline
Sinusoidally maneuvering & $(\pi/6)$ $\sin(\frac{\pi t}{4} )$ & 0 & $\pi/4$ & 0.0150 & 0.0450 & 0.0300 & 0.0914 & 0.1297 & 0.0169\\
\hline
\end{tabular}}
\caption{Simulation cases and Guidance Parameters}
\label{tab:1initials}
\end{table}
\end{center}
\end{tiny}

\begin{figure*}[h!]
\centering
\subcaptionbox{Stationary Target\label{fig:stationaryinput}}{\includegraphics[width=.2\textwidth,height=9.5cm,keepaspectratio,trim={1cm 0.3cm 0.8cm .08cm}]{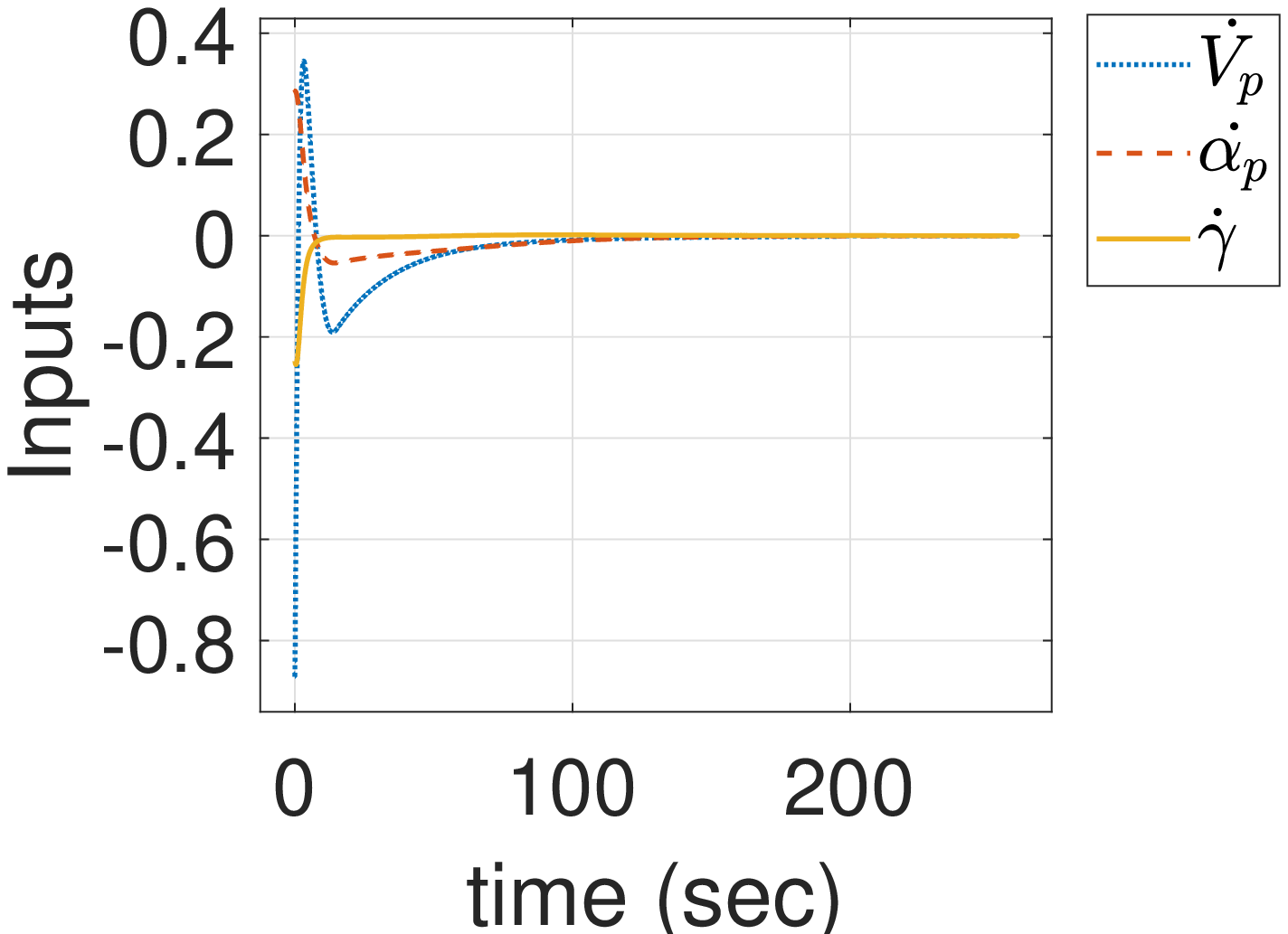}}%
\hfill 
\subcaptionbox{Non-Maneuvering Target\label{fig:slineinput}}{\includegraphics[width=.2\textwidth,height=9.5cm,keepaspectratio,trim={1cm 0.3cm 0.8cm .08cm}]{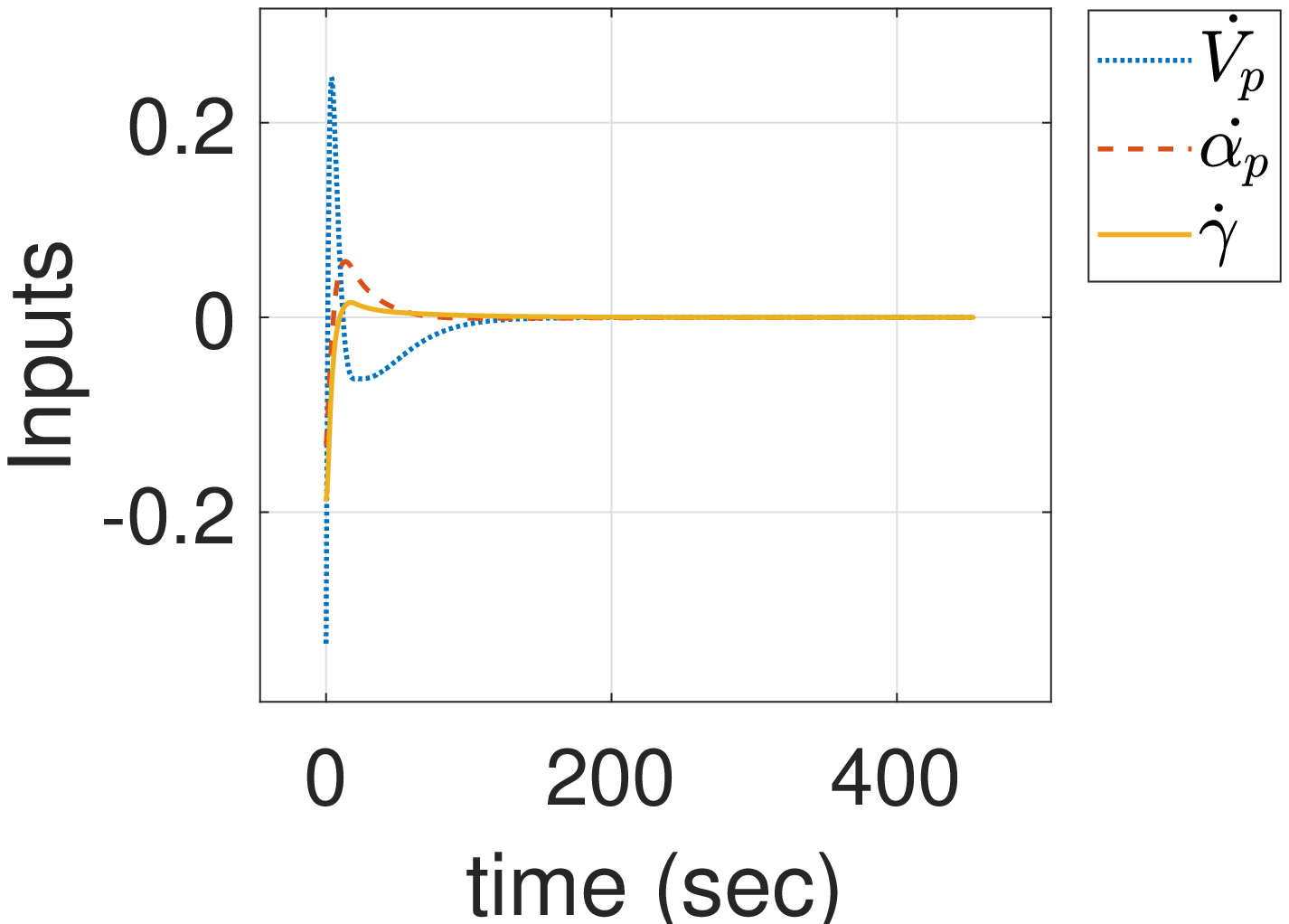}}%
\hfill 
\subcaptionbox{Constant Maneuvering Target\label{fig:circularinput}}{\includegraphics[width=.2\textwidth,height=9.5cm,keepaspectratio,trim={1cm 0.3cm 0.8cm .08cm}]{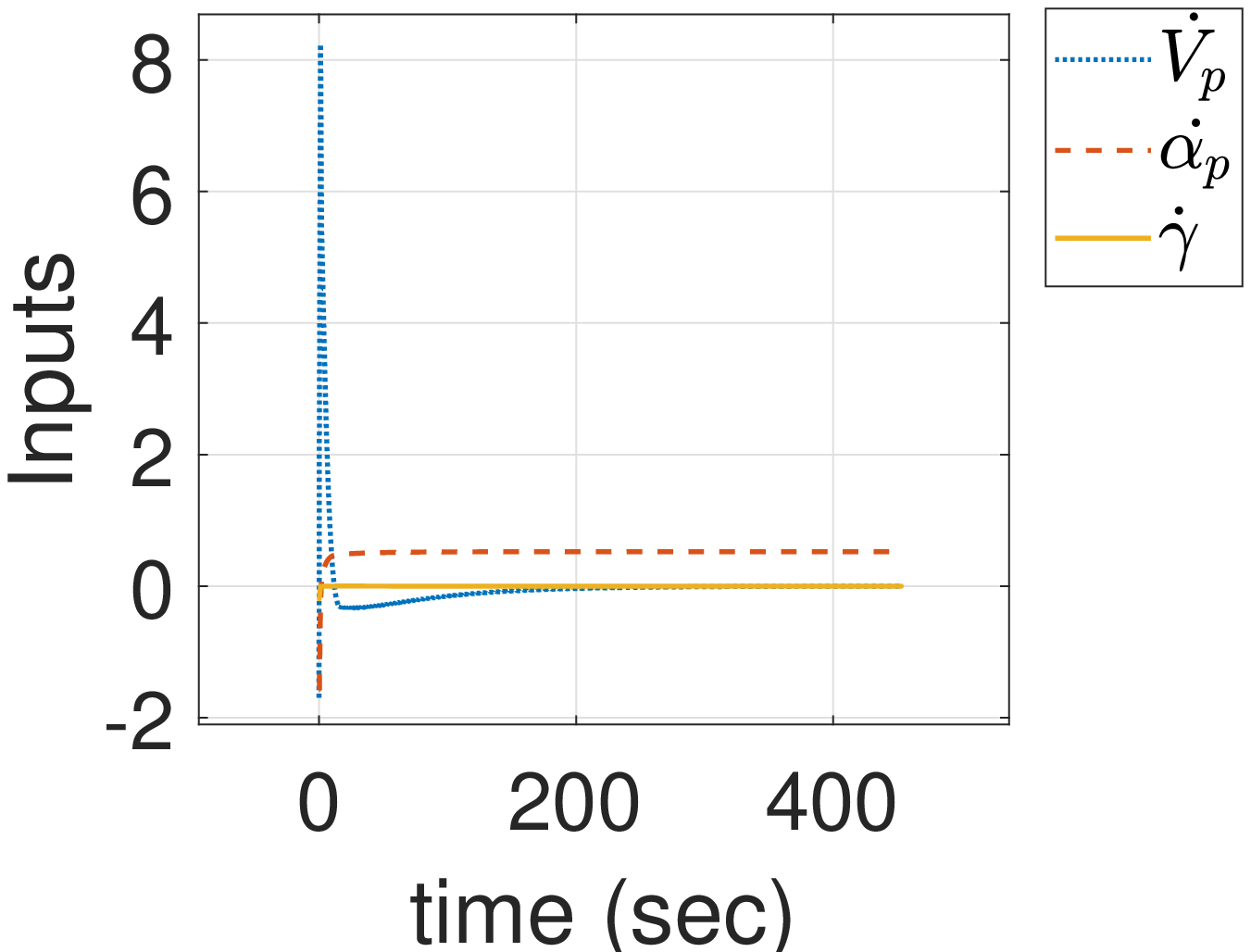}}%
\hfill 
\subcaptionbox{Sinusoidally Maneuvering Target\label{fig:sinusoidalinput}}{\includegraphics[width=.2\textwidth,height=9.5cm,keepaspectratio,trim={1cm 0.3cm 0.8cm .08cm}]{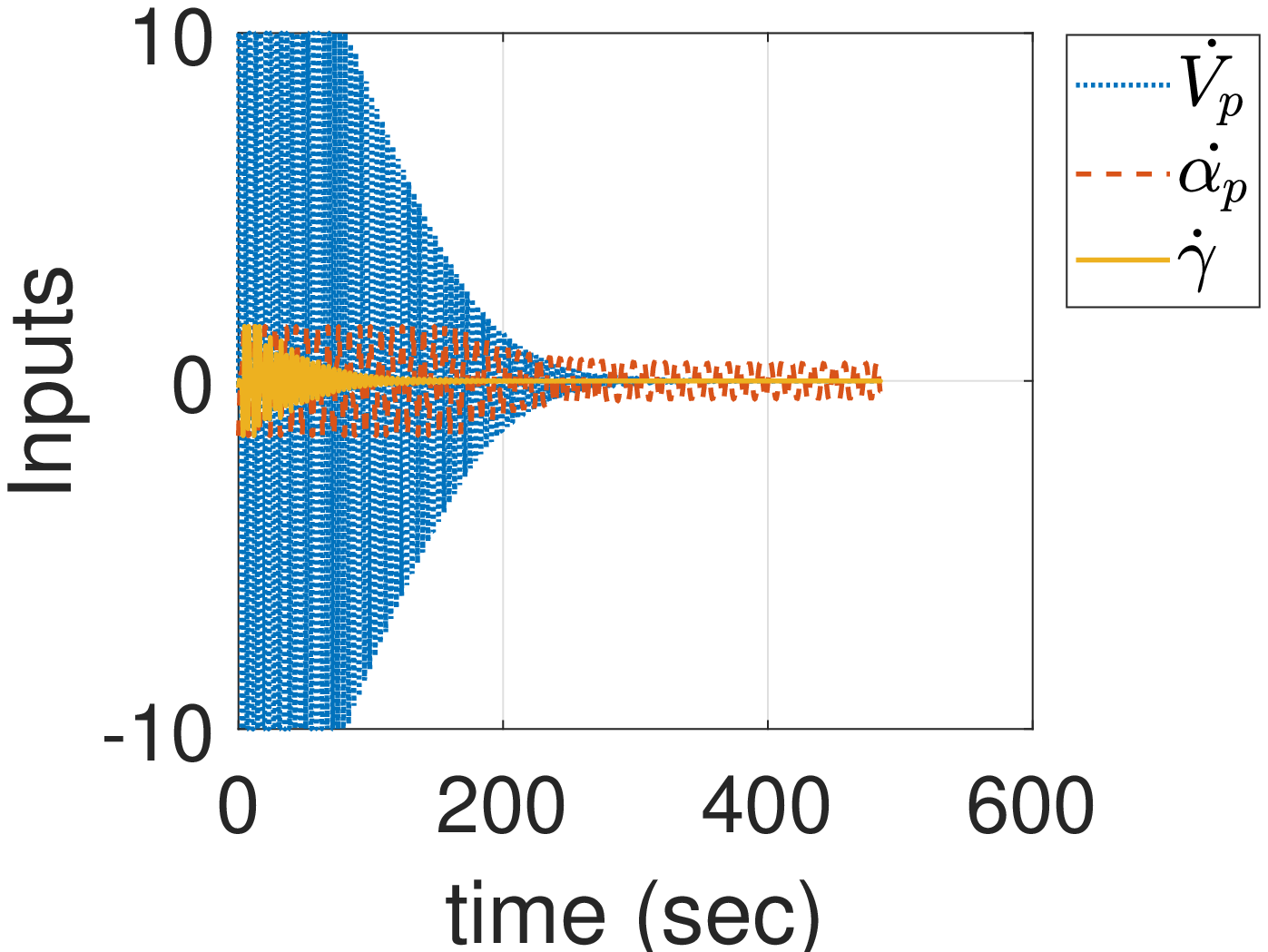}}%
\caption{Guidance Commands for UAV}
\label{fig:inputs}
\end{figure*}

\begin{figure*}[h!]
\centering
\subcaptionbox{Stationary Target\label{fig:stationaryheadspeed}}{\includegraphics[width=.2\textwidth,height=9.5cm,keepaspectratio,trim={1cm 0.3cm 0.0cm .08cm}]{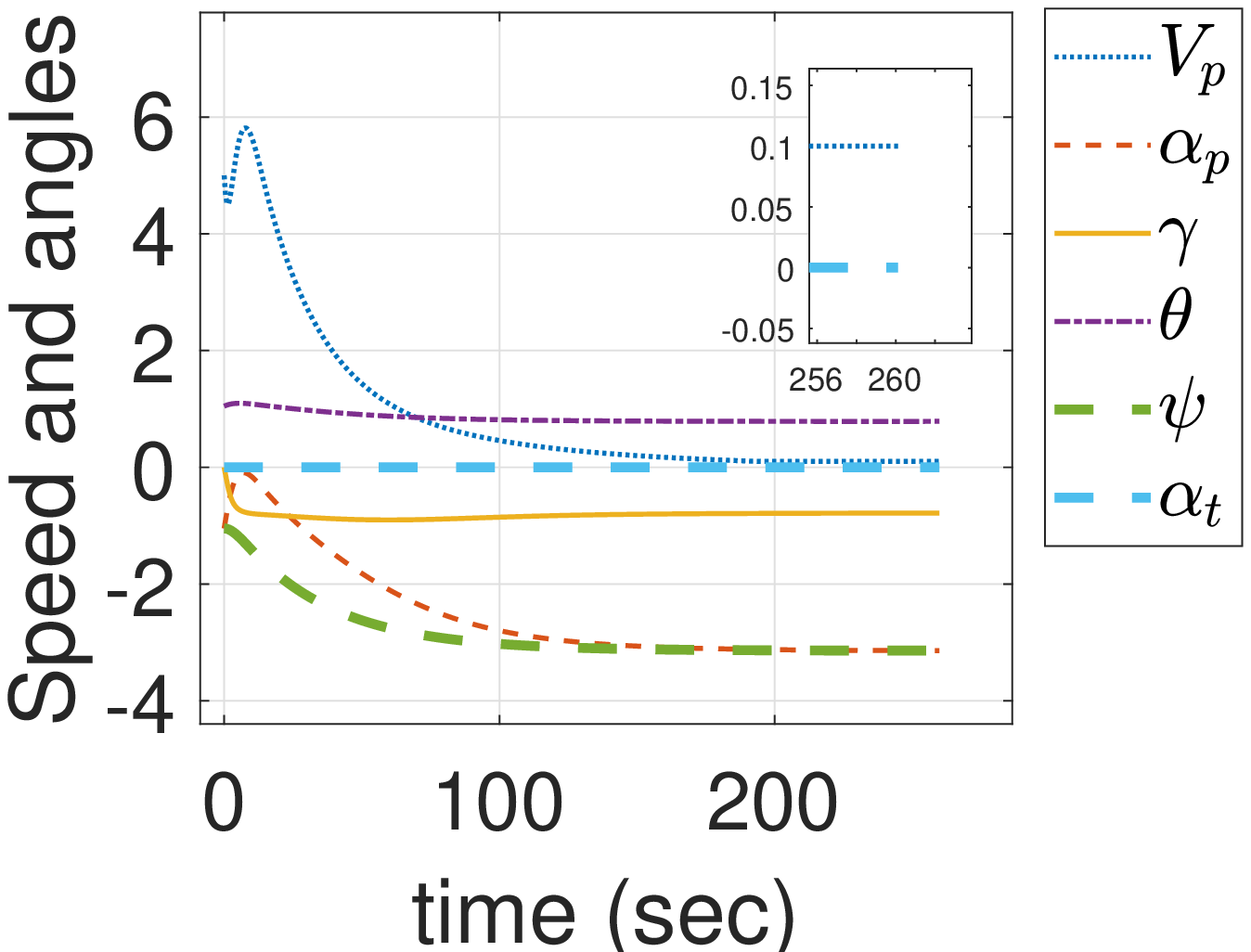}}%
\hfill 
\subcaptionbox{Non-Maneuvering Target\label{fig:slineheadspeed}}{\includegraphics[width=.2\textwidth,height=9.5cm,keepaspectratio,trim={1cm 0.3cm 0.0cm .08cm}]{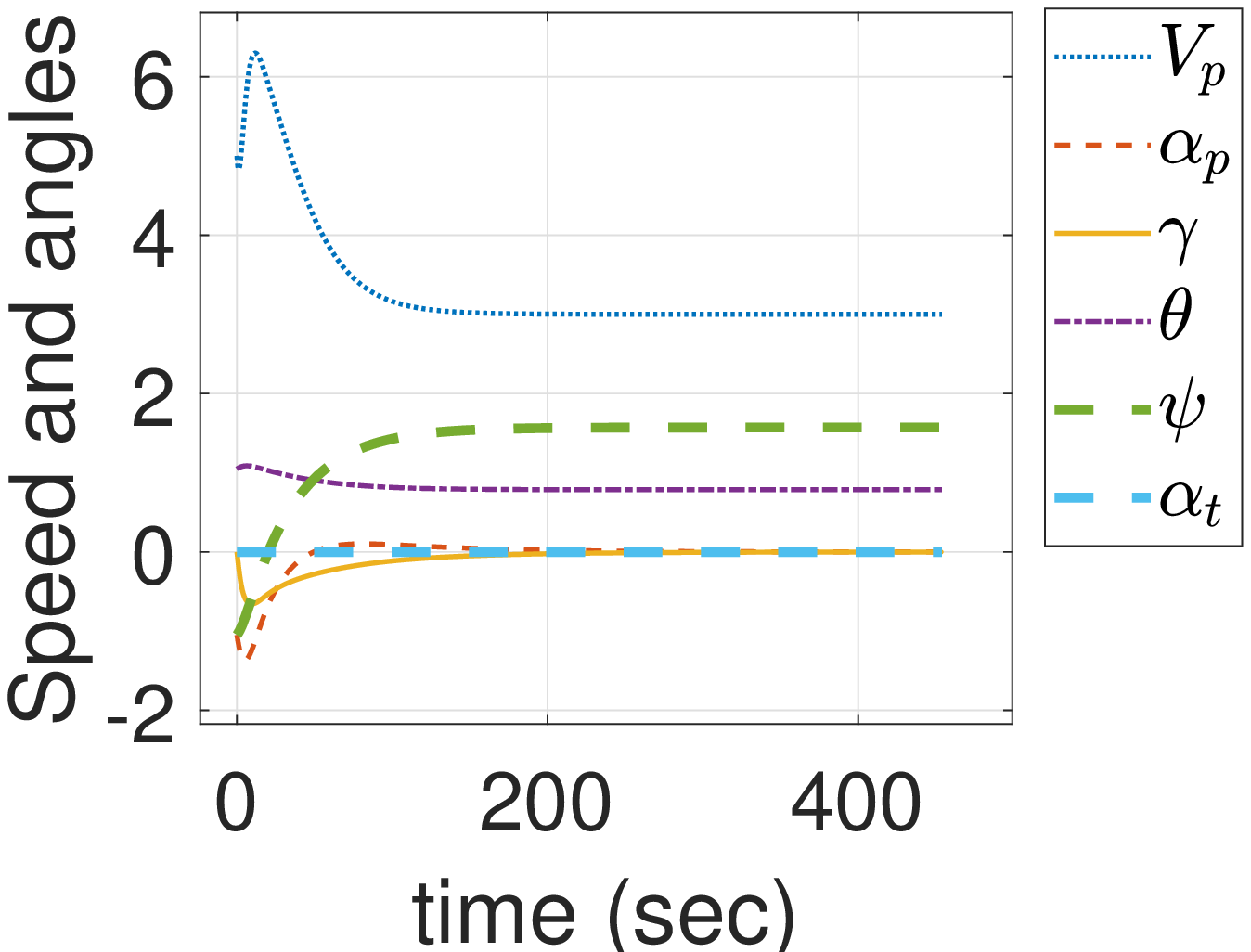}}%
\hfill 
\subcaptionbox{Constant Maneuvering Target\label{fig:circularheadspeed}}{\includegraphics[width=.2\textwidth,height=9.5cm,keepaspectratio,trim={1cm 0.3cm 0.0cm .08cm}]{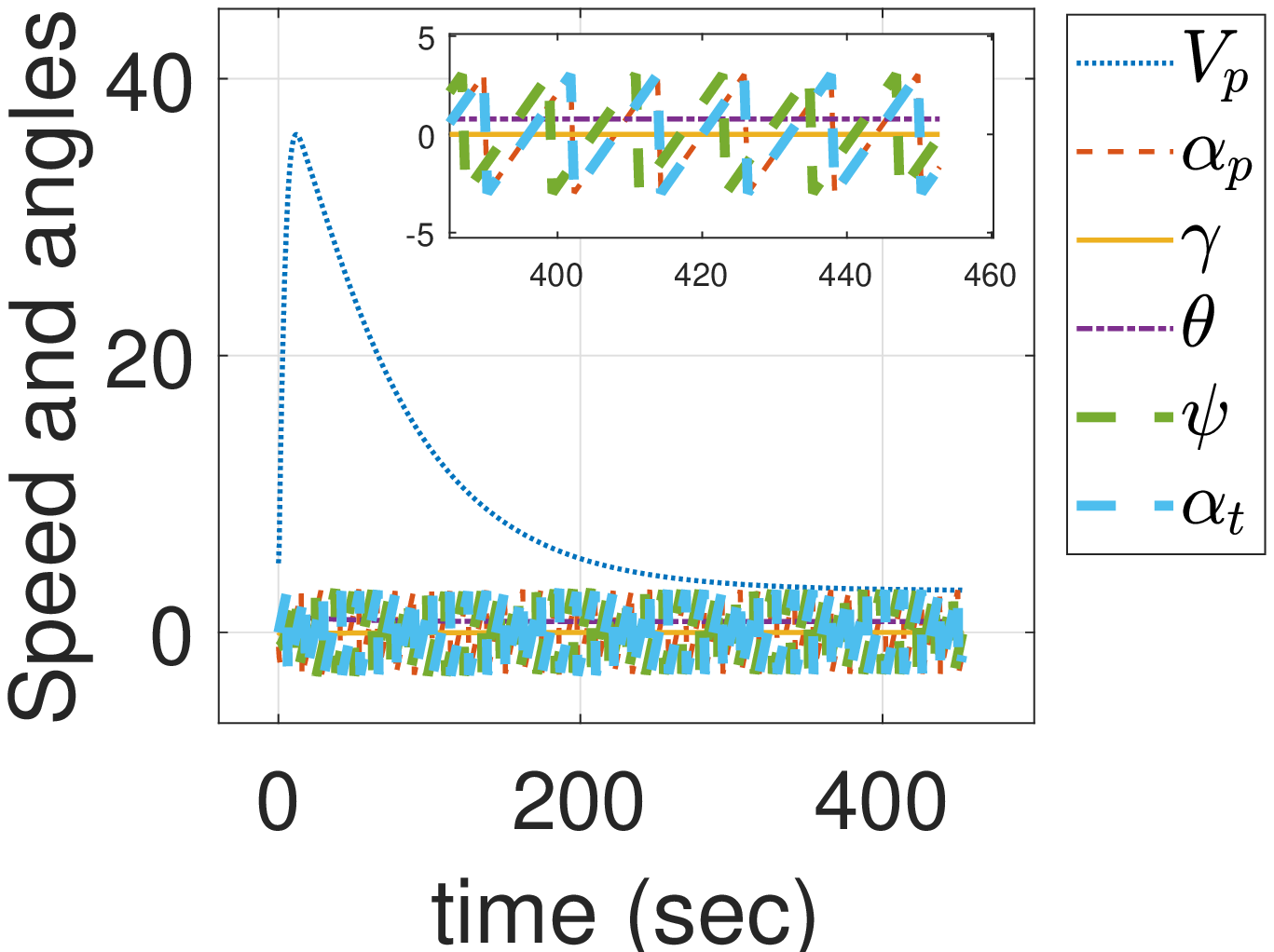}}%
\hfill 
\subcaptionbox{Sinusoidally Maneuvering Target\label{fig:sinusoidalheadspeed}}{\includegraphics[width=.2\textwidth,height=9.5cm,keepaspectratio,trim={1cm 0.3cm 0.0cm .08cm}]{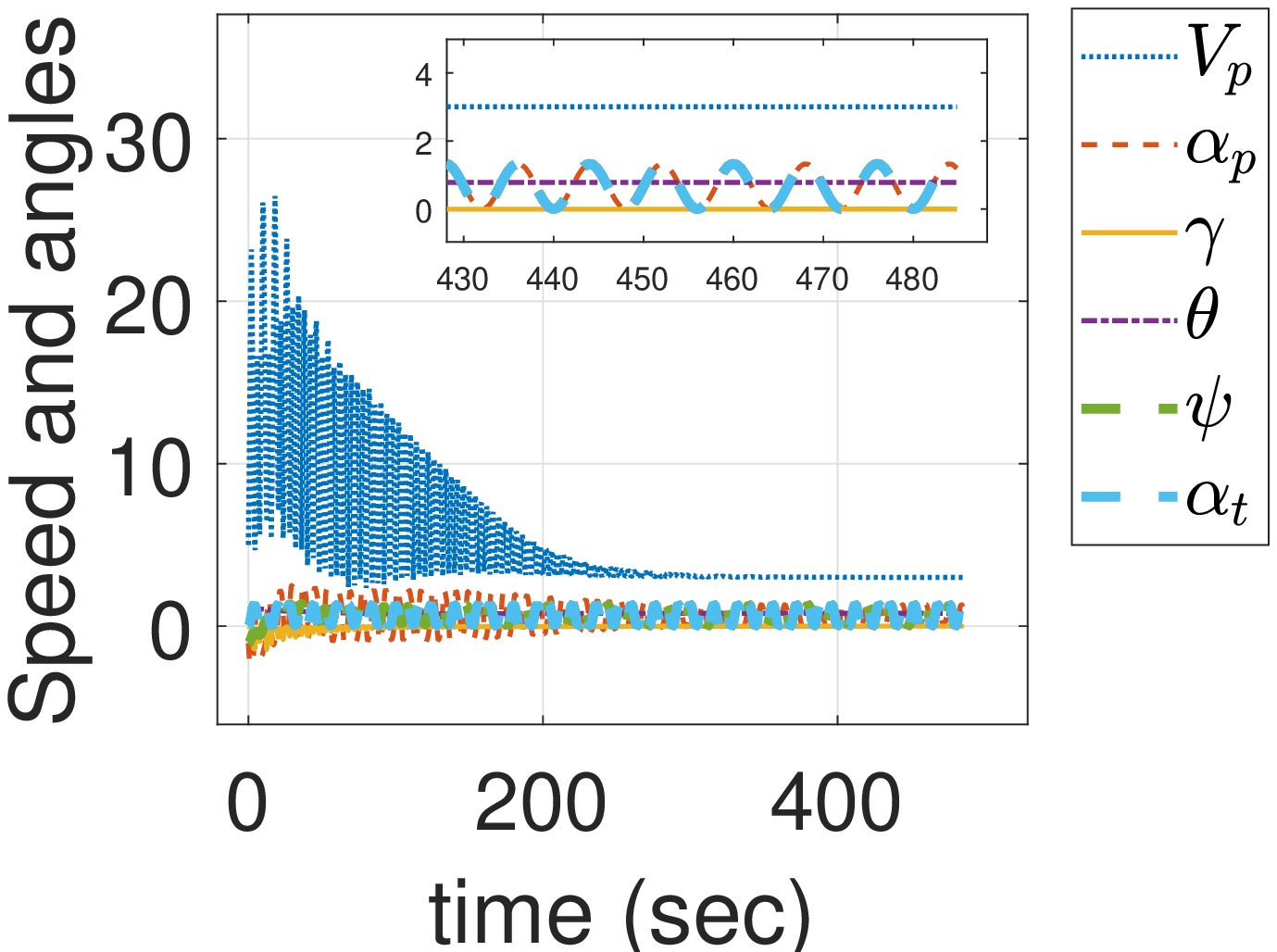}}%
\caption{UAV Speed and angles}
\label{fig:headspeeds}
\end{figure*}

\begin{figure*}[h!]
\centering
\subcaptionbox{Stationary Target\label{fig:stationarys}}{\includegraphics[width=.2\textwidth,height=9.5cm,keepaspectratio,trim={1cm 0.3cm 0.8cm .08cm}]{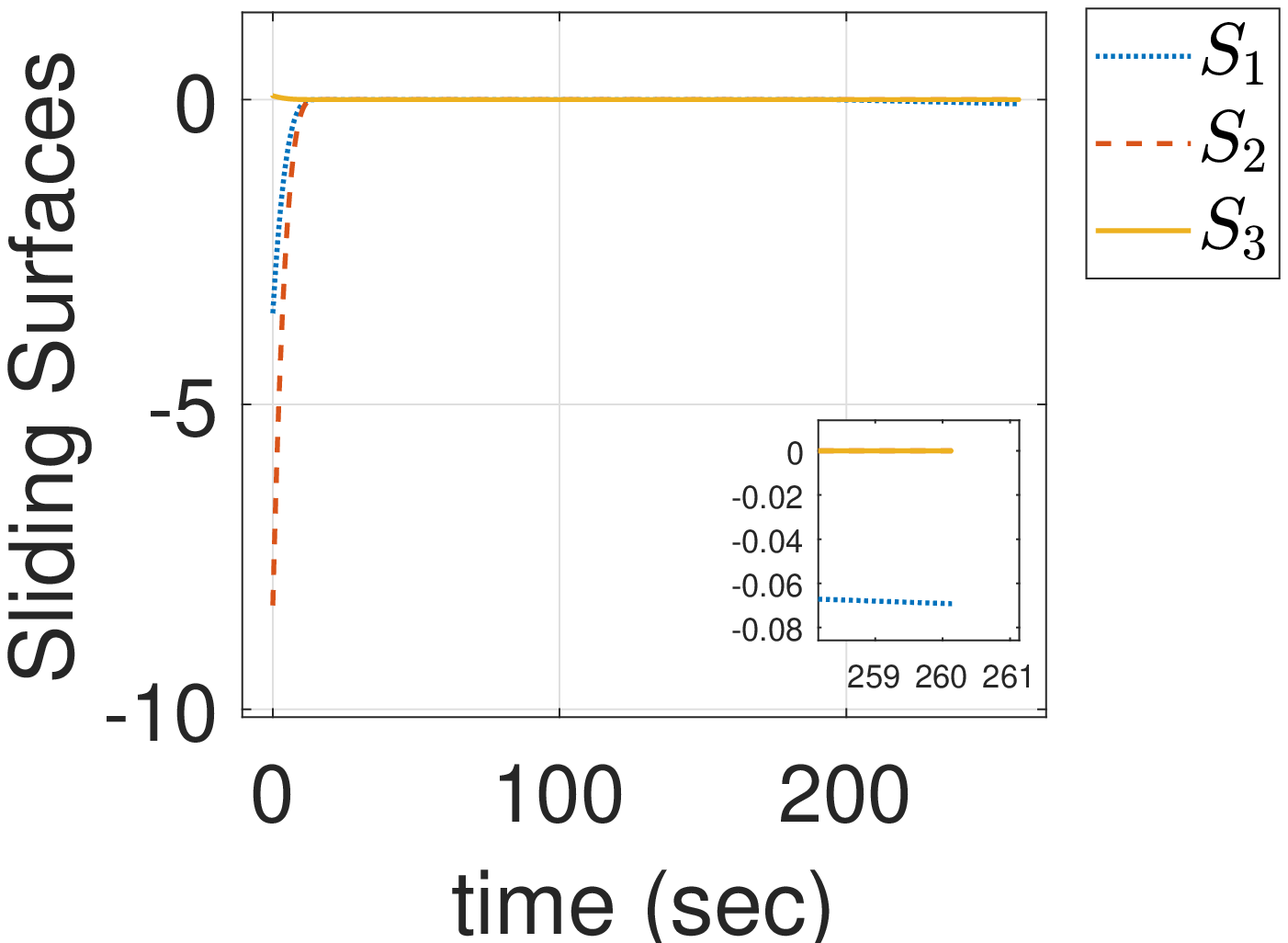}}%
\hfill 
\subcaptionbox{Non-Maneuvering Target\label{fig:slines}}{\includegraphics[width=.2\textwidth,height=9.5cm,keepaspectratio,trim={1cm 0.3cm 0.8cm .08cm}]{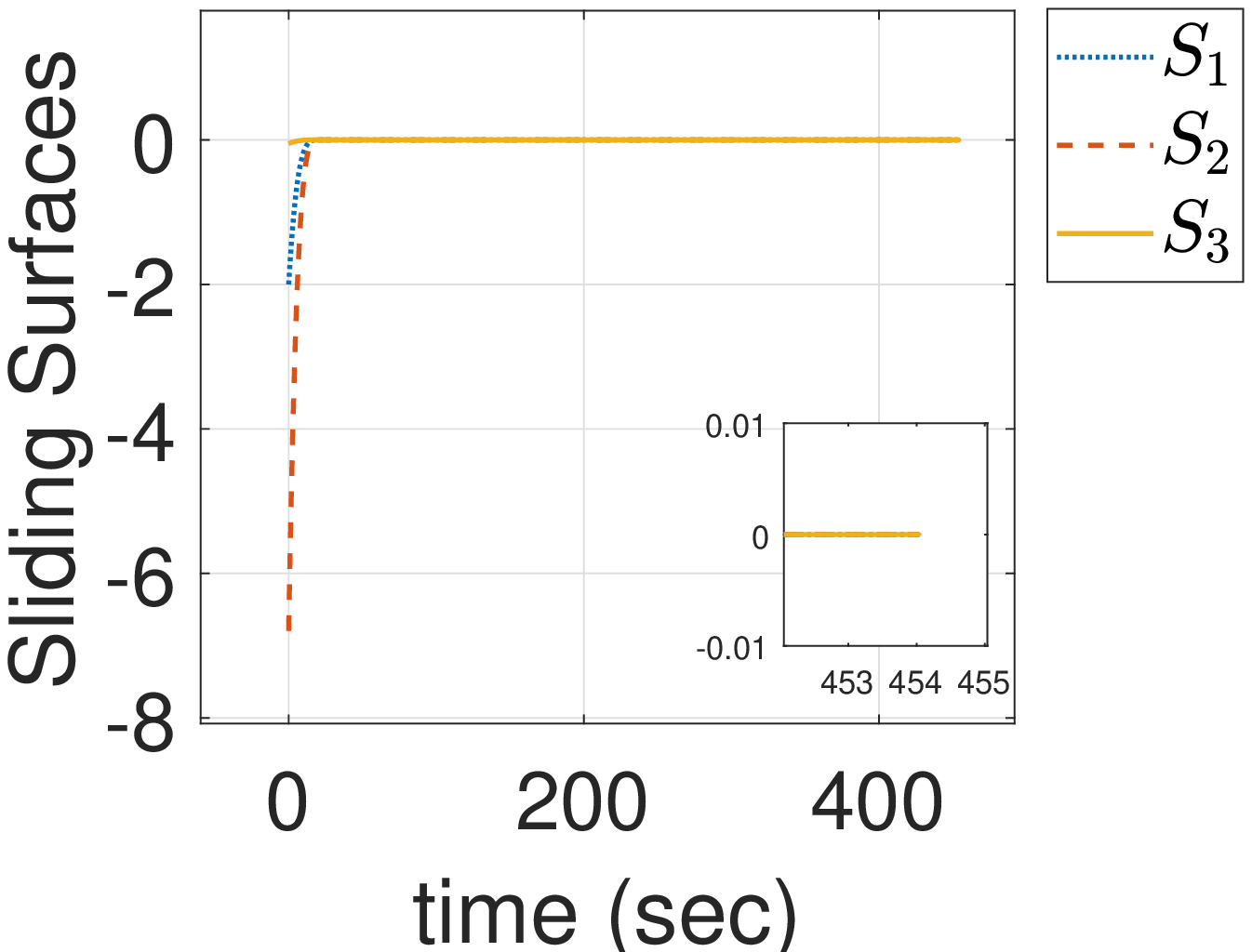}}%
\hfill 
\subcaptionbox{Constant Maneuvering Target\label{fig:circulars}}{\includegraphics[width=.2\textwidth,height=9.5cm,keepaspectratio,trim={1cm 0.3cm 0.8cm .08cm}]{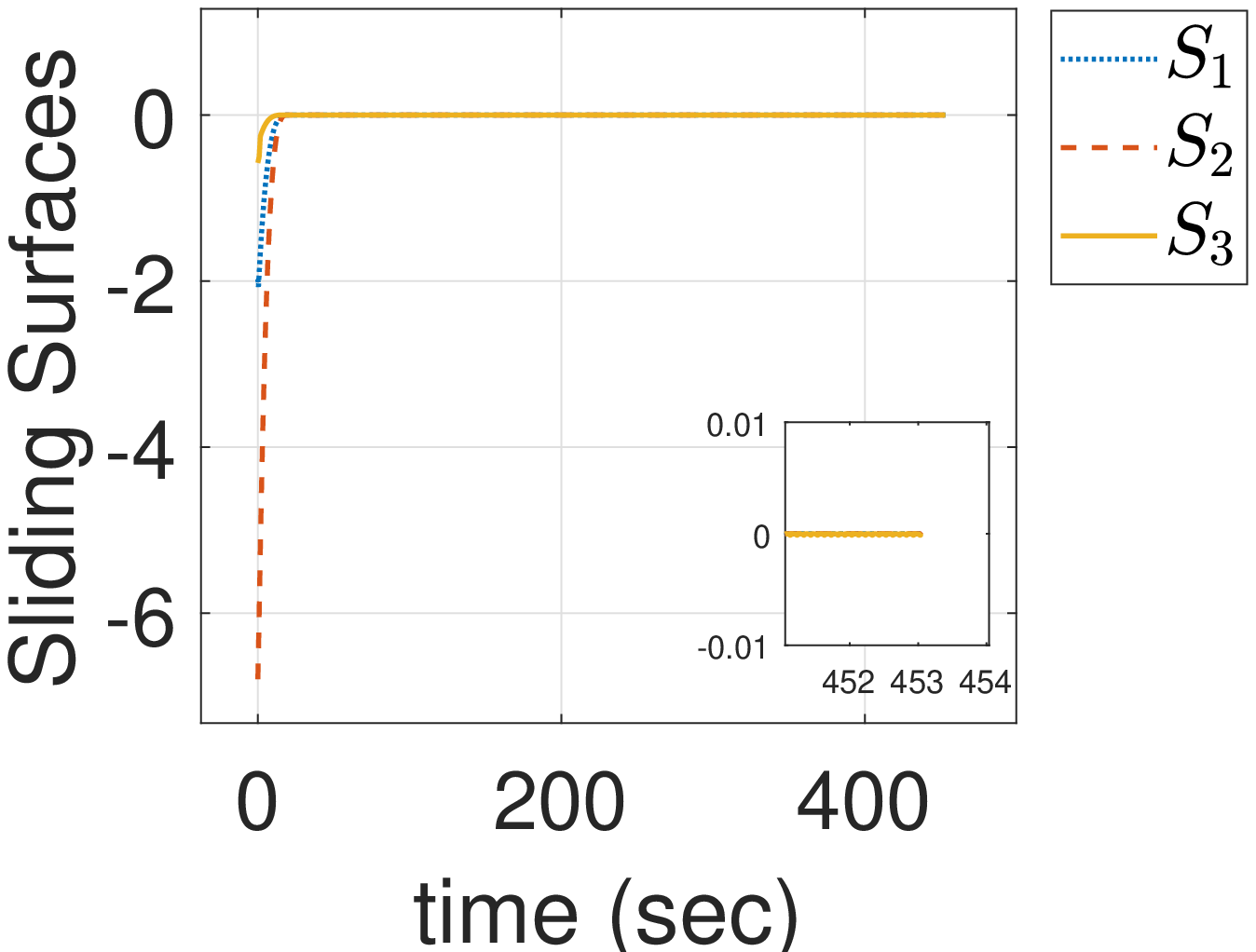}}%
\hfill 
\subcaptionbox{Sinusoidally Maneuvering Target\label{fig:sinusoidals}}{\includegraphics[width=.2\textwidth,height=9.5cm,keepaspectratio,trim={1cm 0.3cm 0.8cm .08cm}]{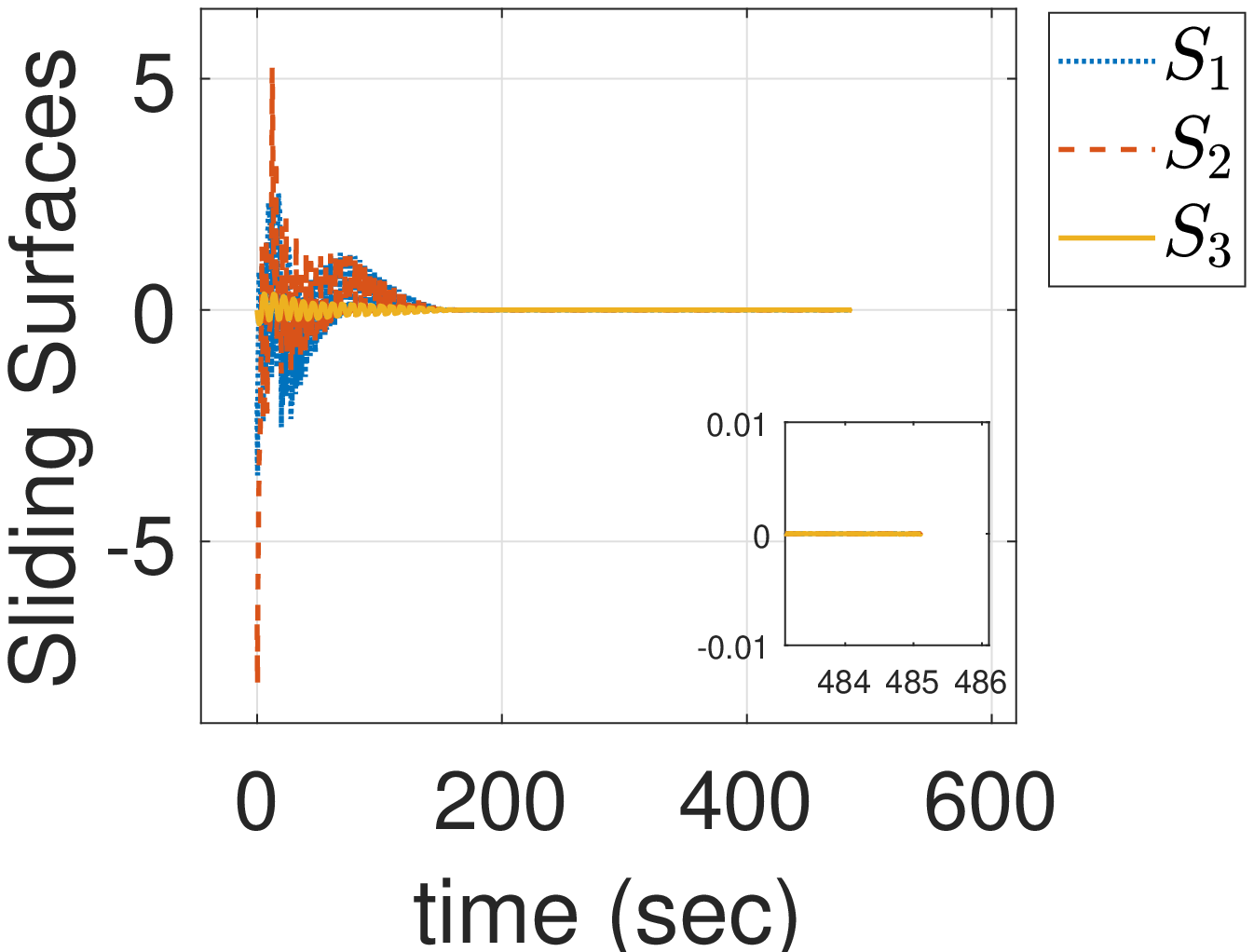}}%
\caption{Sliding Variables}
\label{fig:slideswithtime}
\end{figure*}

\subsection{Stationary Target :}\label{subsec:simsstationary}
In Section \ref{subsec:guidanceInputs}, it was mentioned that unique solutions to the system of equations (Eq. \eqref{eq:input}) exist only when $V_{p} \neq 0$ and $\cos(\gamma) \neq 0$. In the case of stationary targets, at the end of landing phase $V_{p} \to 0$ as $R \to 0$ for a smooth landing. This poses a problem in the end-game phase, 
which is avoided by suitable manipulation suggested in Algorithm \ref{alg:inputsummary}. As a consequence, the speed of the pursuer is very close to zero at touchdown ($\lim_{R \to 0} V_p = 0.1$ m/s) as can be seen in Fig. \ref{fig:stationaryheadspeed}, and there is a slight error in the convergence of sliding variable $S_1$ as can be seen in Fig. \ref{fig:stationarys}. Since the errors in the range rates and the sliding variables at touchdown are sufficiently small, they could be neglected in most practical scenarios. More importantly, from Fig. \ref{fig:stationaryheadspeed}, it is noted that at touchdown, $\lim_{R \to 0}(\psi-\alpha_t) = -\pi$ rad (for stationary target $\alpha_t=0$) and $\lim_{R \to 0}\theta = \pi/4$ rad, which satisfy the desired approach angles indicated in Table \ref{tab:1initials}. Near to the start of the engagement, a maximum UAV speed of 6 m/s is demanded, which is within acceptable bounds for the UAV.

\subsection{Non-maneuvering Moving Target :}\label{subsec:simsnonman}
In this case, the UAV is to land on a non-maneuvering moving target with a heading direction, $\alpha_{t}=0$. From Fig. \ref{fig:slinedist}, it can be seen that $\dot{R_{xy}}$ and $\dot{R_{z}}$ converge to zero at touchdown, implying successful soft landing on the target. From Fig. \ref{fig:slineheadspeed} it can be observed that at touchdown, $\lim_{R \to 0}(\psi-\alpha_t) = \pi/2$ rad, $\lim_{R \to 0}\theta = \pi/4$ rad, $\lim_{R \to 0}V_p = \lim_{R \to 0}V_t= 3$ m/s, $\lim_{R \to 0}\alpha_p = \lim_{R \to 0}\alpha_t=0$ rad and $\lim_{R \to 0}\gamma=0$ rad, which satisfy the objectives of Soft-landing and achieving desired approach angles as given in Eqs. \eqref{eq:desiredheadspeed} and \eqref{eq:desiredapproach}. In this case also near to the start of the engagement, a maximum UAV speed of 7 m/s is demanded, which is within acceptable bounds for the UAV for landing on the target moving at 3 m/s.

\subsection{Constant Maneuvering Target :}\label{subsec:simscircular}
In this case, the target executes a circular motion with an initial heading angle, $\alpha_{t0} = 0$ rad with a constant maneuver rate as shown in Table \ref{tab:1initials}. As seen from Fig. \ref{fig:circulars}, the sliding variables converge in a short time, and thus, the trajectory of the UAV is largely governed by the dynamics of the system on the sliding mode, in which $\dot\alpha_p$ remains almost constant (=$\pi/6$ rad/s) following constant $\dot\alpha_p$, as can be seen from Fig. \ref{fig:circularinput}. Besides this, in the sliding mode $R_{xy}$ and $R_z + \tan{\theta_{des}} R_{xy}$  decreases exponentially. As a consequence, it is observed in Fig. \ref{fig:circulartraj} that the UAV's trajectory is helical in nature, with the radius decreasing as the distance between the UAV and the target decreases. This also leads to a very high requirement of the UAV speed in the initial phase of the engagement, which is reflected in Fig. \ref{fig:circularheadspeed} and in very high guidance command $\dot{V_p}$ as can be seen in Fig. \ref{fig:circularinput}. This problem of high speed and guidance commands for maneuvering targets is further explained in Section \ref{subsec:improvement}. It may also be noted from Figs. \ref{fig:circulars} and \ref{fig:2circularheadspeed} that the sliding variables converge in the expected fashion, and the objectives of soft landing and achieving desired approach angles as given in Eqs. \eqref{eq:desiredheadspeed} and \eqref{eq:desiredapproach}, respectively, are satisfied ($\lim_{R \to 0}(\psi-\alpha_t) = \pi/2$ rad, $\lim_{R \to 0}\theta = \pi/4$ rad, $\lim_{R \to 0}V_p = \lim_{R \to 0}V_t= 3$ m/s, $\lim_{R \to 0}\alpha_p = \lim_{R \to 0}\alpha_t=0$ rad and $\lim_{R \to 0}\gamma=0$ rad).

\subsection{Sinusoidally Maneuvering Target :}\label{subsec:simssinusoidal}
In this case, the target is maneuvering in a sinusoidal fashion as given in Table \ref{tab:1initials} with the initial heading angle $\alpha_{t0} = 0$ rad. The initial conditions are set to be the same as the other cases considered for simulations in Sections \ref{subsec:simsstationary} - \ref{subsec:simscircular}. Similar to the constant-maneuvering target case, in this case also required $V_p$ and guidance command $\dot{V}_p$ become quite high near the start of the engagement as can be seen in Figs. \ref{fig:sinusoidalheadspeed} and \ref{fig:sinusoidalinput}, respectively. The variation of $\dot{\alpha}_p$ also tends to follow a sinusoidal pattern similar to that of $\dot{\alpha}_t$. 
From Fig. \ref{fig:sinusoidalinput}, it is noted that the guidance commands are oscillatory in nature, with the amplitude of oscillation high initially as during this time the target maneuvers at large distances from the UAV. Since $\dot{V_p}$ was capped at $10 m/s^2$ as indicated in Section \ref{sec:results}, guidance command $\dot{V_p}$ fluctuates rapidly between -10 and +10 $m/s^2$ near the start of the engagement. Consequently, the desired speed for the UAV is found to fluctuate rapidly in Fig. \ref{fig:sinusoidalheadspeed}. Sliding variables also take longer time to converge as seen from Fig. \ref{fig:sinusoidals}.
\newline

\section{Two-Phase Guidance Scheme}\label{sec:twophase}

\subsection{Motivation for two-phase scheme}\label{subsec:improvement}
From Fig. \ref{fig:inputs}, it can be seen that the magnitudes of the guidance commands generated by the presented guidance scheme Algorithm \ref{alg:inputsummary} are maximum at the start of the engagement for all types of targets considered for simulations. This follows mainly from Eq. \eqref{eq:Sdotdefn} that dictates the magnitudes of the rate of change of sliding variable ($|\dot{S}_i|$) to be higher when magnitude of the sliding variable ($|(S_i)|$) is higher. A higher magnitude of $\dot{S_i}$ as well as large $R_{xy}$ at the initial portion of the engagement results in higher guidance command. 
This effect is found to be more severe in case of maneuvering UGV due to the dynamics of $S_3=\psi - \alpha_t - \zeta_{des}$ 
containing the terms like non-zero $\dot{\alpha}_t$ and higher magnitude of $S_3$ getting multiplied by higher $R_{xy}$ as can be seen in Eqs. \eqref{eq:expand1}-\eqref{eq:expand3}. This poses a drawback in the performance of the presented guidance scheme (Algorithm \ref{alg:inputsummary}). To obviate this problem a two-phase guidance scheme is posed in this section, in which a pre-fixed desired value for $\psi$ is achieved at the first phase rather than achieving desired $\psi - \alpha_{t}$ in the first phase itself, while control of the terminal azimuth angle w.r.t. $\alpha_t$ is given priority in the second phase.

\subsection{Synthesis of Guidance command }\label{subsec:2phasesynthesis}
Following the discussion given above in Section \ref{subsec:improvement}, a switching condition is presented in this section based on the magnitude of $R_{xy}$. The sliding variables considered for the two-phase guidance scheme are now modified as follows:

if  $R_{xy} > R_{switch}$ (Phase-1)
\begin{equation} \label{eq:Sdefn2i}
{S = 
\begin{bmatrix}
\dot{R}_{xy}+k_{a}R_{xy} \\ 
\dot{R_{z}}+\tan(\theta_{des})\dot{R}_{xy} + k_{b}(R_{z} + \tan(\theta_{des})R_{xy})  \\ 
(\dot{\psi}) + k_{c}(\psi - ( \xi))
\end{bmatrix}
}
\end{equation}

else (Phase-2)
\begin{equation} \label{eq:Sdefn2ii}
{S = 
\begin{bmatrix}
\dot{R}_{xy}+k_{a}R_{xy} \\ 
\dot{R_{z}}+\tan(\theta_{des})\dot{R}_{xy} + k_{b}(R_{z} + \tan(\theta_{des})R_{xy})  \\ 
(\dot{\psi} - \dot{\alpha_{t}}) + k_{c}(\psi - (\alpha_{t} + \zeta_{des}))
\end{bmatrix}
}
\end{equation}

Note that in each phase of this two-phase guidance scheme also, dynamics of the sliding variables are considered same as in Eq. \eqref{eq:Sdotdefn}, and guidance commands are obtained following same method as discussed in Section \ref{subsec:guidanceInputs}. Stability in each phase can also be justified following similar logic as in the proof of Theorem \ref{th:stability}. Here, $R_{switch}$ is an important parameter, which should not be selected as very high to avoid the problems posed by the guidance scheme in Algorithm \ref{alg:inputsummary}. However, it should not be too small to achieve an allowable bounded azimuth angle dynamics. So, a trade-off is needed in tuning $R_{switch}$, which could be achieved based on offline trials.

\begin{figure*}[b!]
\centering
\subcaptionbox{Stationary Target\label{fig:2stationarytraj}}{\includegraphics[width=0.2\textwidth,height=9.5cm,keepaspectratio]{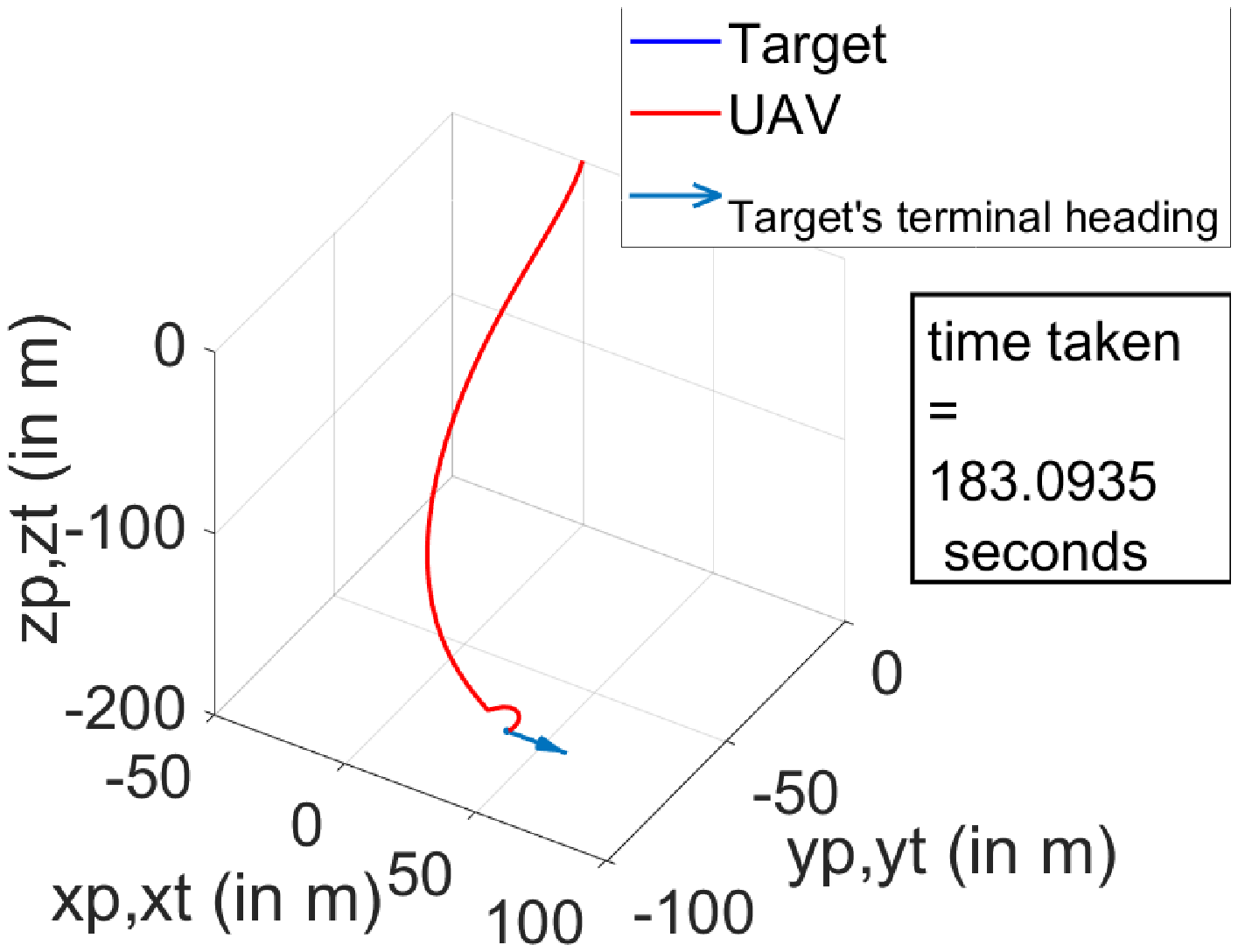}}%
\hfill 
\subcaptionbox{Non-Maneuvering Target\label{fig:2slinetraj}}{\includegraphics[width=0.2\textwidth,height=9.5cm,keepaspectratio]{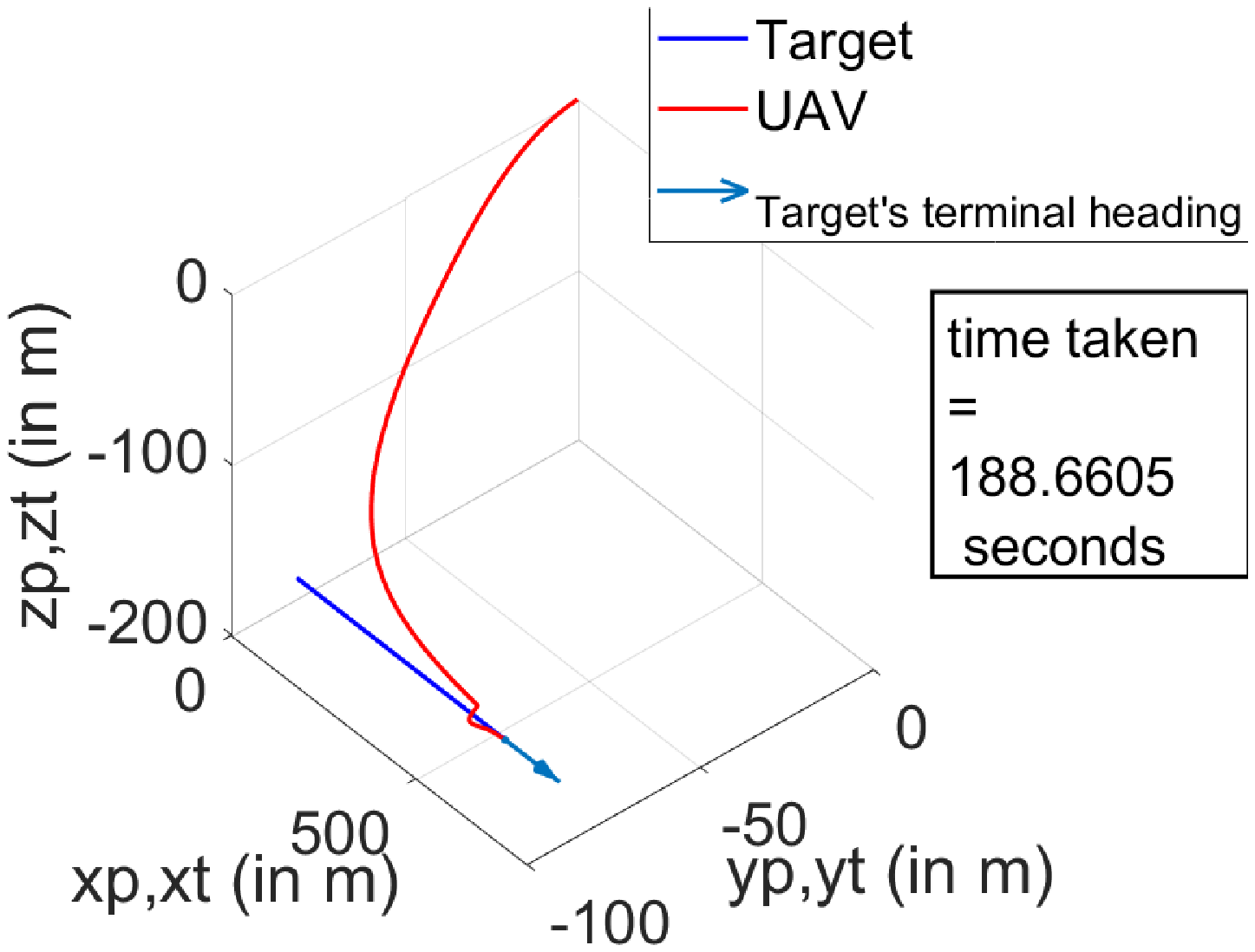}}%
\hfill 
\subcaptionbox{Constant Maneuvering Target\label{fig:2circulartraj}}{\includegraphics[width=0.2\textwidth,height=9.5cm,keepaspectratio]{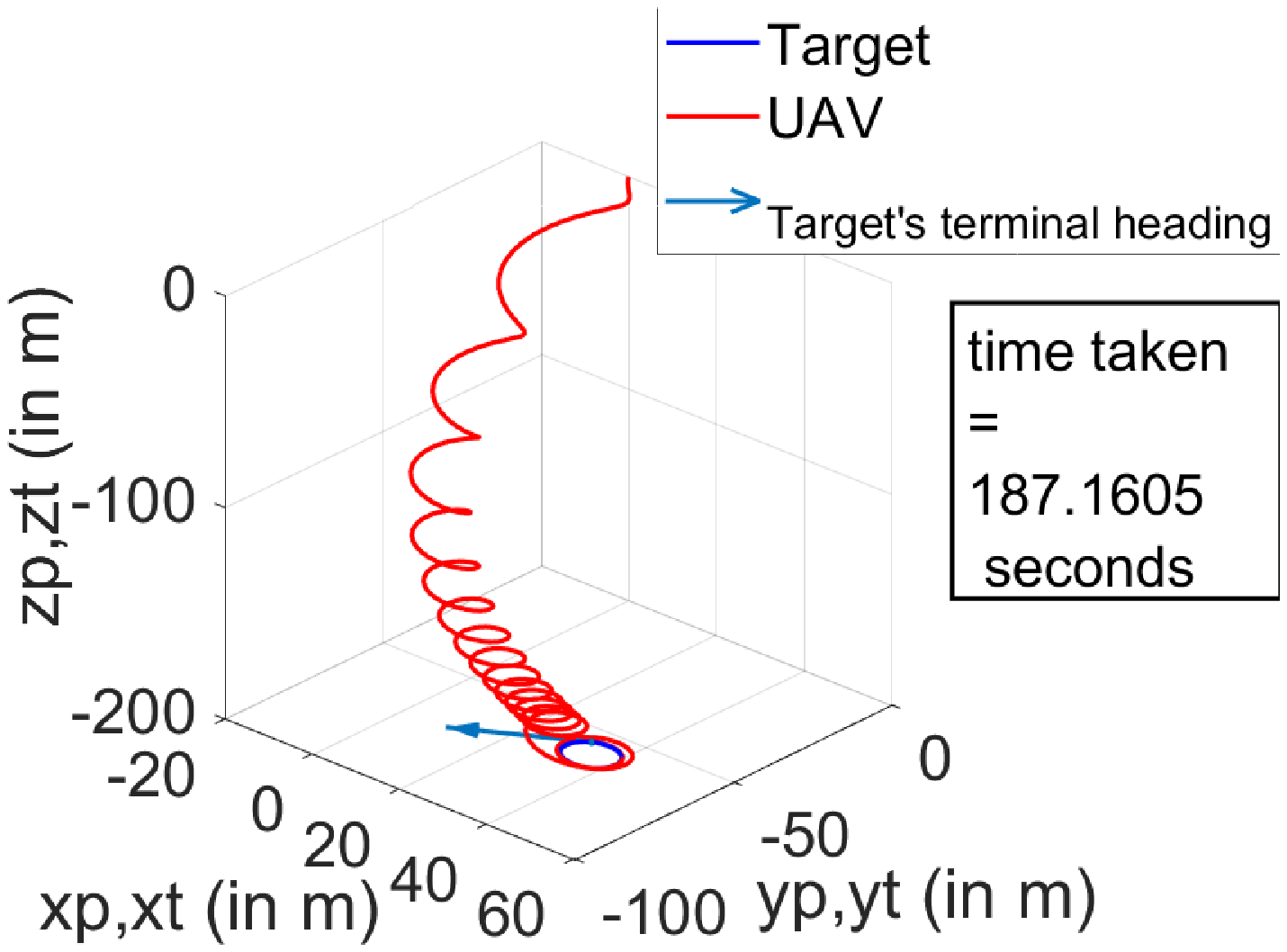}}%
\hfill 
\subcaptionbox{Sinusoidally Maneuvering Target\label{fig:2sinusoidaltraj}}{\includegraphics[width=0.2\textwidth,height=9.5cm,keepaspectratio]{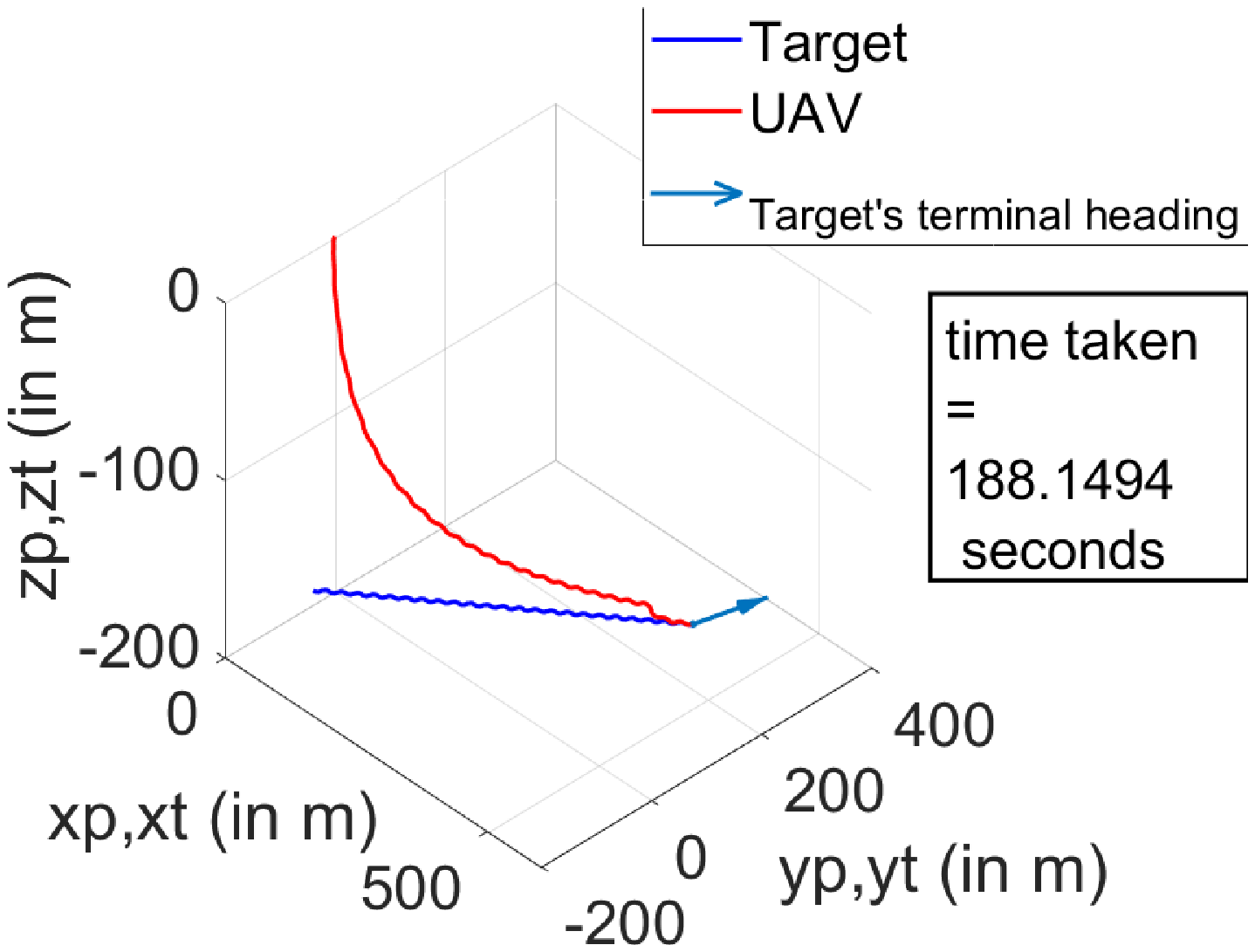}}%
\caption{Trajectory plots for UAV and target}
\label{fig:2trajs}
\end{figure*}

\begin{figure*}[b!]
\centering
\subcaptionbox{Stationary Target\label{fig:2stationarydist}}{\includegraphics[width=.2\textwidth,height=9.5cm,keepaspectratio,trim={0.5cm 0.0cm 0.0cm .08cm}]{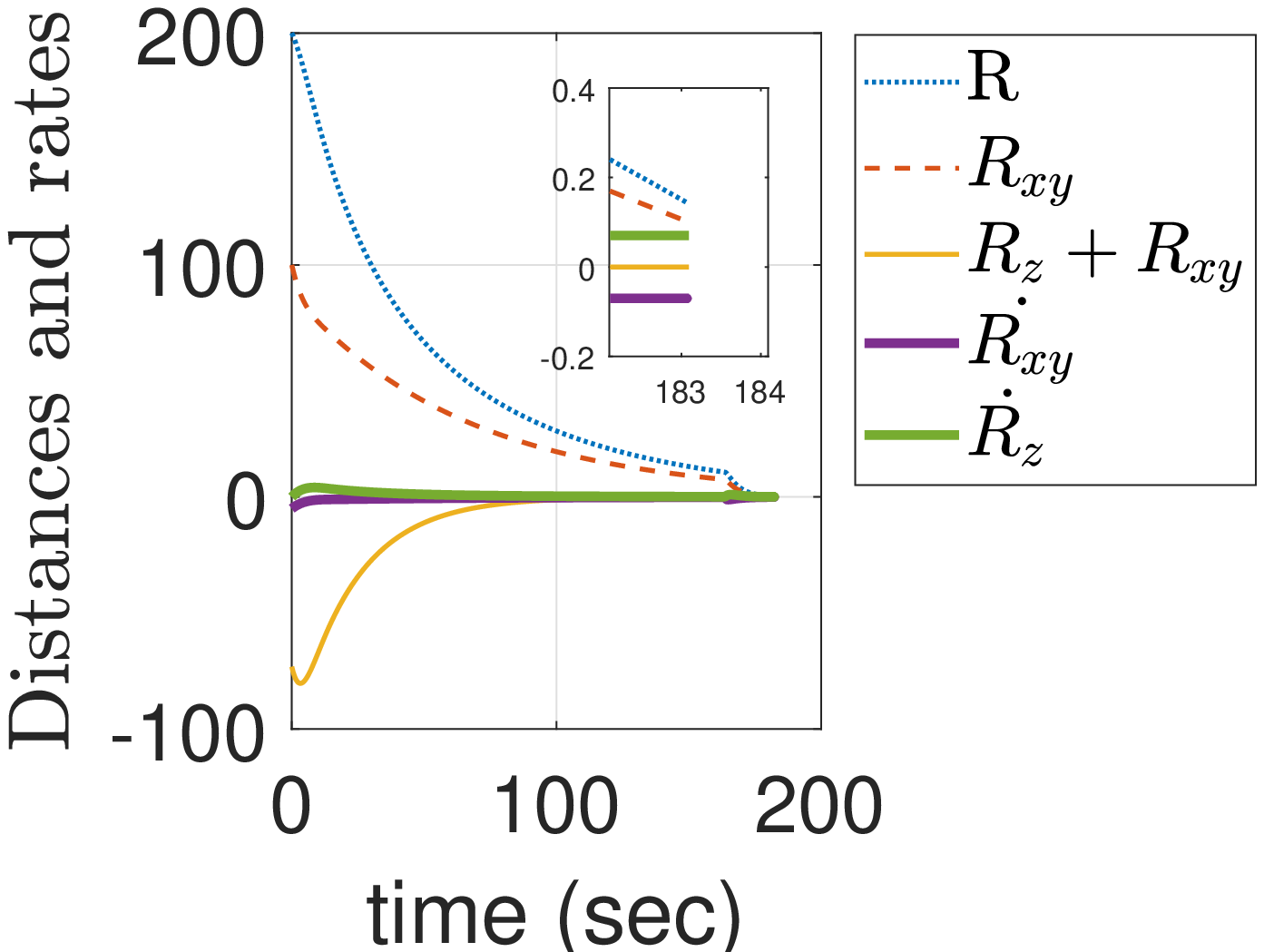}}%
\hfill 
\subcaptionbox{Non-Maneuvering Target\label{fig:2slinedist}}{\includegraphics[width=.2\textwidth,height=9.5cm,keepaspectratio,trim={0.5cm 0.0cm 0.0cm .08cm}]{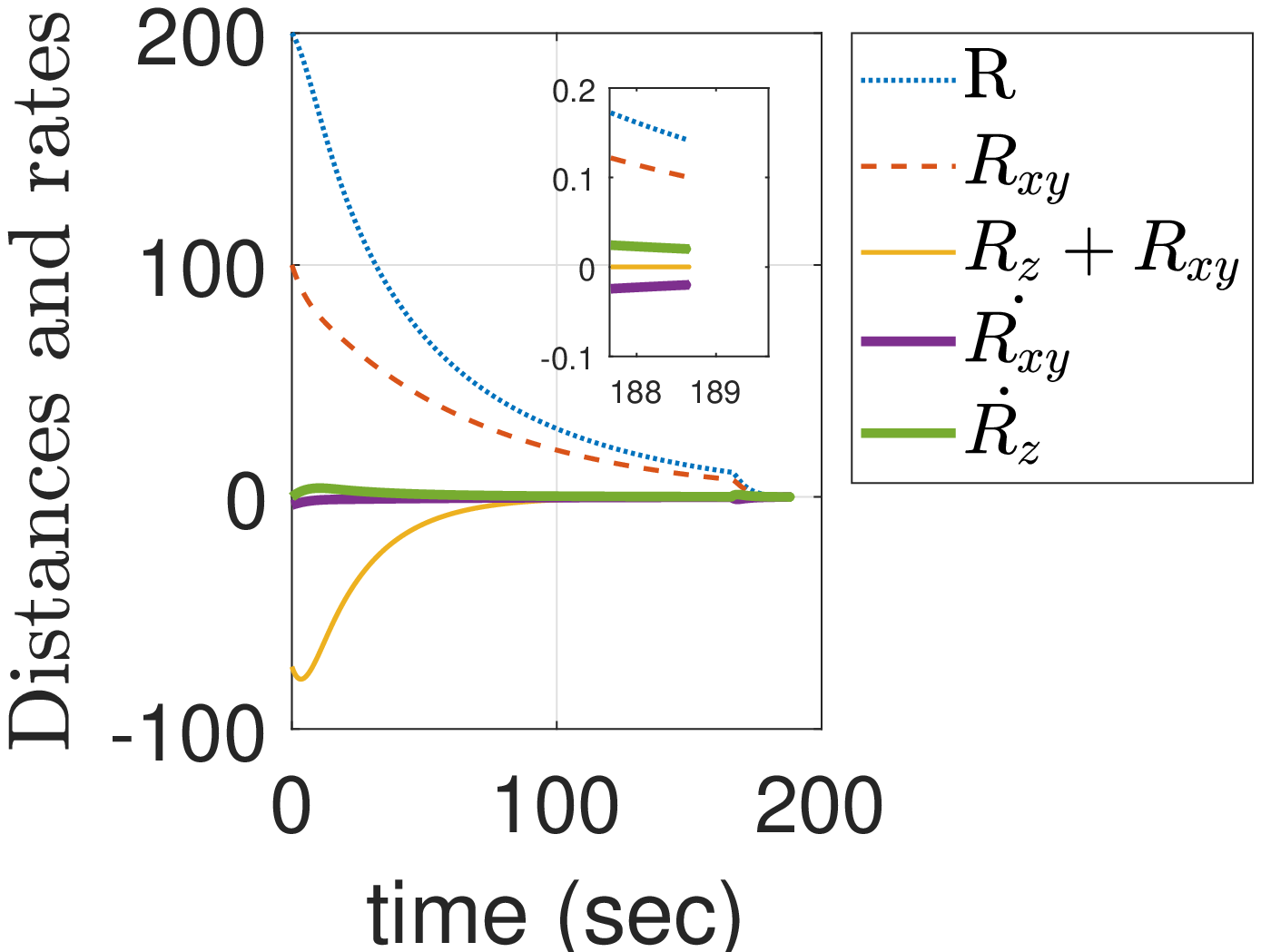}}%
\hfill 
\subcaptionbox{Constant Maneuvering Target\label{fig:2circulardist}}{\includegraphics[width=.2\textwidth,height=9.5cm,keepaspectratio,trim={0.5cm 0.0cm 0.0cm .08cm}]{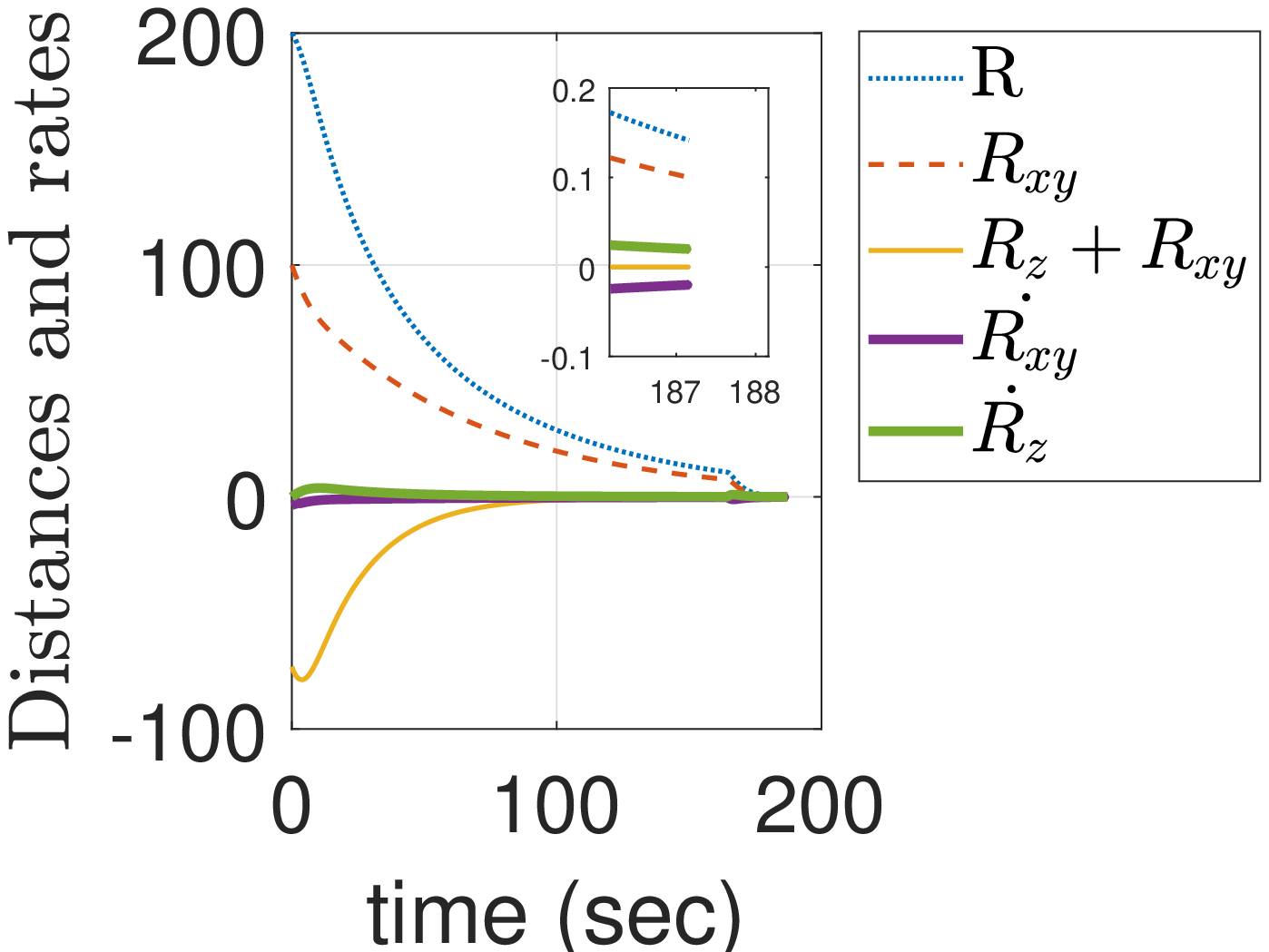}}%
\hfill 
\subcaptionbox{Sinusoidally Maneuvering Target\label{fig:2sinuosidaldist}}{\includegraphics[width=.2\textwidth,height=9.5cm,keepaspectratio,trim={0.5cm 0.0cm 0.0cm .08cm}]{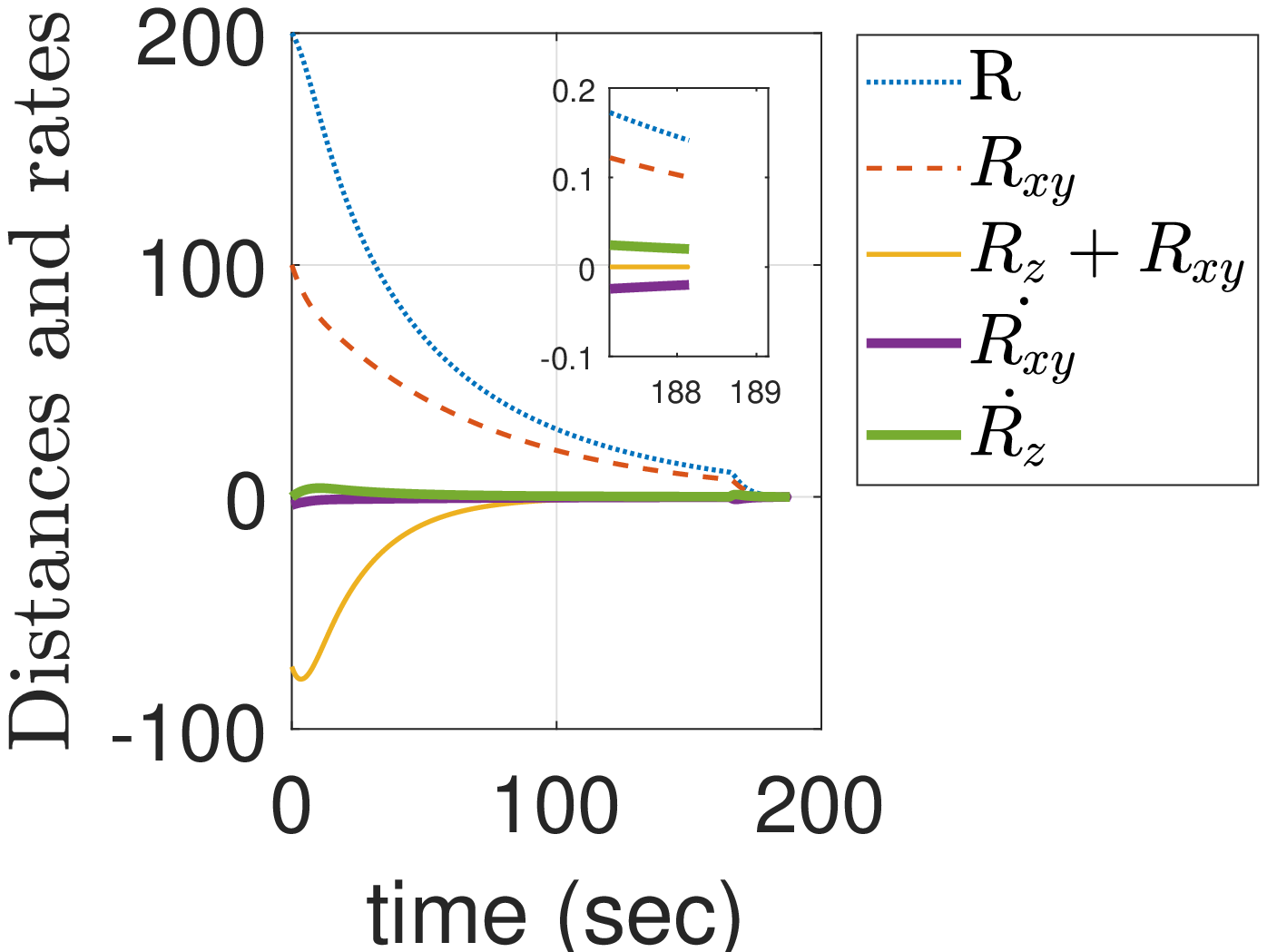}}%
\caption{Distance from target and it's projections on the xy-plane and along the z-axis}
\label{fig:2dists}
\end{figure*}

\begin{figure*}[b!]
\centering
\subcaptionbox{Stationary Target\label{fig:2stationaryinput}}{\includegraphics[width=.2\textwidth,height=9.5cm,keepaspectratio,trim={1cm 0.3cm 0.8cm .08cm}]{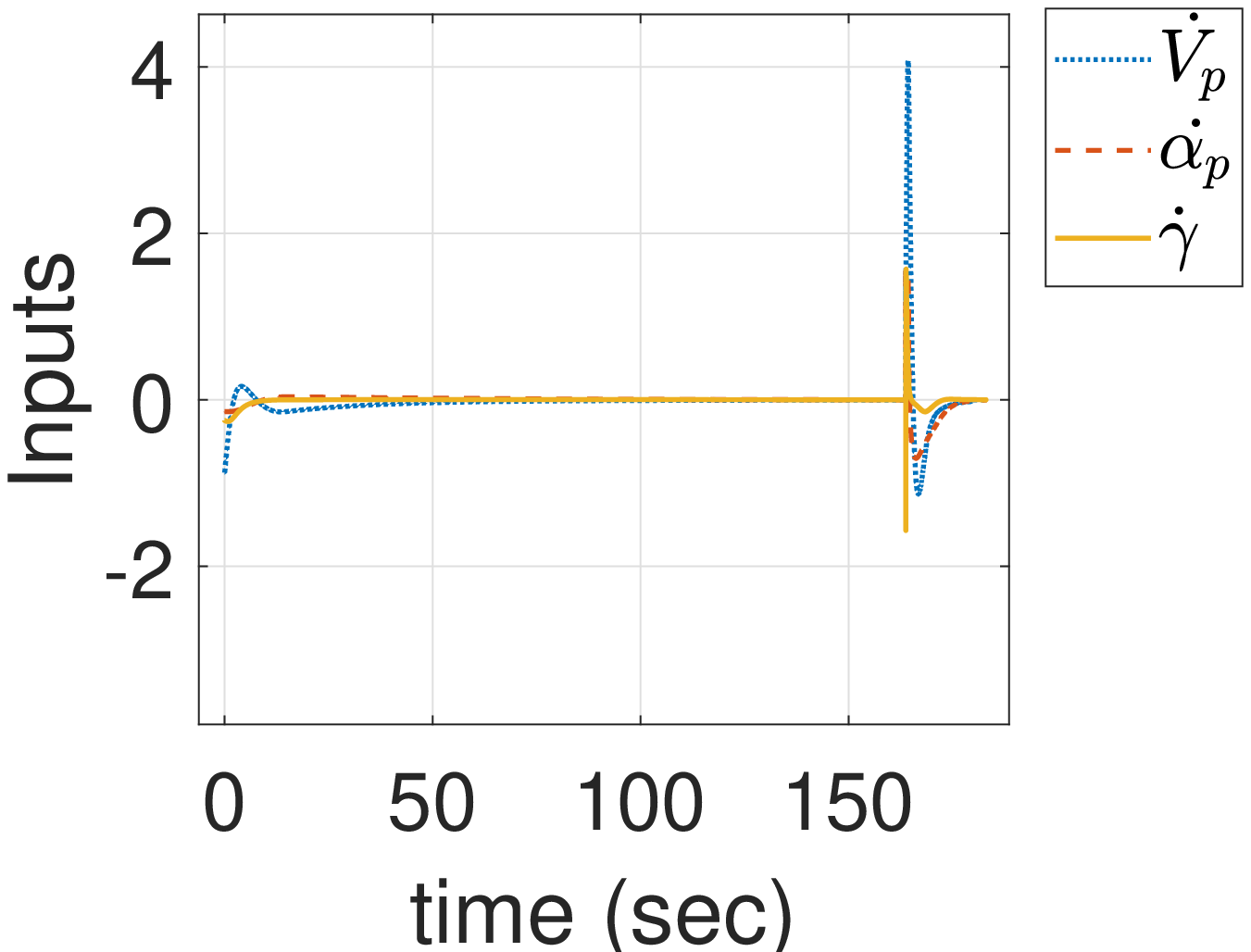}}%
\hfill 
\subcaptionbox{Non-Maneuvering Target\label{fig:2slineinput}}{\includegraphics[width=.2\textwidth,height=9.5cm,keepaspectratio,trim={1cm 0.3cm 0.8cm .08cm}]{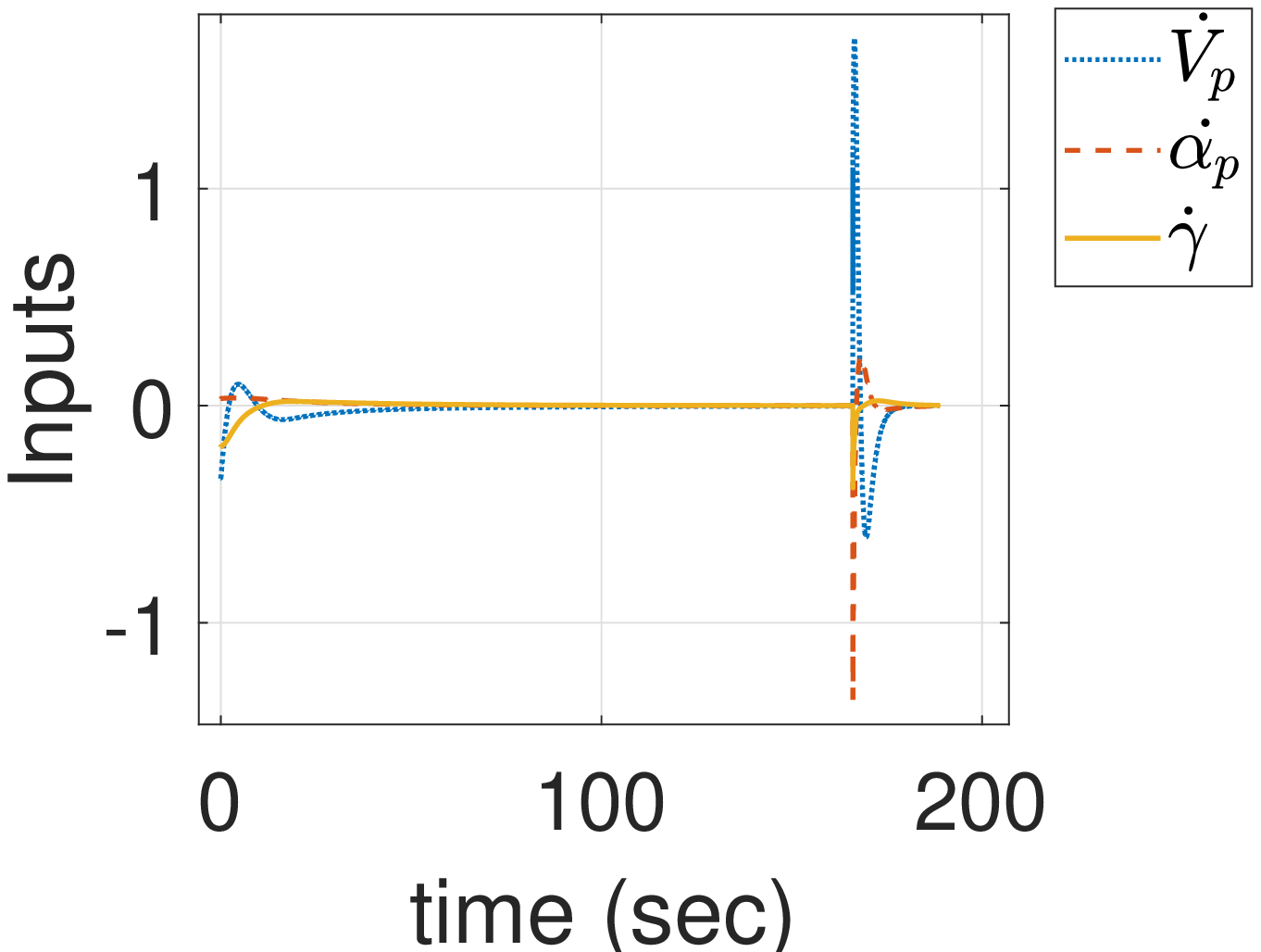}}%
\hfill 
\subcaptionbox{Constant Maneuvering Target\label{fig:2circularinput}}{\includegraphics[width=.2\textwidth,height=9.5cm,keepaspectratio,trim={1cm 0.3cm 0.8cm .08cm}]{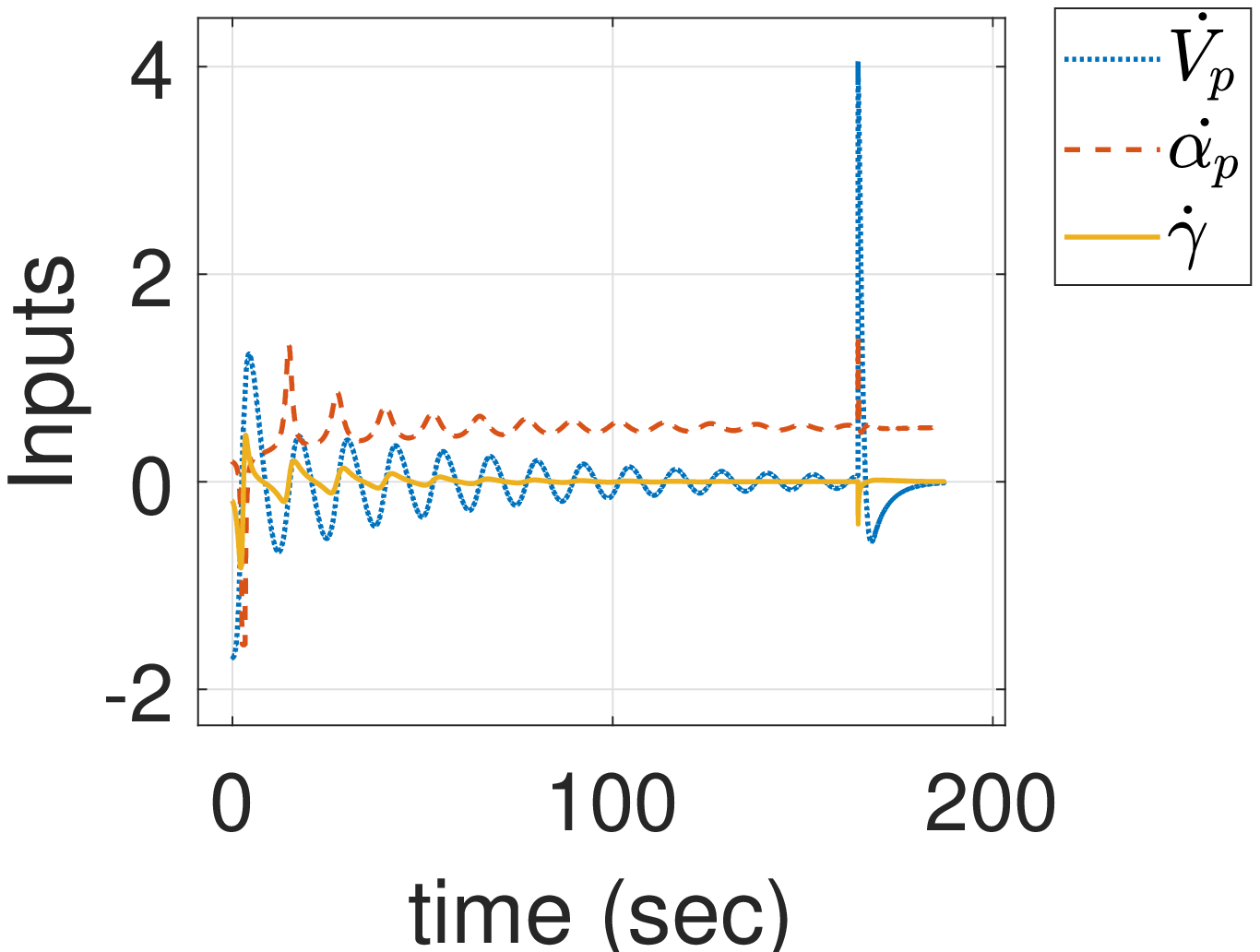}}%
\hfill
\subcaptionbox{Sinusoidally Maneuvering Target\label{fig:2sinusoidalinput}}{\includegraphics[width=.2\textwidth,height=9.5cm,keepaspectratio,trim={1cm 0.3cm 0.8cm .08cm}]{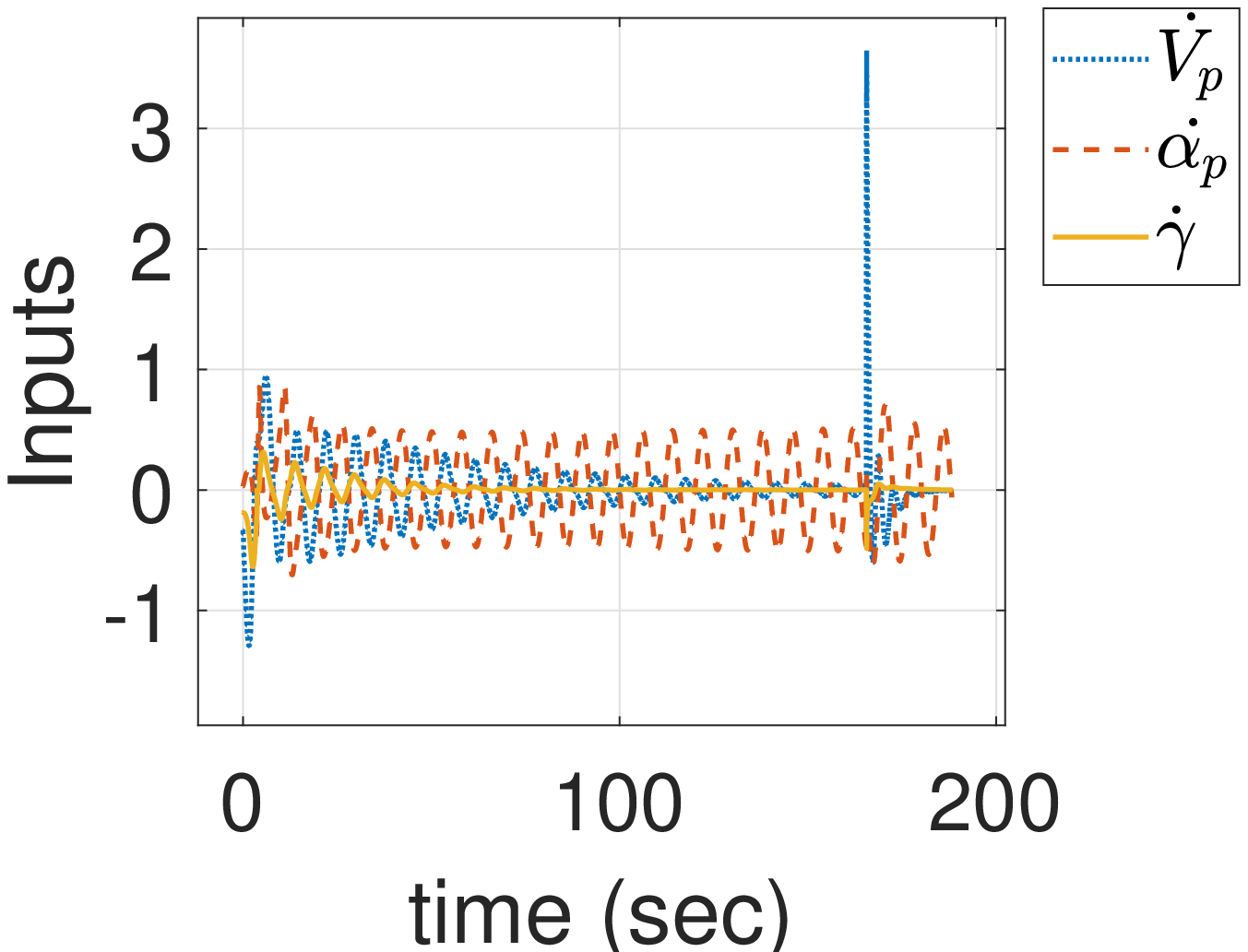}}%
\caption{Guidance Commands for UAV}
\label{fig:2inputs}
\end{figure*}

\begin{figure*}[b!]
\centering
\subcaptionbox{Stationary Target\label{fig:2stationaryheadspeed}}{\includegraphics[width=.2\textwidth,height=9.5cm,keepaspectratio,trim={1cm 0.3cm 0.0cm .08cm}]{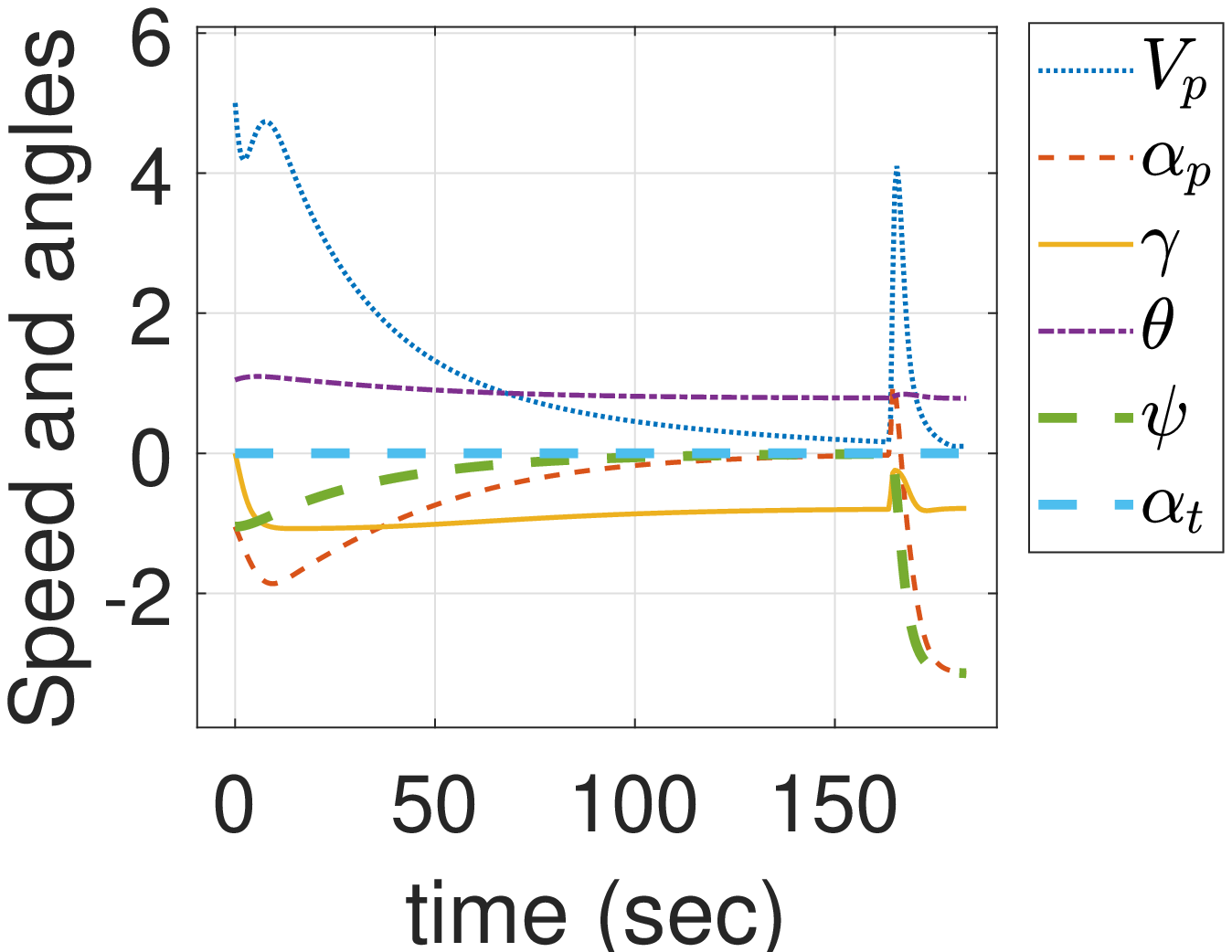}}%
\hfill 
\subcaptionbox{Non-Maneuvering Target\label{fig:2slineheadspeed}}{\includegraphics[width=.2\textwidth,height=9.5cm,keepaspectratio,trim={1cm 0.3cm 0.0cm .08cm}]{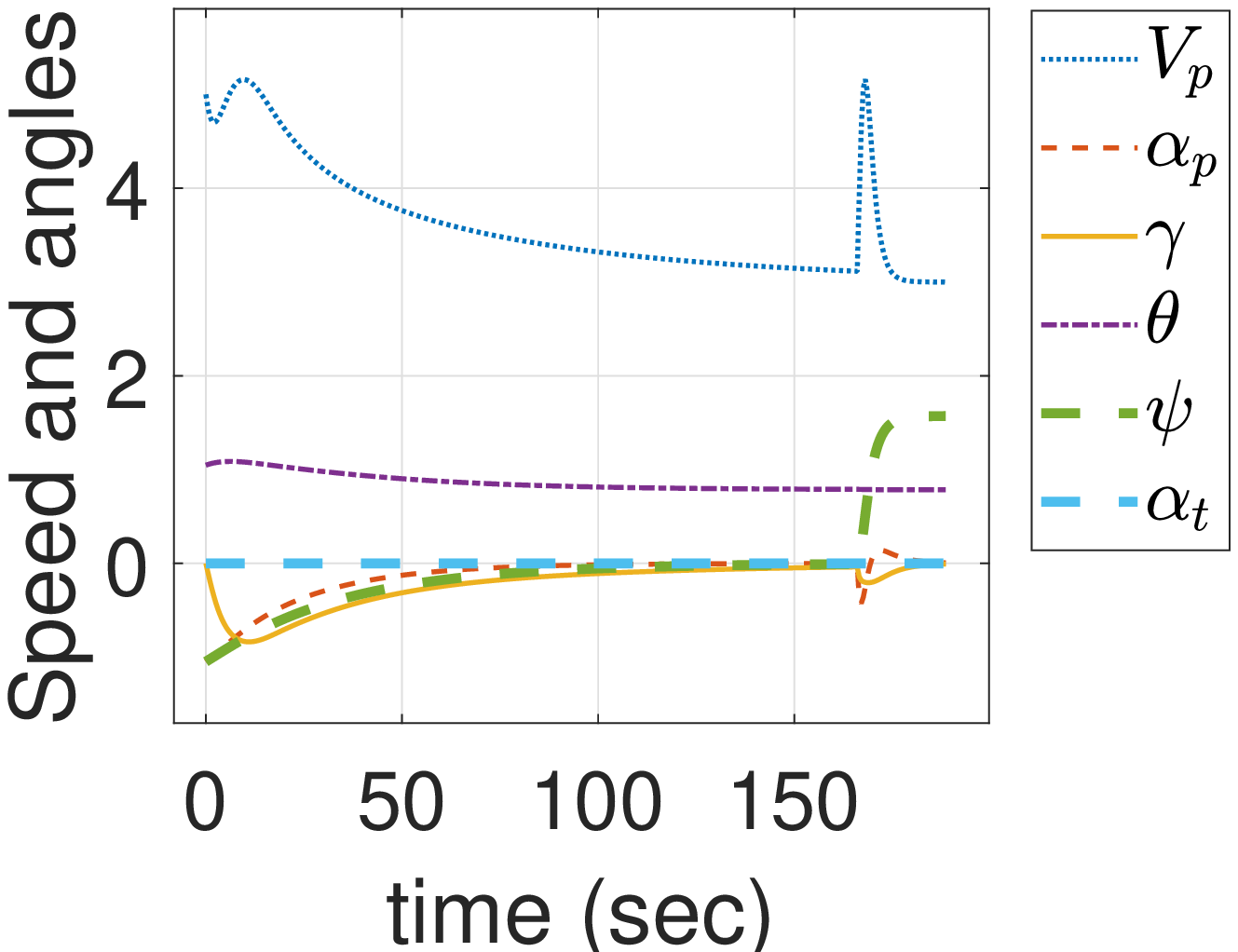}}%
\hfill 
\subcaptionbox{Constant Maneuvering Target\label{fig:2circularheadspeed}}{\includegraphics[width=.2\textwidth,height=9.5cm,keepaspectratio,trim={1cm 0.3cm 0.0cm .08cm}]{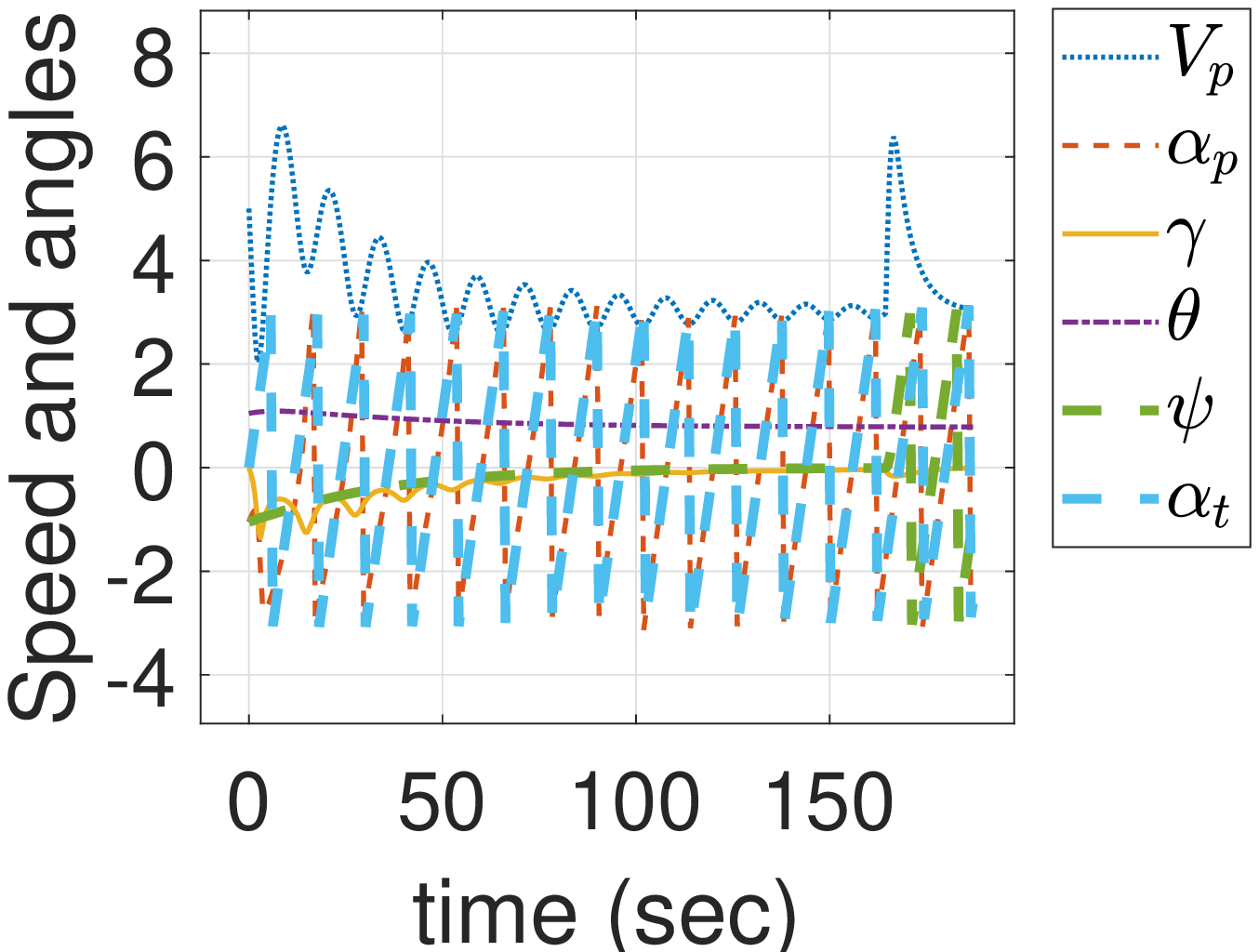}}%
\hfill 
\subcaptionbox{Sinusoidally Maneuvering Target\label{fig:2sinusoidalheadspeed}}{\includegraphics[width=.2\textwidth,height=9.5cm,keepaspectratio,trim={1cm 0.3cm 0.0cm .08cm}]{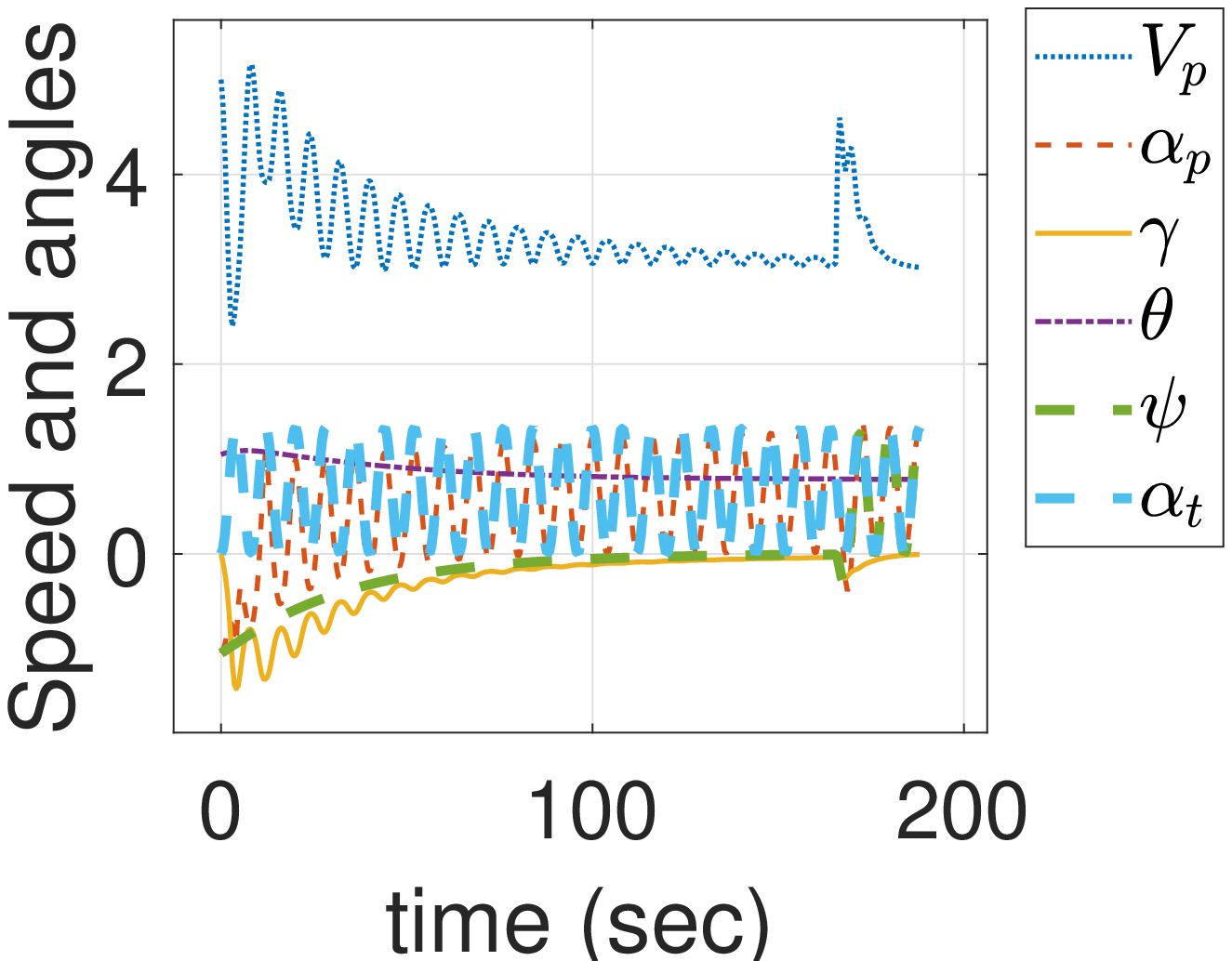}}%
\caption{UAV Speed and angles}
\label{fig:2headspeeds}
\end{figure*}

\subsection{Simulations}\label{subsec:2phasesims}
MATLAB simulations are presented in this section for the two-phase guidance scheme presented above in Section \ref{subsec:2phasesynthesis} assuming point mass models for UAV and UGV, similar to that in  Section \ref{sec:results}. The same four cases of interest as earlier, namely, stationary, nonmaneuvering, constant-maneuvering, and sinusoidally maneuvering targets are considered. The initial conditions of these simulations for each case are set to be the same as the respective case in Section \ref{sec:results}.
The parameters $k_a, k_b, k_c, k_1,$ and $k_2$ in the first phase, and the parameters $m,n$ are are also set to be the same as those in Section \ref{sec:results}. However, $k_3$ is different in the first phase itself. And, in the second phase these parameters, except $m$ and $n$, change w.r.t. first phase. The guidance parameters used in the two-phase guidance simulations are presented elaborately in Table \ref{tab:2params}. 
For the purpose of simulation, $R_{switch}$ is set as 7.5 m, and in the first phase, the desired azimuth angle ($\xi$) is set at zero rad. The heading rate of the target ($\dot{\alpha_t}$) and the desired approach angles ($\zeta_{des}, \theta_{des}$) for each case are same as those considered in Section \ref{sec:results}.
The results of simulations for the two-phase guidance scheme are given in Figs. \ref{fig:2trajs} to \ref{fig:2slideswithtime}.

\begin{center}
\begin{table}[h!]
\centering
\scalebox{0.9}[0.8]{
\begin{tabular}{
|p{0.08\textwidth}|p{0.02\textwidth}|p{0.06\textwidth}|p{0.035\textwidth}|p{0.035\textwidth}|p{0.035\textwidth}|p{0.035\textwidth}|p{0.035\textwidth}|}
\hline
 Target type & Pha- se & $k_a$ $=$ $1.5/R_{xy0}$ & $k_b$ $=$ $3k_a$ & $k_c$ $=$ $2k_a$ & $k_1$ & $k_2$ & $k_3$\\ 
\hline
Stationary &  1 & 0.0150 & 0.0450 & 0.0300 & 0.1395 & 0.1784 & 0.0363 \\
& 2 & 0.2001 & 0.6002 & 0.4001 & 0.3543 & 0.1256 & 0.3442\\
\hline
Non- & 1 & 0.0150 & 0.0450 & 0.0300 & 0.0914 & 0.1297 & 0.0169\\
maneuvering & 2 & 0.2000 & 0.6001 & 0.4001 & 0.3542 & 0.1231 & 0.2828\\ 
\hline
Constant- & 1 & 0.0150 & 0.0450 & 0.0300 & 0.0914 & 0.1297 & 0.0169\\
maneuvering & 2 & 0.2000 & 0.6000 & 0.4000 & 0.3542 & 0.1263 & 0.2564\\
\hline
Sinusoidally & 1 & 0.0150 & 0.0450 & 0.0300 & 0.0914 & 0.1297 & 0.0169\\
maneuvering & 2 & 0.2000 & 0.6001 & 0.4000 & 0.3542 & 0.1242 & 0.1869\\
\hline
\end{tabular}}
\caption{Guidance Parameters in two-phase scheme}
\label{tab:2params}
\end{table}
\end{center}

\begin{figure*}[t!]
\centering
\subcaptionbox{Stationary Target\label{fig:2stationarys}}{\includegraphics[width=.2\textwidth,height=9.5cm,keepaspectratio,trim={1cm 0.3cm 0.7cm .08cm}]{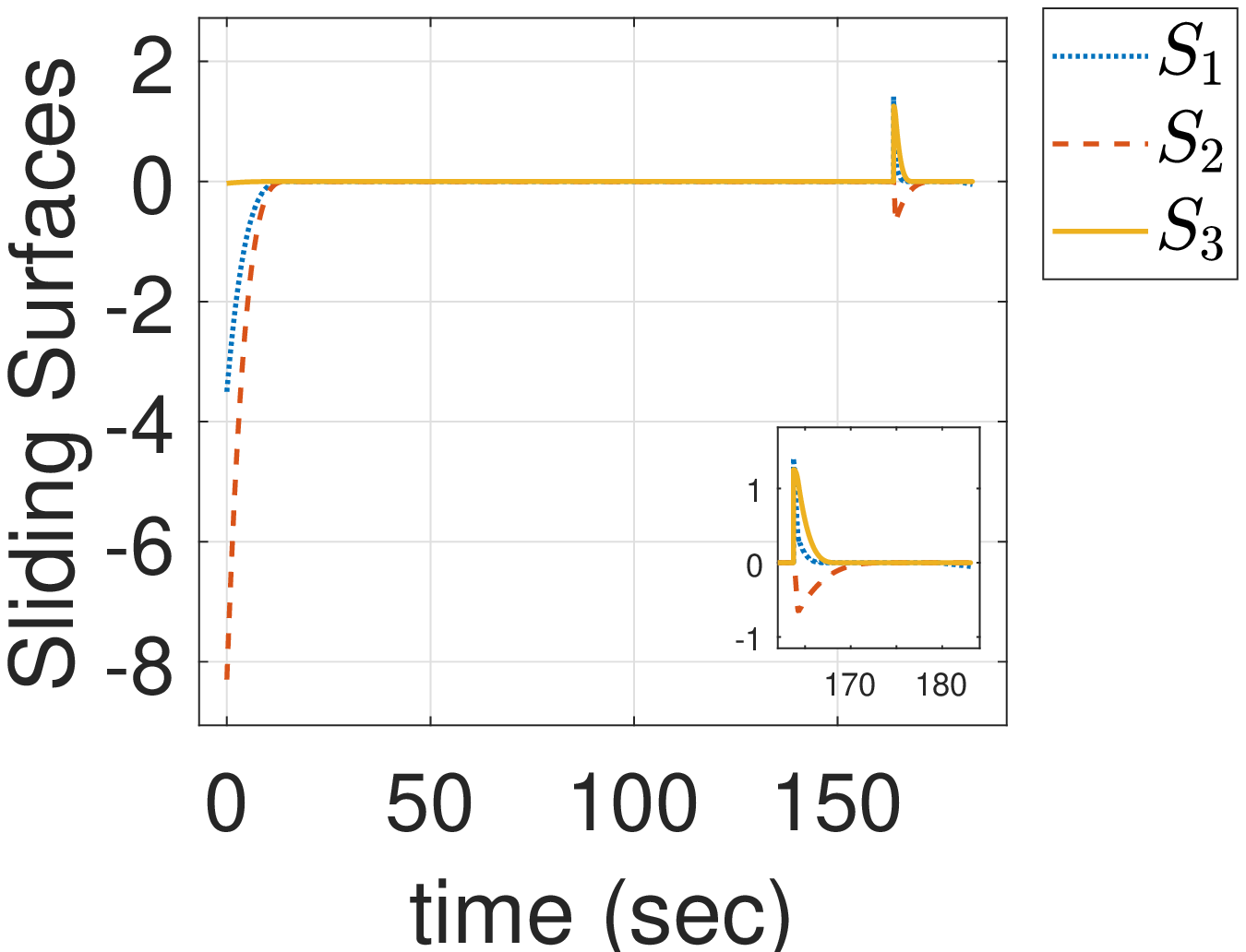}}%
\hfill 
\subcaptionbox{Non-Maneuvering Target\label{fig:2slines}}{\includegraphics[width=.2\textwidth,height=9.5cm,keepaspectratio,trim={1cm 0.3cm 0.7cm .08cm}]{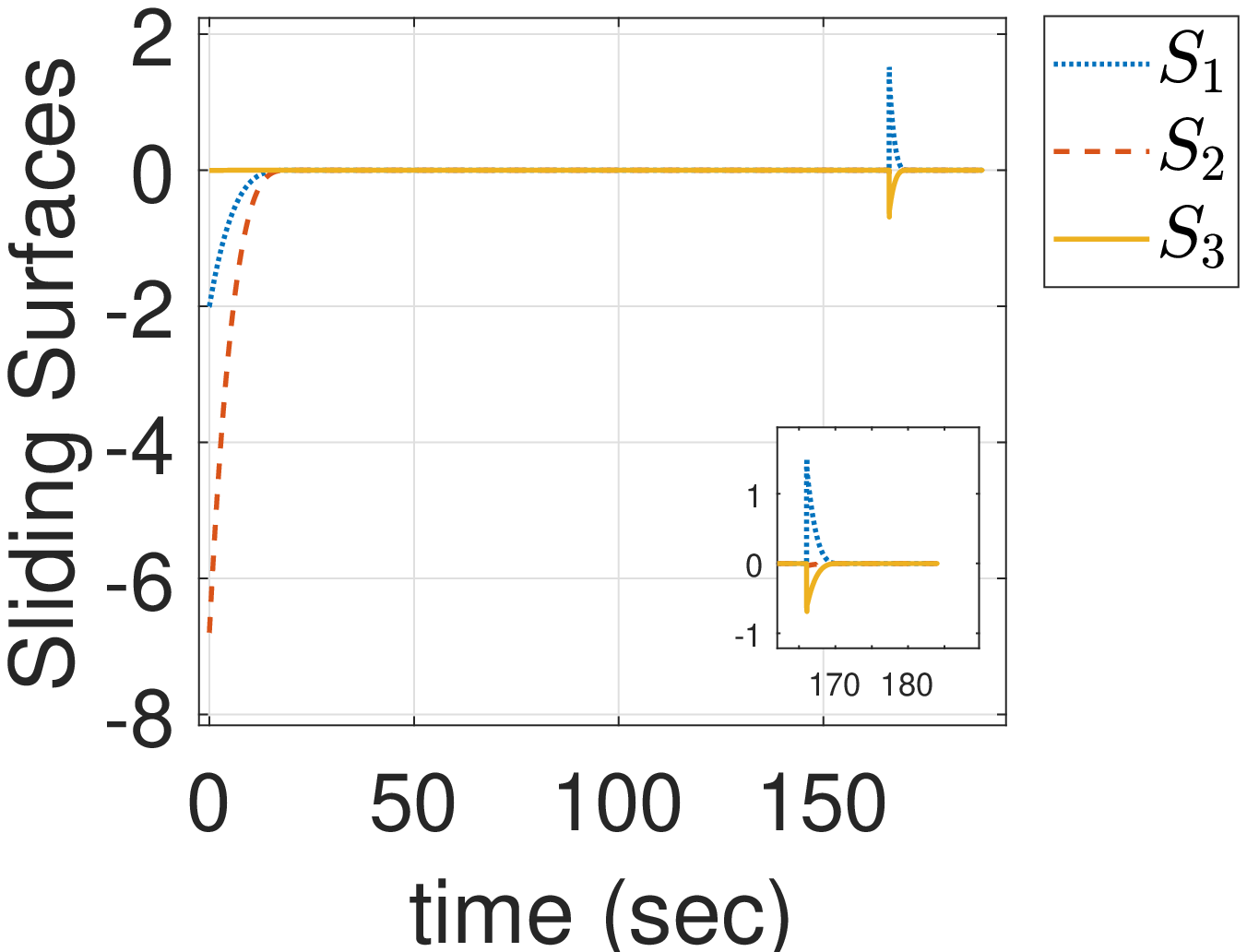}}%
\hfill 
\subcaptionbox{Constant Maneuvering Target\label{fig:2circulars}}{\includegraphics[width=.2\textwidth,height=9.5cm,keepaspectratio,trim={1cm 0.3cm 0.7cm .08cm}]{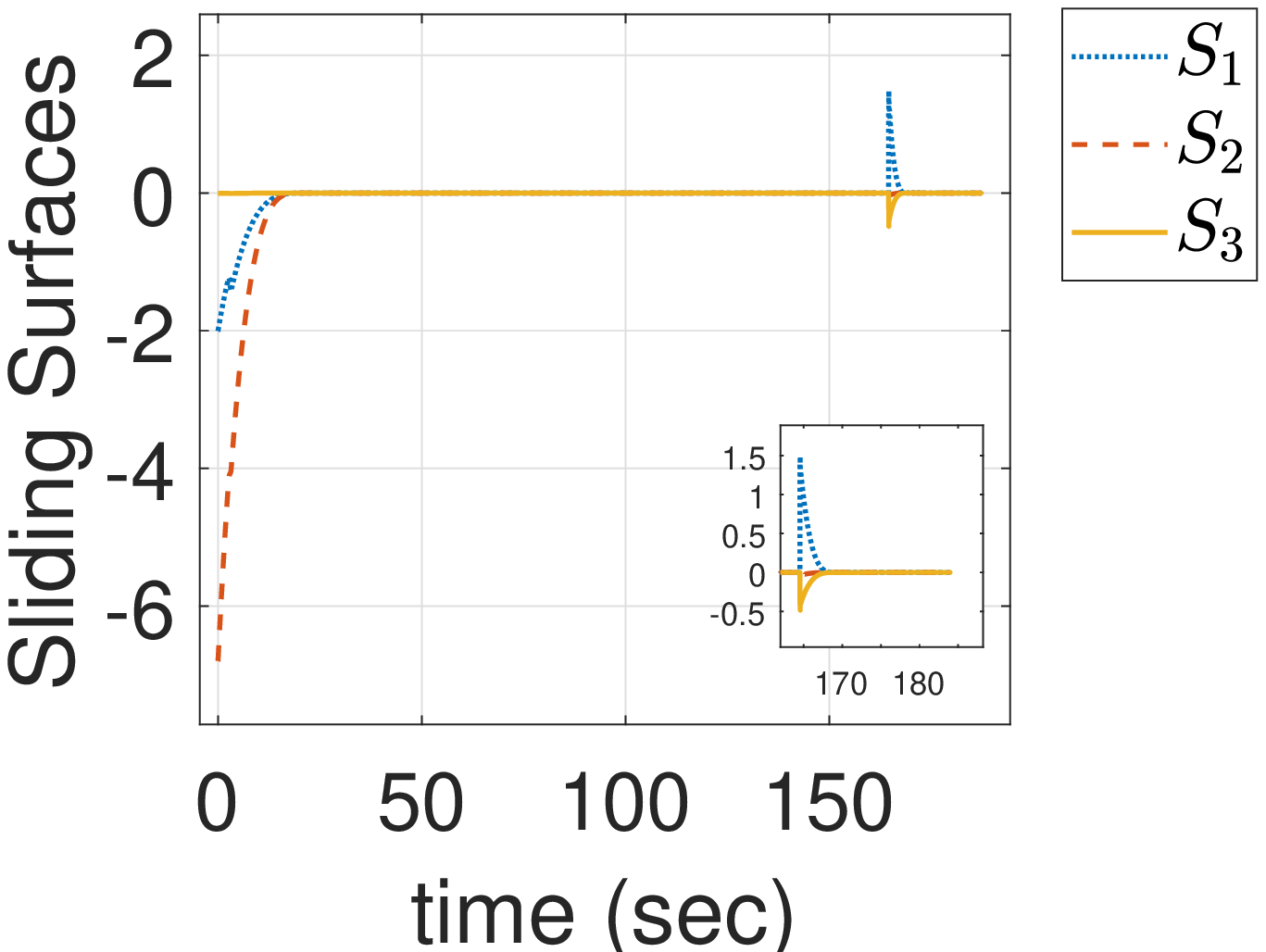}}%
\hfill 
\subcaptionbox{Sinusoidally Maneuvering Target\label{fig:2sinusoidals}}{\includegraphics[width=.2\textwidth,height=9.5cm,keepaspectratio,trim={1cm 0.3cm 0.7cm .08cm}]{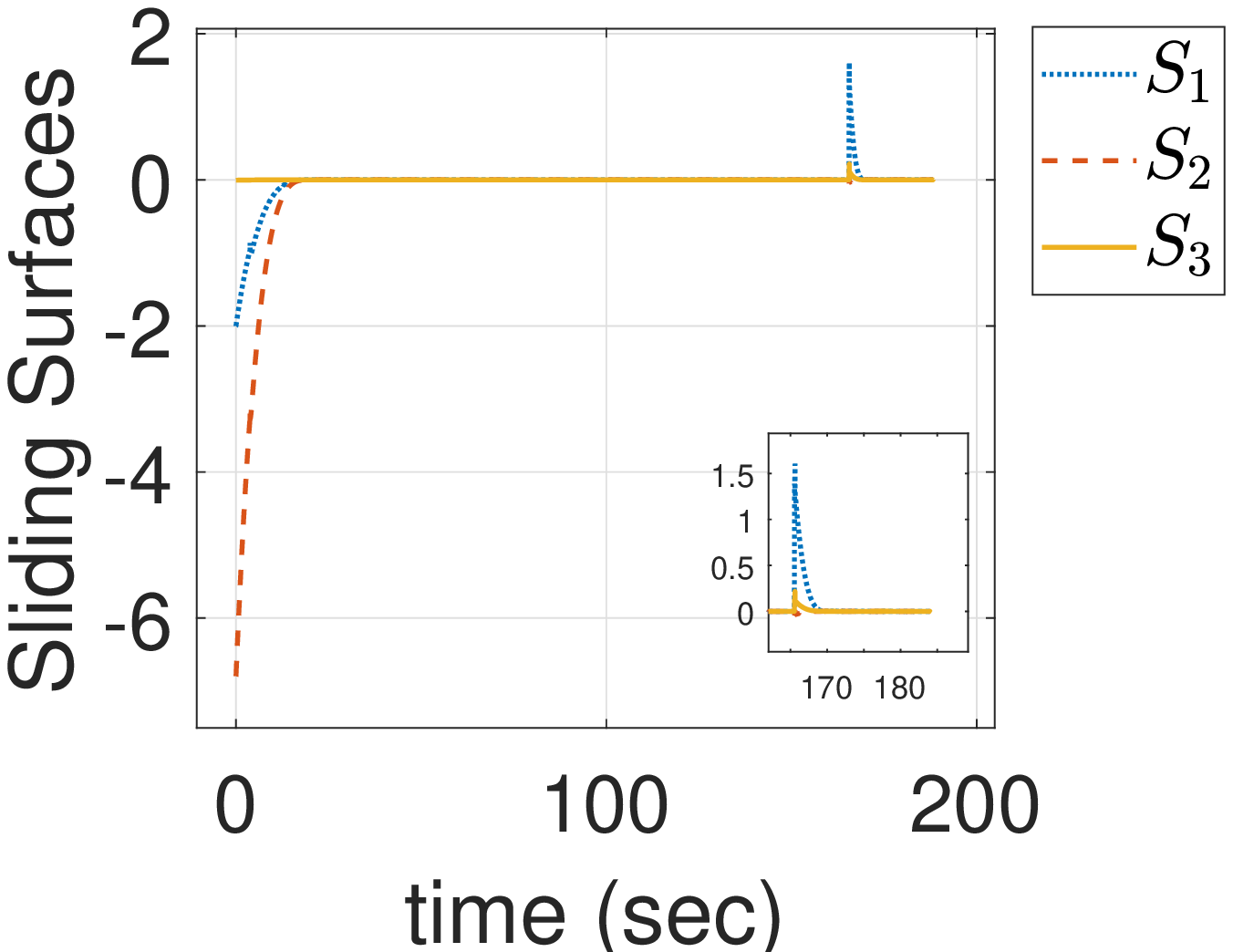}}%
\caption{Sliding Variables}
\label{fig:2slideswithtime}
\end{figure*}


\subsection{Inferences}
From Figs. \ref{fig:2trajs} and \ref{fig:2headspeeds}, it can be seen that as per the designed two-phase guidance scheme above, the UAV first approaches the target at an LOS angle $\psi =0$, and when $R_{xy}<7.5$ m, moves toward the desired azimuth angle relative to the target's heading angle ($\psi_{des} = \alpha_t+\zeta_{des}$) in all four cases presented. Fig. \ref{fig:2inputs} depicts the landing guidance commands of the two-phase scheme presented above. Clearly, in the initial portion of landing also the guidance commands are much smaller in Fig. \ref{fig:2inputs} compared to those in Fig. \ref{fig:inputs} justifying the motivation of the two-phase scheme. As the sliding variable $S_3$ changes (see Eq. \eqref{eq:Sdefn2ii}) in the second phase as $R_{xy}$ becomes equal to $R_{switch}=7.5$m, guidance commands shoot up in a bounded manner at the switching instant and then again decay within a short time. 
Overall, the guidance commands $\dot{V_p}, \dot{\alpha_p}, \dot{\gamma}$ are observed in Fig. \ref{fig:2inputs} to be bounded by $4 m/s^2$, $\pi/2$ rad/s, $\pi/4$ rad/s, which are well within guidance command constraint bounds unlike the inputs depicted in Fig. \ref{fig:inputs} from the one-phase scheme presented in Section \ref{sec:guidancelaw}.
This also results in desired UAV speed to be quite low, less than 7 m/s in all four cases considered as can be noticed in Fig. \ref{fig:2headspeeds}. This also holds true even for large initial distances from the target. Besides these, it can also be observed from Fig. \ref{fig:2dists} that the times taken for landing in all cases by the two-phase guidance scheme are significantly lesser than that for the respective cases by single-phase scheme (shown in Fig. \ref{fig:trajs}) from Section \ref{sec:results}. This is primarily because of re-selection of tuning parameters ($k_a ,k_b ,k_c ,k_1 ,k_2 ,k_3$) in the second phase that increases the values of $|\dot{R_{xy}}|$ and $|\dot{R_z}|$ at each instant as compared to the one-phase scheme in Section \ref{sec:results}. Thus, it can be inferred that the two-phase algorithm proposed in this section performs better than the one-phase guidance scheme proposed in Section \ref{sec:guidancelaw} in several aspects, which is more pertinent for targets with high maneuvering capabilities.

\section{Software-In-The-Loop Simulations}\label{sec:SITL}
Motivated by the satisfactory performance of the two-phase guidance scheme on point mass model of UAV and UGV (Target), in this section, more realistic testbeds are considered for software-in-the-loop (SITL) simulation. It involves a set-up containing physics engines related to the IRIS quadcopter flight dynamics and the Husky ground rover vehicle dynamics models \cite{irishusky}, which are pre-included in the Gazebo simulator software, the Ardupilot desktop-executable autopilot software, and the robot operating system (ROS). Here, the guidance commands for the UAV ($\dot{V_p}, \dot{\alpha_p}, \dot{\gamma}$) are integrated to obtain velocity commands ($V_p , \alpha_p,$ and $\gamma$) at each instant, which are then fed to the UAV's ardupilot at a frequency of 60 hz. The UAV's autopilot in turn tracks these velocity commands in the SITL simulation. The simulations are terminated when the UAV reaches a threshold range of $R_{xy} = 0.3$ m and $R_{z}= 0.3$ m from the target.

Also, recall that in section \ref{subsec:2phasesims}, the guidance parameters were tuned such that high UAV speeds were avoided in the initial phase, when $R_{xy}$ was large. However, as $R_{xy}$ decreases, $|\dot{R_{xy}}|$ also decreases significantly as it varies linearly with $R_{xy}$ on the sliding surface. Thus, in order to maintain a consistent speed range and avoid very high or very low speeds for the UAV for most of the mission time, a multi-stage guidance scheme is considered, in which the parameters of the basic sliding mode-based guidance scheme (from Section \ref{sec:twophase}) are re-tuned every time $R_{xy}$ becomes half of its initial value or increases by 5 m from its initial value in either phase. Here, it should be noted that though the parameters are re-tuned in different stages, the guidance scheme is still termed as two-phased, where the sliding variables represented by $S$ are defined to be the same as before and $S_3$ is switched when $R_{xy} < R_{switch}$ for terminal-azimuth-angle control, in the same way as described in Section  \ref{subsec:2phasesynthesis}.

Apart from the considerations of tuning different guidance parameters given in Section \ref{subsec:ParameterSelection}, for tuning of the parameters $k_a, k_b$ and $k_c$ the followings are additionally considered in the SITL simulation. When the initial errors in LOS angles ($\psi_0 - \psi_{des}$ and $\theta_0-\theta_{des}$) are large, the magnitude of desired $\dot{R_{xy}}$ shouldn't be too high or too low. Subsequently, when $R_{xy}$ is sufficiently small, the rates of convergence of $|\psi - \psi_{des}|$ and $|\theta-\theta_{des}|$ should be higher than those for larger $R_{xy}$ and same $|\psi - \psi_{des}|$ and $|\theta-\theta_{des}|$.

\subsection{Results}
Similar to Sections \ref{sec:results} and \ref{sec:twophase}, four cases of interest are considered, namely, stationary, nonmaneuvering, constant-maneuvering, and sinusoidally maneuvering targets. Here, in all of the phases where $R_{xy} \geq R_{switch}= 7.5$ m, where $S$ is as defined in Eq. \eqref{eq:Sdefn2i}, the desired LOS angle in 2D ($\psi_{des} =\xi$) is set as zero rad. The target is initially placed at a distance ($R_0$) of $20\sqrt{2}$ m from the UAV, with $R_{xy0} = 20$ m and $R_{z0} = -20$ m, along with the initial azimuth angle as $\psi_0 = 0$ rad. The UAV initiates its motion with $V_{p0}=5$ m/s and $\alpha_{p0} = \pi$ rad in all cases. The target's speed is set at $V_{t}= 3$ m/s. 
The initial heading angle, the heading rate of the target and the desired approach angles in all four cases are presented in Table \ref{tab:Ginitials}.

\begin{center}
\begin{table}[h!]
\centering
\begin{tabular}{ |p{0.11\textwidth}|p{0.025\textwidth}|p{0.07\textwidth}|p{0.06\textwidth}|p{0.035\textwidth}|}
\hline
Target type & $\alpha_{t0}$ (rad) & $\dot{\alpha_{t}}(t)$ (rad/s) & $\zeta_{des} =$ \scalebox{.7}[1.0]{($\psi_{des} - \alpha_t$)} (rad) & $\theta_{des}$ (rad)\\ 
\hline
Stationary & 0 & 0 & $\pi$ & $\pi/4$ \\ 
\hline
Nonmaneuvering & $\pi$/4 & 0 & $\pi/2$ & $\pi/4$ \\ 
\hline
Constant-maneuvering & 0 & $\pi/12$ & $\pi/2$ & $\pi/4$ \\ 
\hline
Sinusoidally maneveuvering & 0 & $(\pi/6)$ $\sin(\pi t/4)$ & $-3\pi/4$ & $\pi/4$ \\  
\hline
\end{tabular}
\caption{Cases considered for SITL simulations}
\label{tab:Ginitials}
\end{table}
\end{center}

\begin{figure*}[h!]
\centering
\subcaptionbox{Stationary Target\label{fig:Gstationarytraj}}{\includegraphics[width=0.22\textwidth,height=9.5cm,keepaspectratio,trim={2.0cm 0.3cm 0.0cm .08cm}]{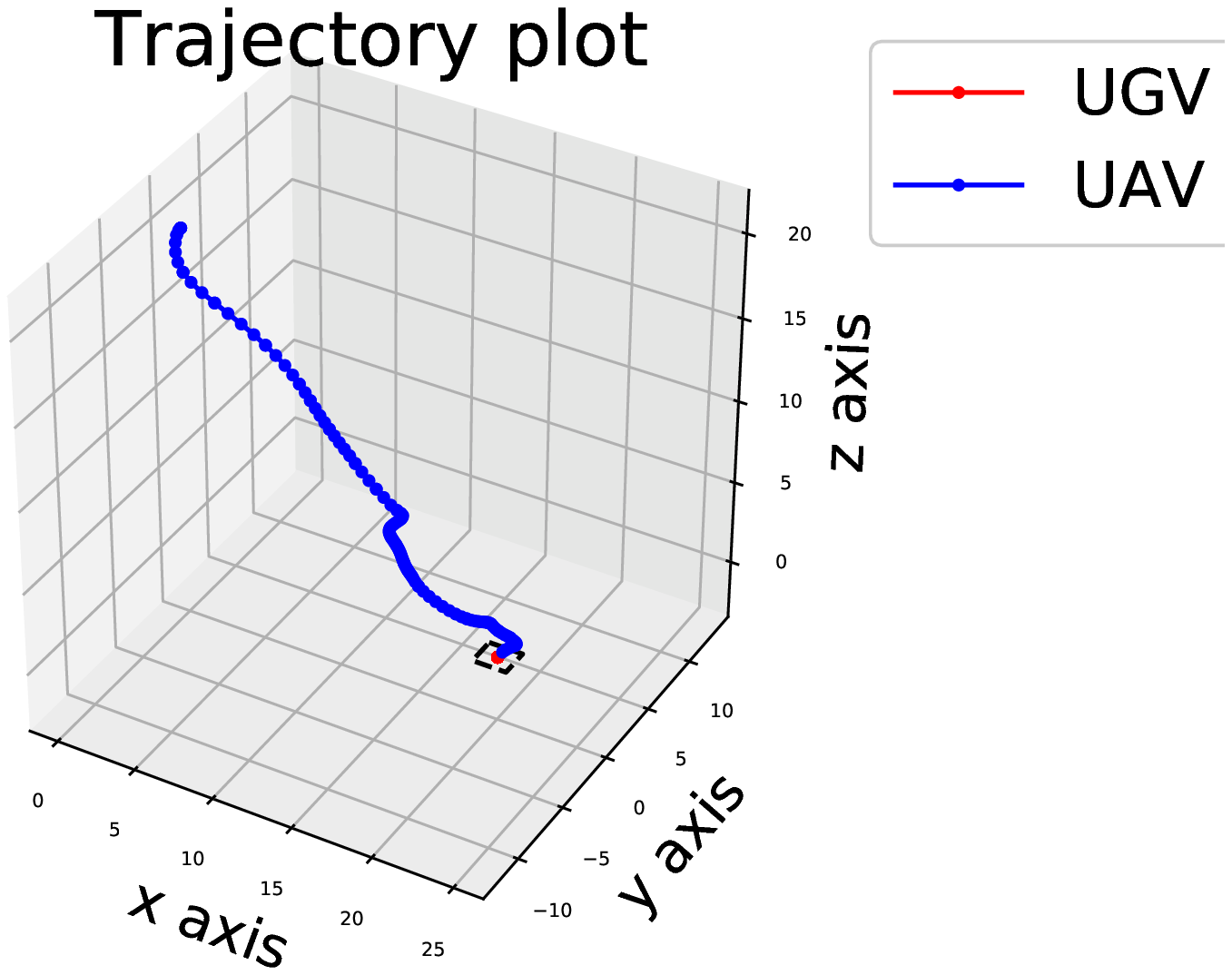}}%
\hfill 
\subcaptionbox{Non-Maneuvering Target\label{fig:Gslinetraj}}{\includegraphics[width=0.23\textwidth,height=9.5cm,keepaspectratio,trim={1.25cm 0.3cm 0.0cm .08cm}]{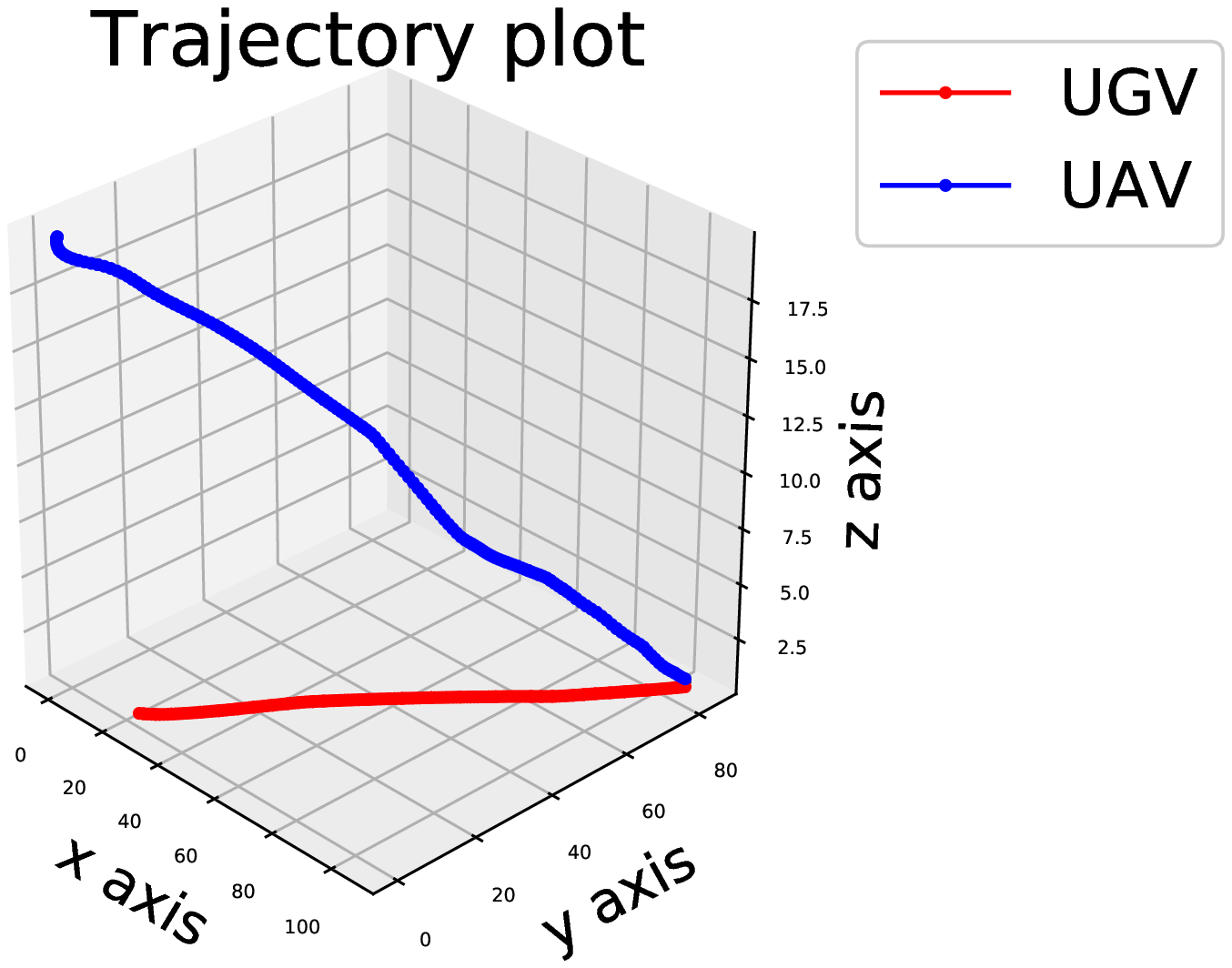}}%
\hfill 
\subcaptionbox{Constant Maneuvering Target\label{fig:Gcirculartraj}}{\includegraphics[width=0.225\textwidth,height=9.5cm,keepaspectratio,trim={0.9cm 0.6cm 0.0cm .08cm}]{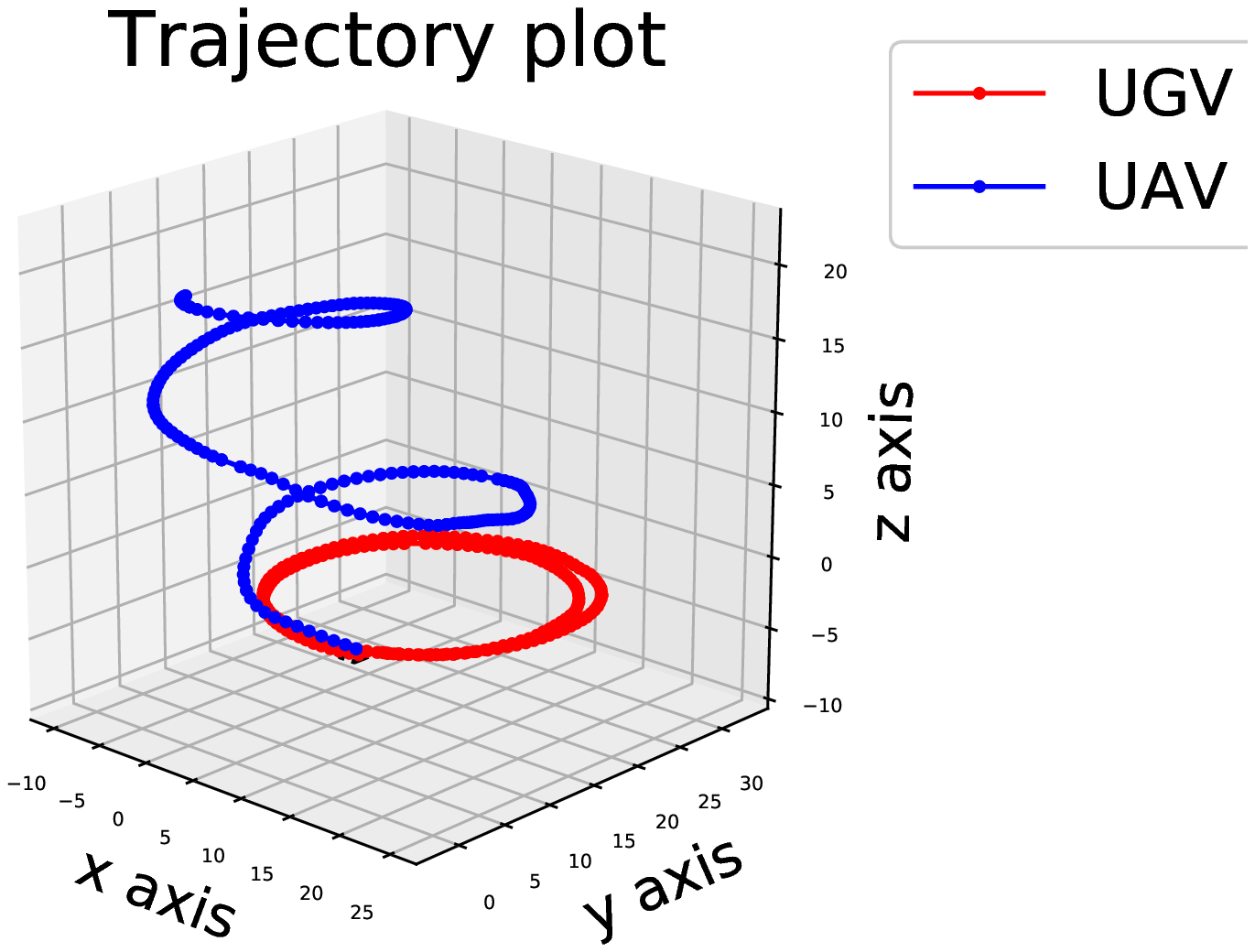}}%
\hfill 
\subcaptionbox{Sinusoidally Maneuvering Target\label{fig:Gsinusoidaltraj}}{\includegraphics[width=0.225\textwidth,height=9.5cm,keepaspectratio,trim={1.5cm 0.3cm 0.0cm .08cm}]{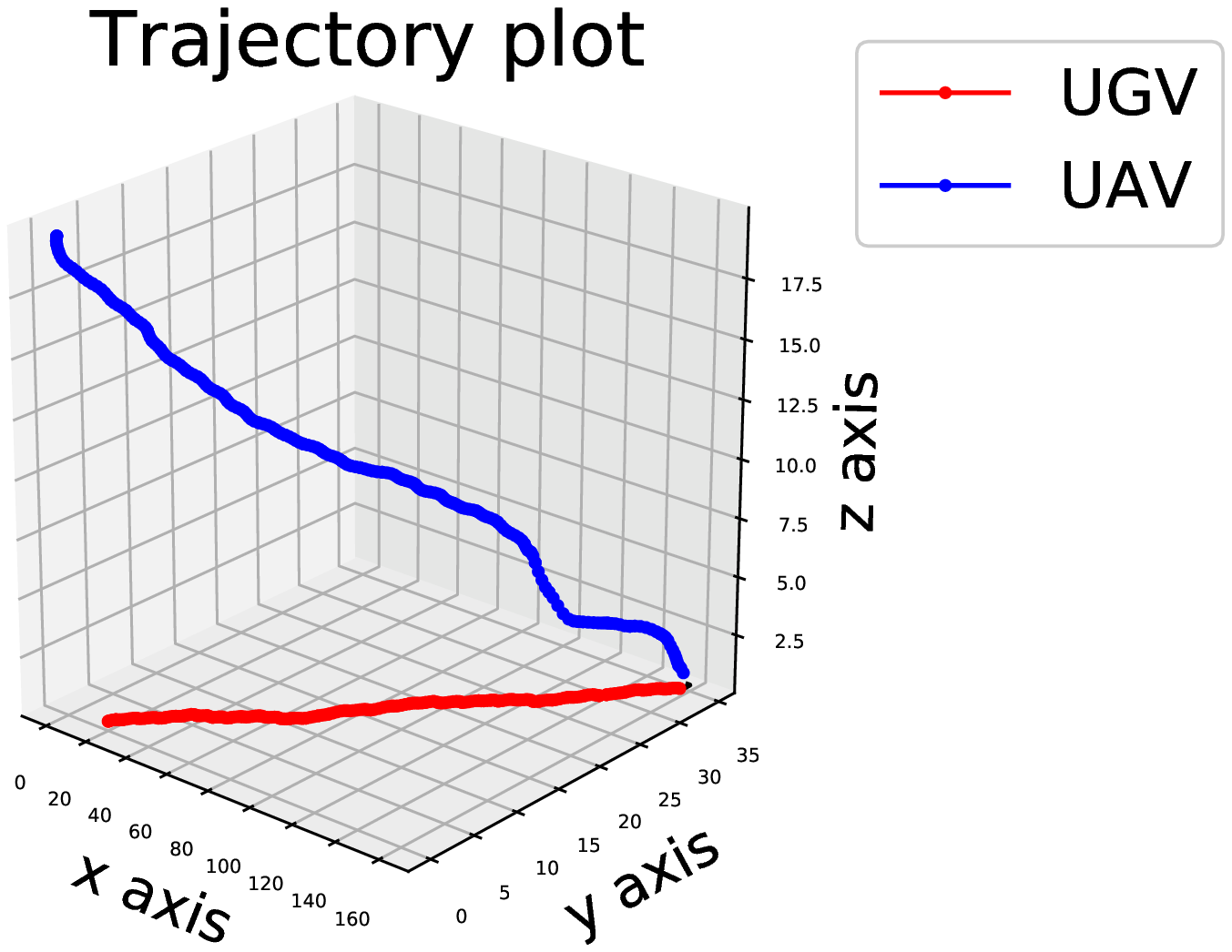}}%
\caption{Trajectory plots for UAV and target}
\label{fig:Gtrajs}
\end{figure*}
\begin{figure*}[h!]
\centering
\subcaptionbox{Stationary Target\label{fig:Gstationarydist}}{\includegraphics[width=.2\textwidth,height=9.5cm,keepaspectratio,trim={0.5cm 0.0cm 0.0cm .08cm}]{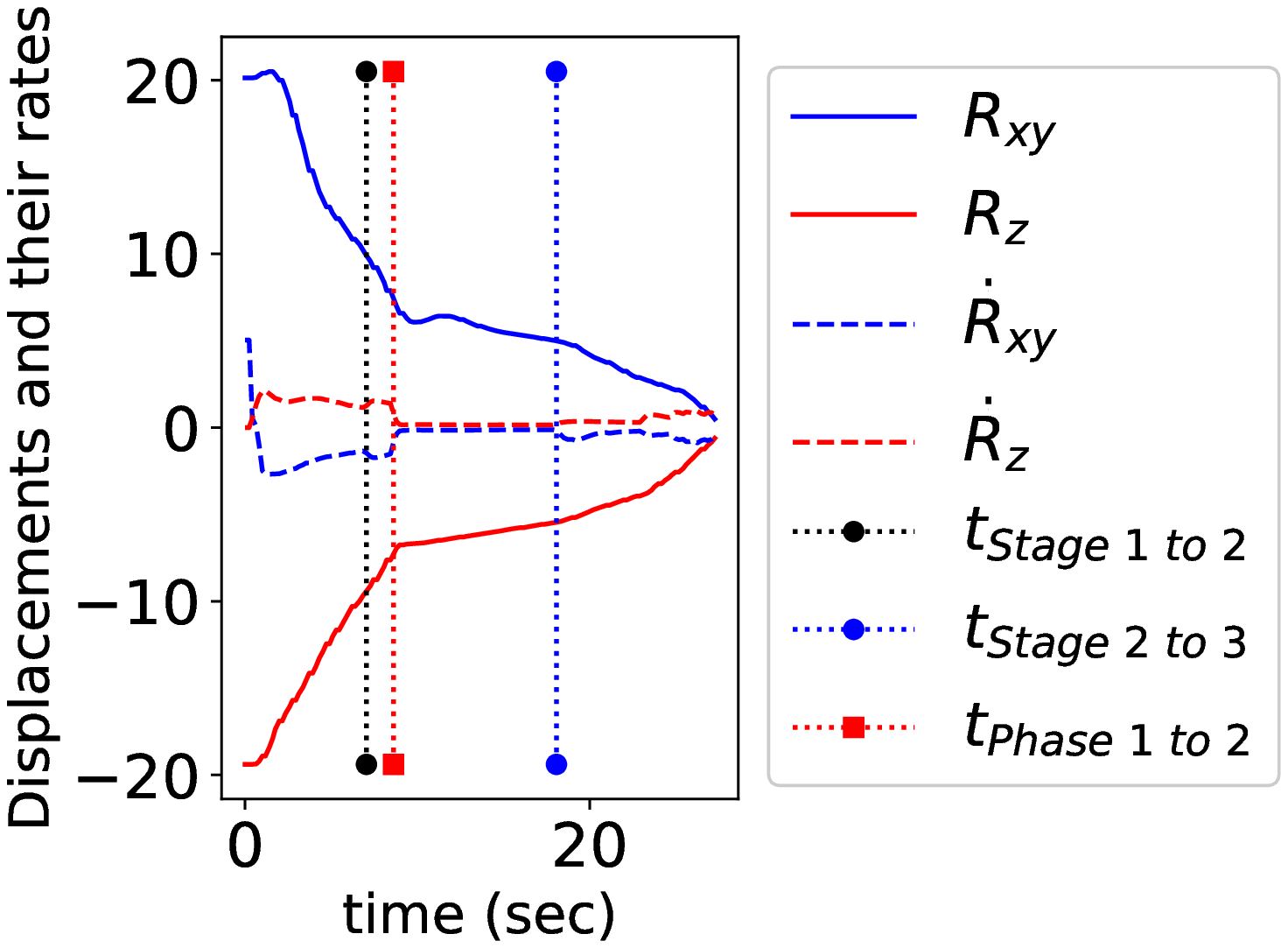}}%
\hfill 
\subcaptionbox{Non-Maneuvering Target\label{fig:Gslinedist}}{\includegraphics[width=.2\textwidth,height=9.5cm,keepaspectratio,trim={0.5cm 0.0cm 0.0cm .08cm}]{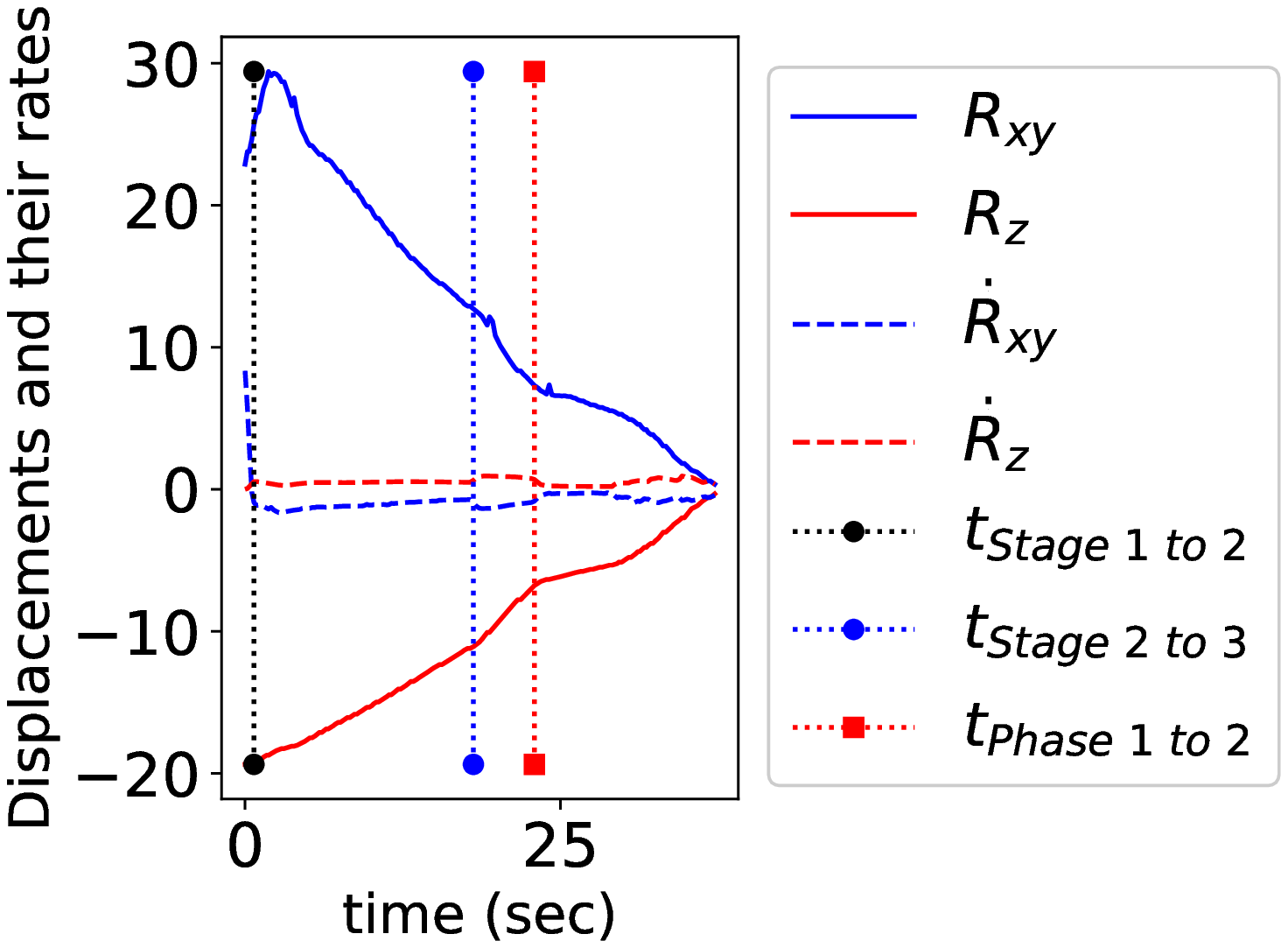}}%
\hfill 
\subcaptionbox{Constant Maneuvering Target\label{fig:Gciruclardist}}{\includegraphics[width=.2\textwidth,height=9.5cm,keepaspectratio,trim={0.5cm 0.0cm 0.0cm .08cm}]{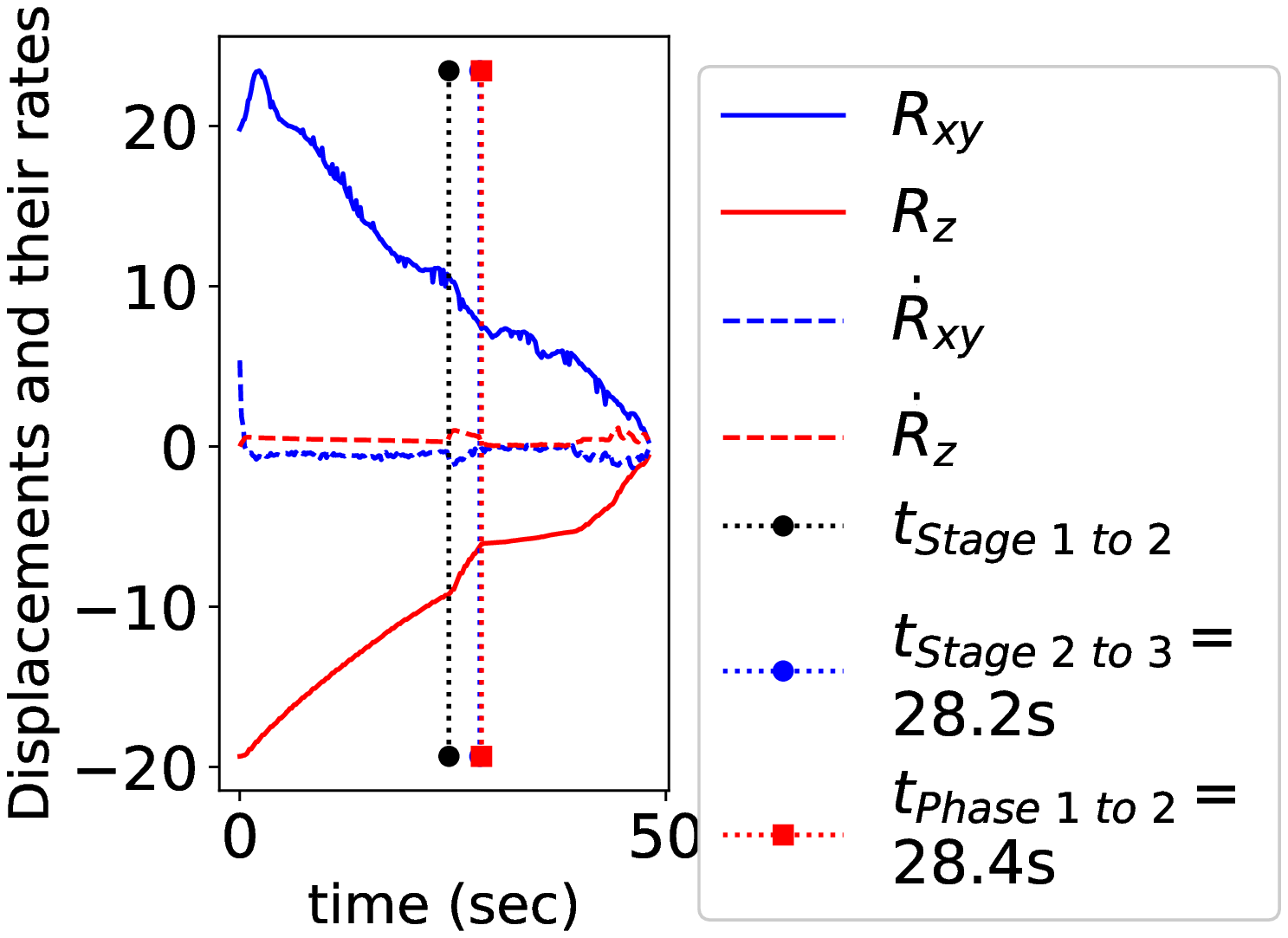}}%
\hfill 
\subcaptionbox{Sinusoidally Maneuvering Target\label{fig:Gsinusoidaldist}}{\includegraphics[width=.2\textwidth,height=9.5cm,keepaspectratio,trim={0.5cm 0.0cm 0.0cm .08cm}]{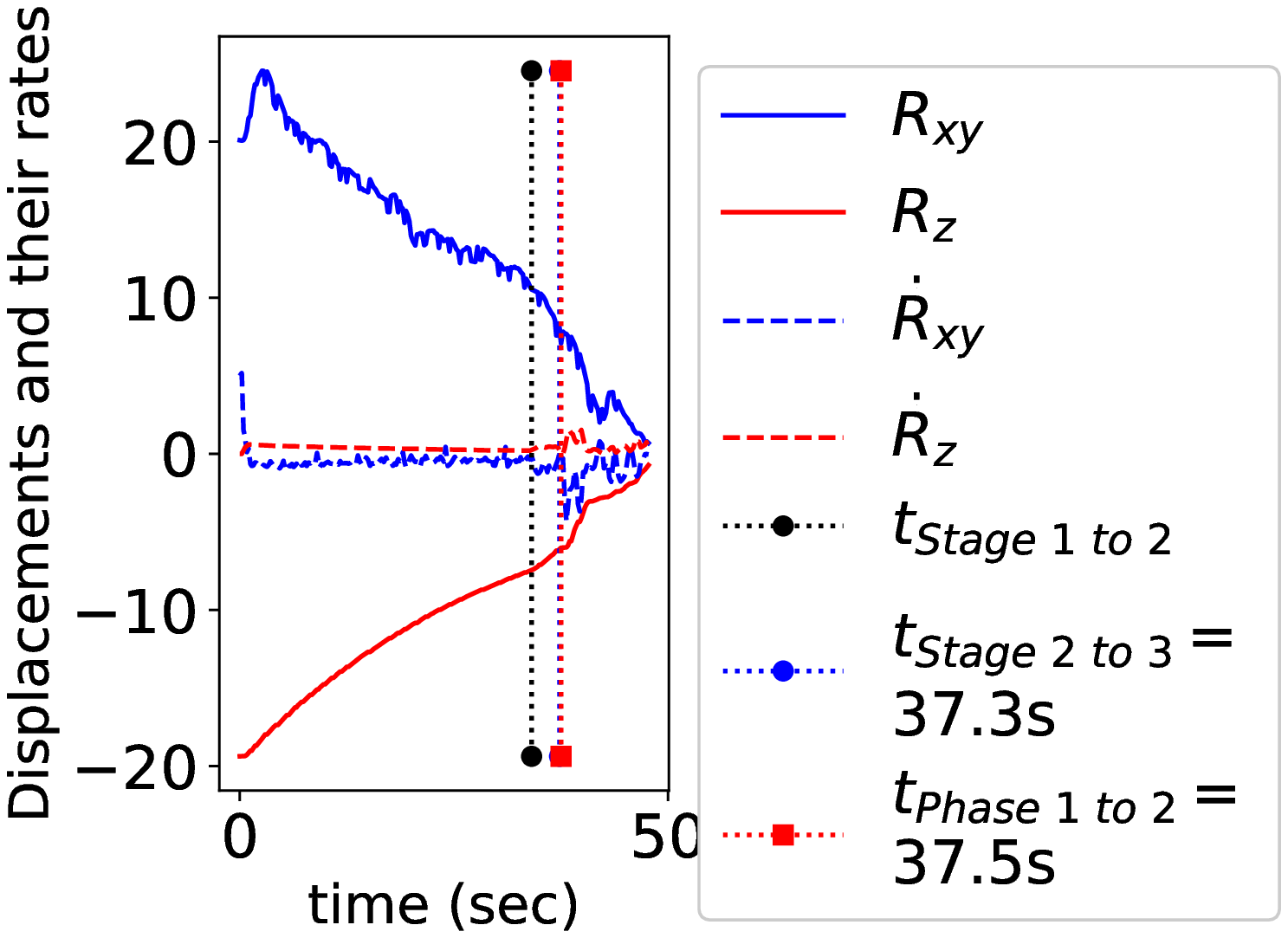}}%
\caption{Distance from target and it's projections on the xy-plane and along the z-axis}
\label{fig:Gdists}
\end{figure*}
\begin{figure*}[h!]
\centering
\subcaptionbox{Stationary Target\label{fig:Gstationaryinput}}{\includegraphics[width=.2\textwidth,height=9.5cm,keepaspectratio,trim={1cm 0.3cm 0.8cm .08cm}]{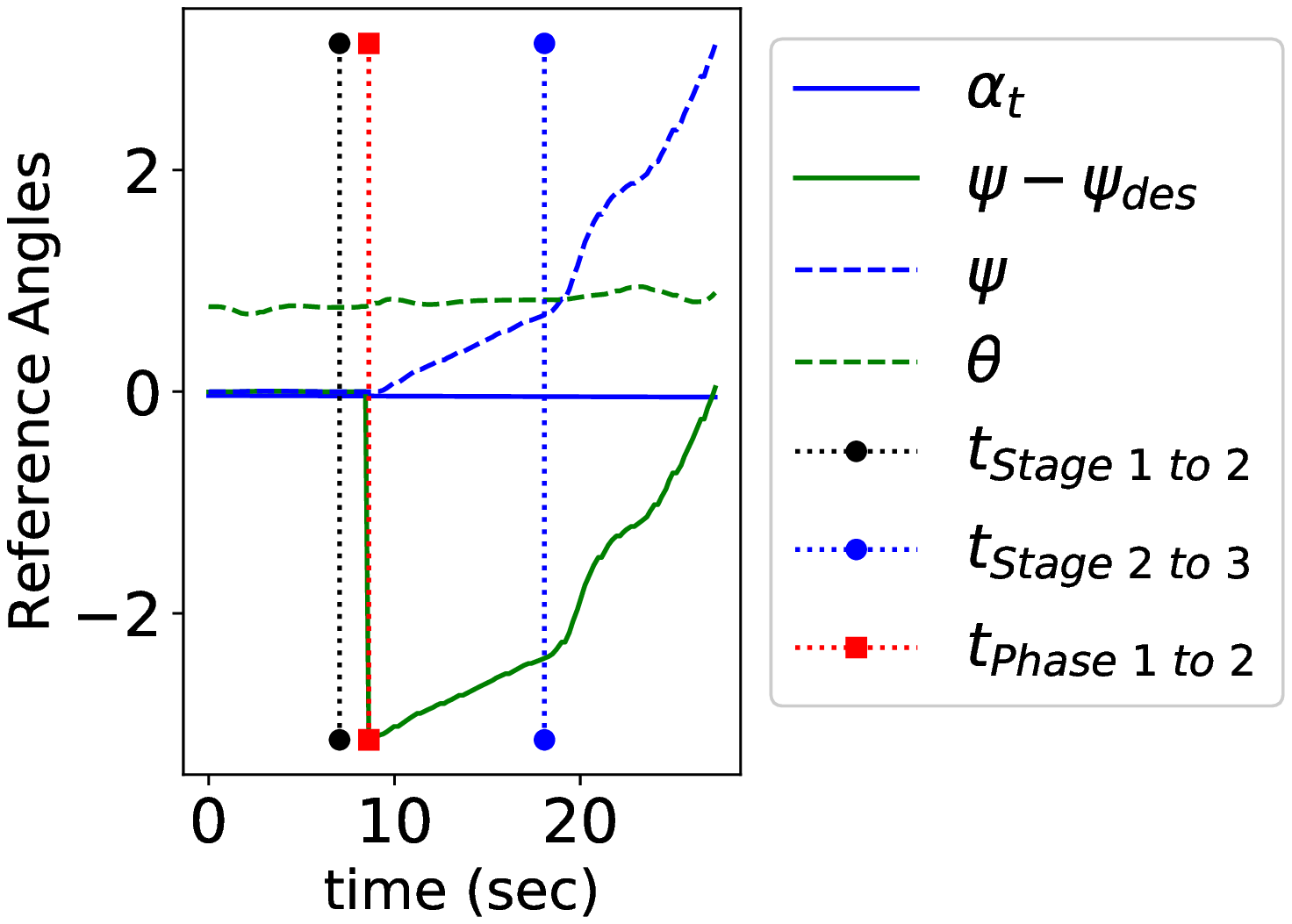}}%
\hfill 
\subcaptionbox{Non-Maneuvering Target\label{fig:Gslineinput}}{\includegraphics[width=.2\textwidth,height=9.5cm,keepaspectratio,trim={1cm 0.3cm 0.8cm .08cm}]{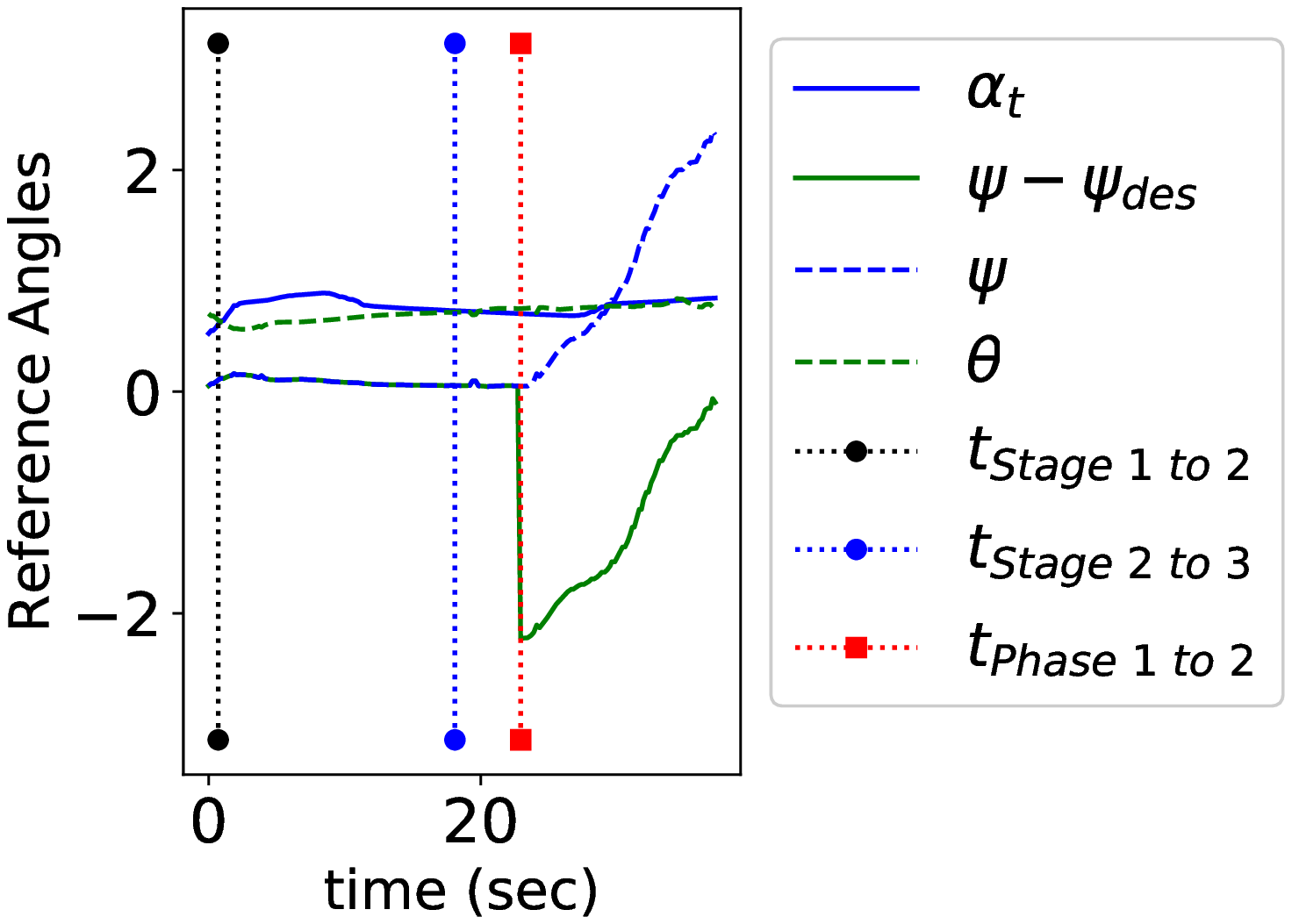}}%
\hfill 
\subcaptionbox{Constant Maneuvering Target\label{fig:Gcircularinput}}{\includegraphics[width=.2\textwidth,height=9.5cm,keepaspectratio,trim={1cm 0.3cm 0.8cm .08cm}]{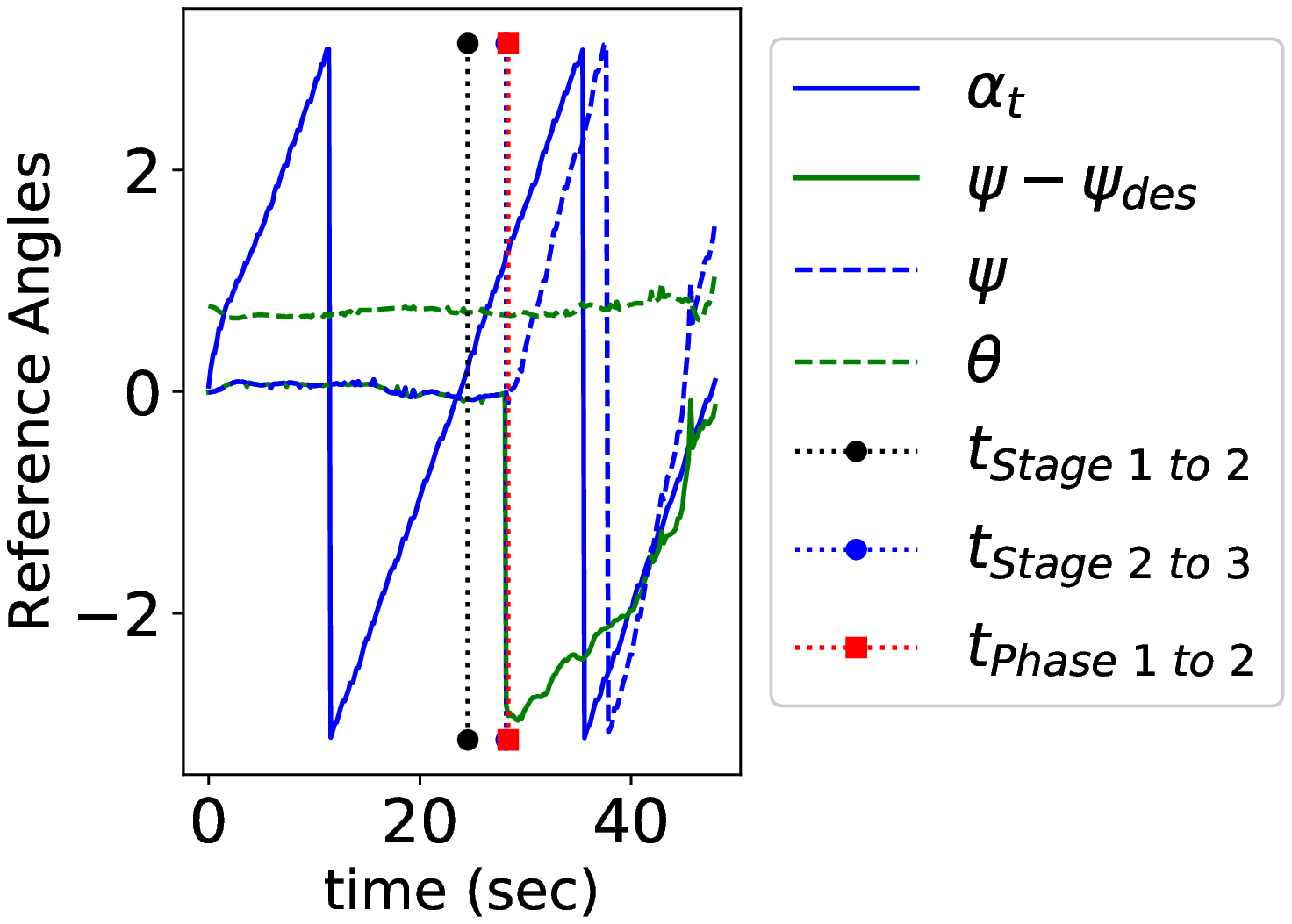}}%
\hfill 
\subcaptionbox{Sinusoidally Maneuvering Target\label{fig:Gsinusoidalinput}}{\includegraphics[width=.2\textwidth,height=9.5cm,keepaspectratio,trim={1cm 0.3cm 0.8cm .08cm}]{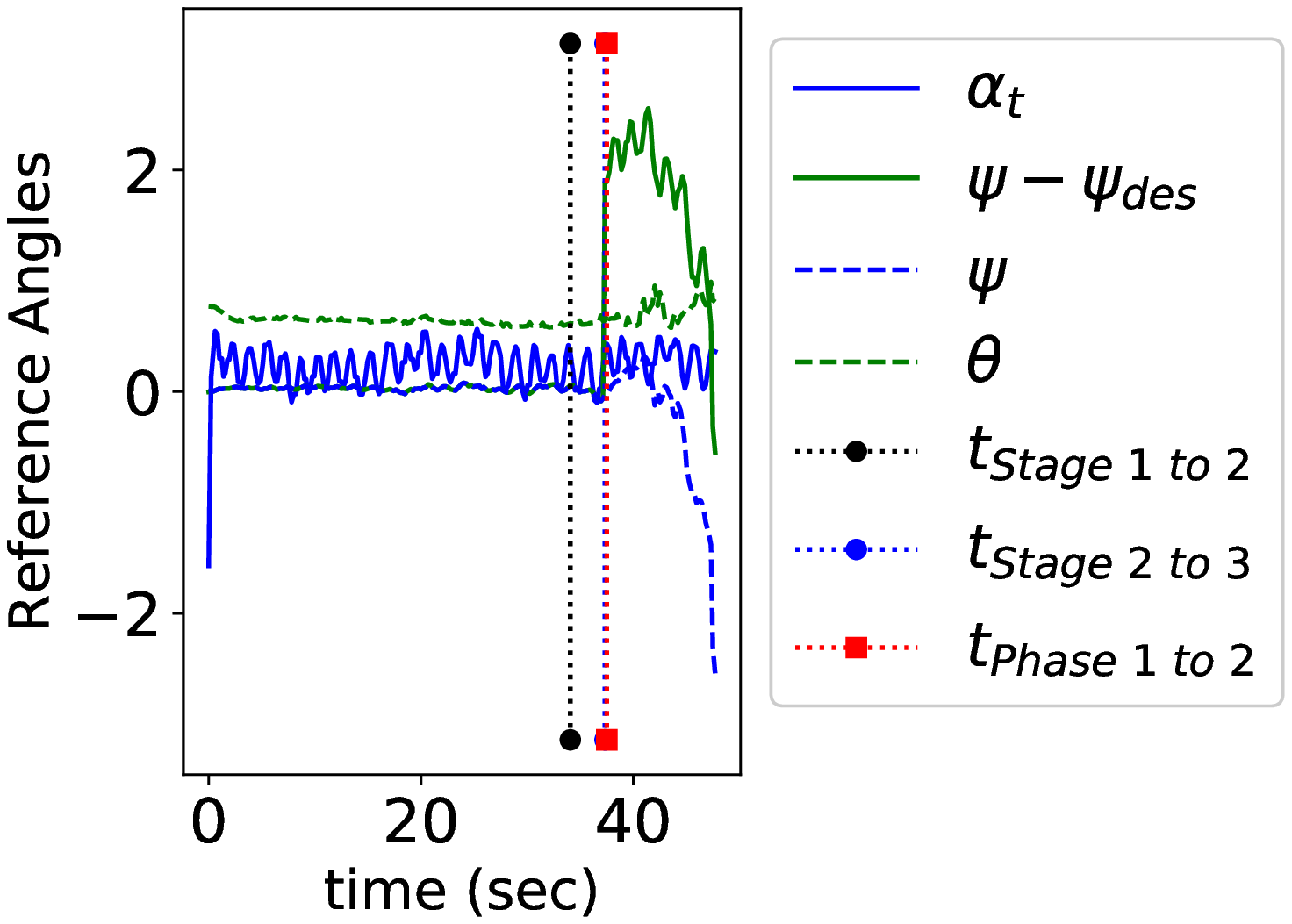}}%
\caption{Reference angles for UAV}
\label{fig:Greferenceangles}
\end{figure*}
\begin{figure*}[h!]
\centering
\subcaptionbox{Stationary Target\label{fig:Gstationaryheadspeed}}{\includegraphics[width=.2\textwidth,height=9.5cm,keepaspectratio,trim={1cm 0.3cm 0.0cm .08cm}]{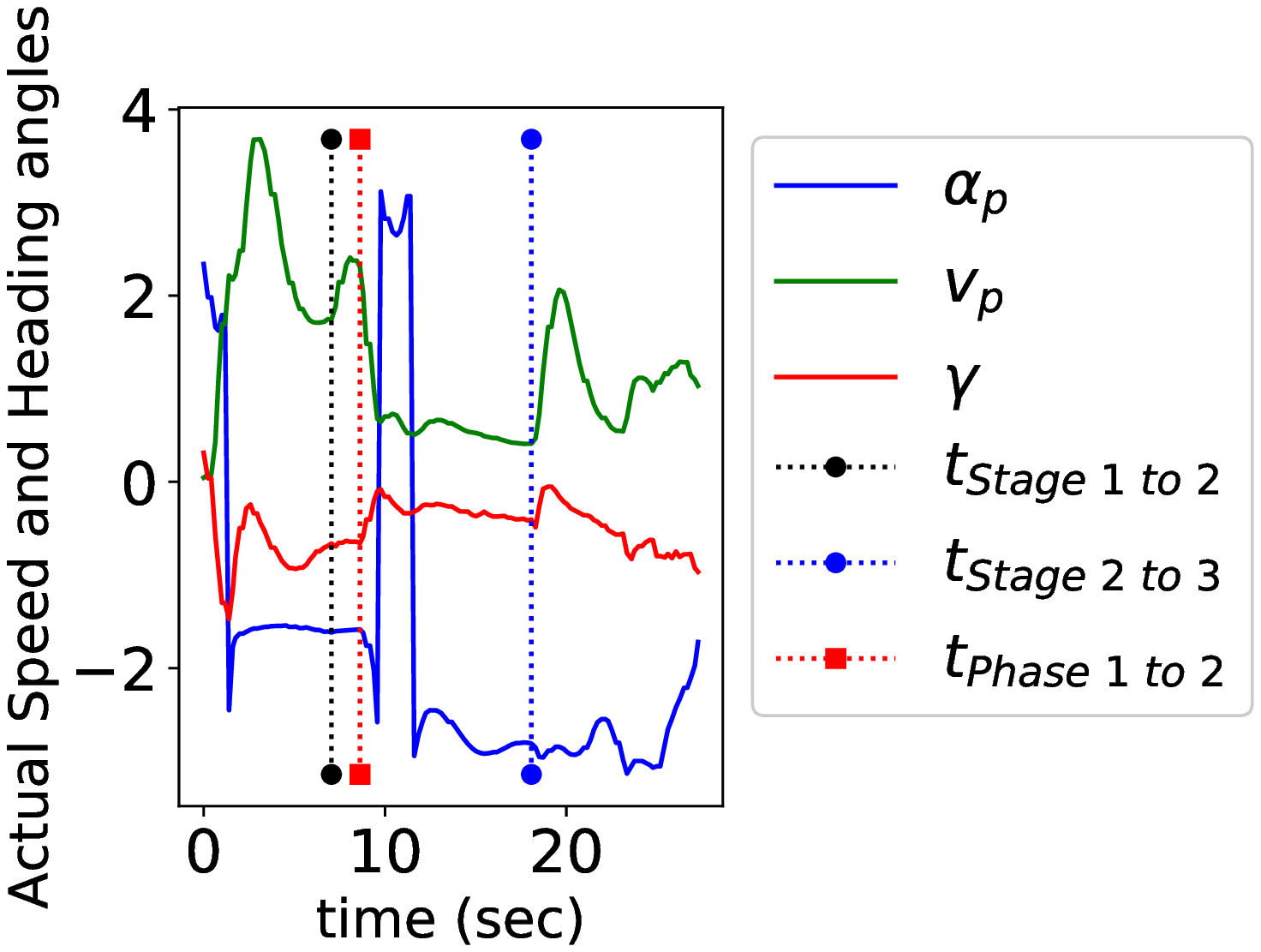}}%
\hfill 
\subcaptionbox{Non-Maneuvering Target\label{fig:Gslineheadspeed}}{\includegraphics[width=.2\textwidth,height=9.5cm,keepaspectratio,trim={1cm 0.3cm 0.0cm .08cm}]{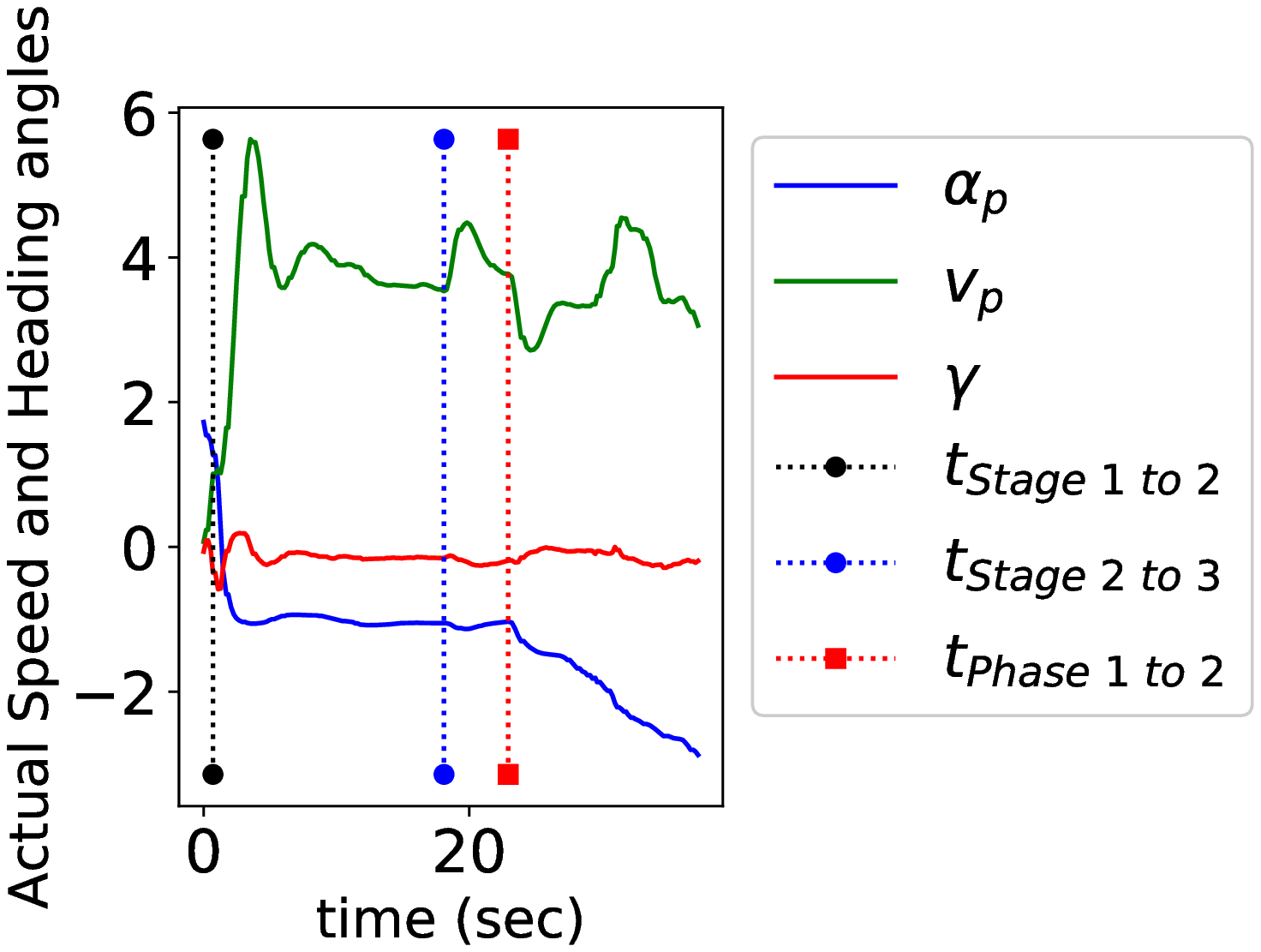}}%
\hfill 
\subcaptionbox{Constant Maneuvering Target\label{fig:Gcircularheadspeed}}{\includegraphics[width=.2\textwidth,height=9.5cm,keepaspectratio,trim={1cm 0.3cm 0.0cm .08cm}]{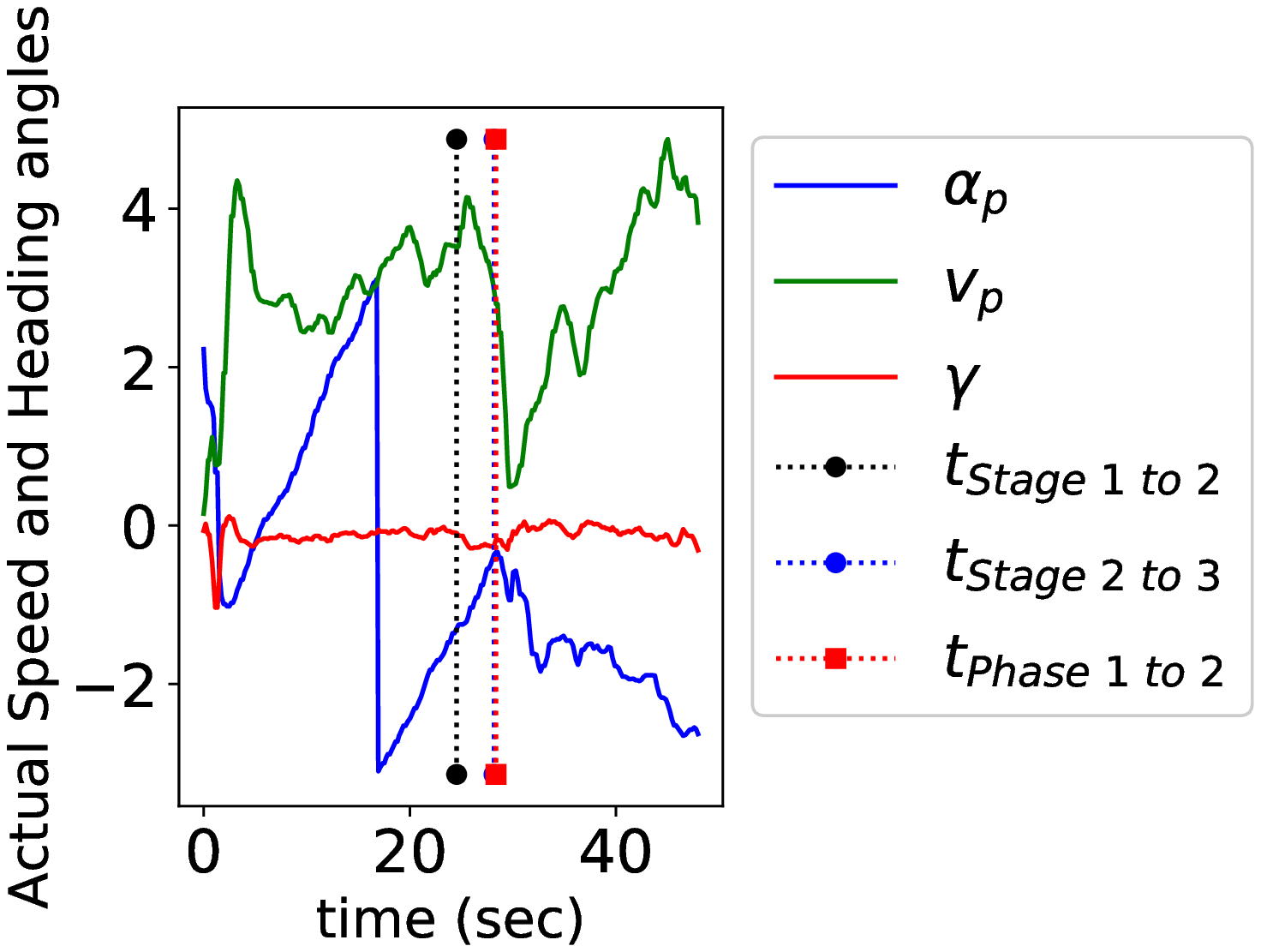}}%
\hfill 
\subcaptionbox{Sinusoidally Maneuvering Target\label{fig:Gsinusoidalheadspeed}}{\includegraphics[width=.2\textwidth,height=9.5cm,keepaspectratio,trim={1cm 0.3cm 0.0cm .08cm}]{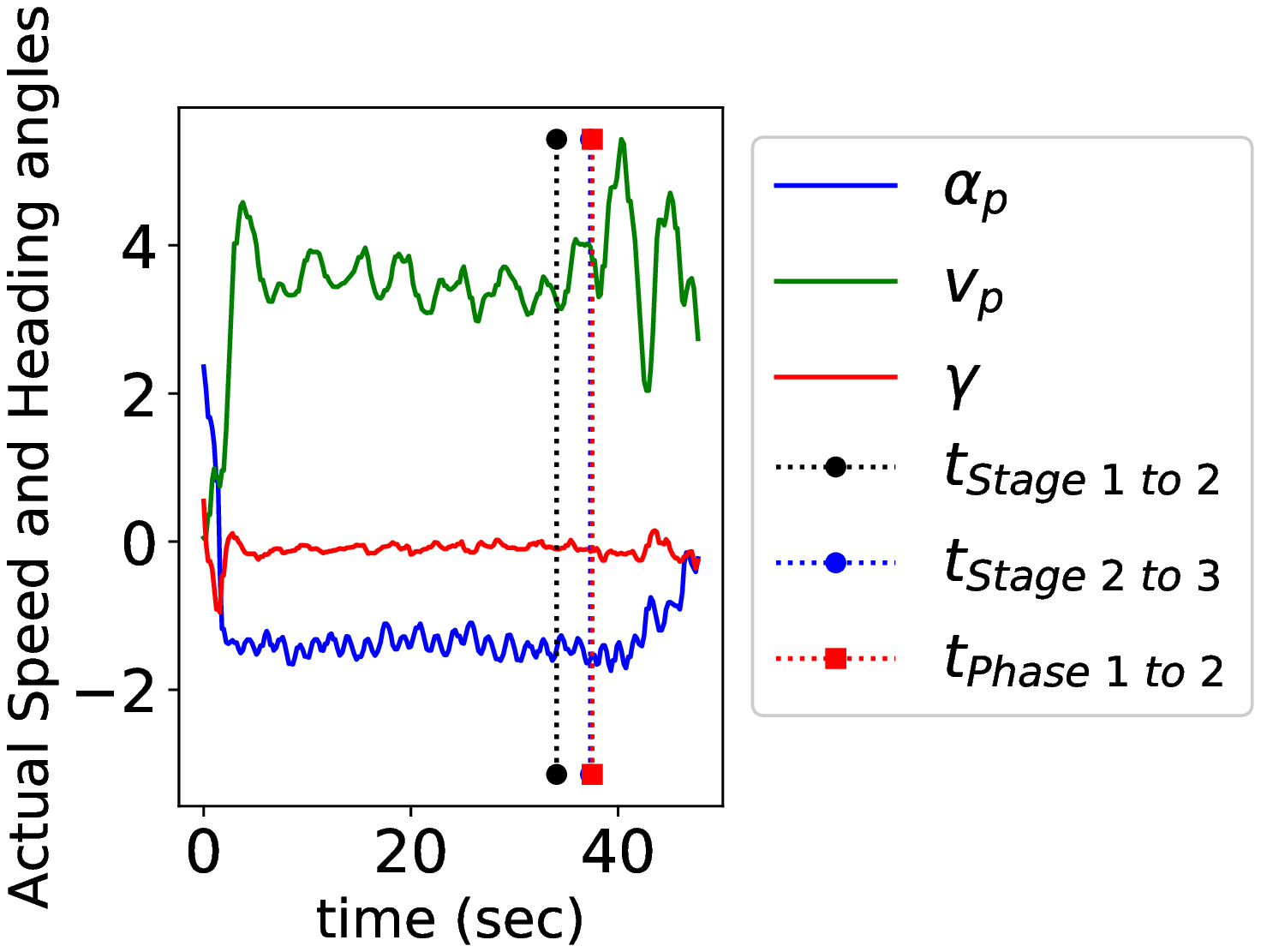}}%
\caption{UAV Speed and heading}
\label{fig:Gheadspeeds}
\end{figure*}
\begin{figure*}[h!]
\centering
\subcaptionbox{Stationary Target\label{fig:Gstationarytuning}}{\includegraphics[width=.2\textwidth,height=9.5cm,keepaspectratio,trim={1cm 0.3cm 0.7cm .08cm}]{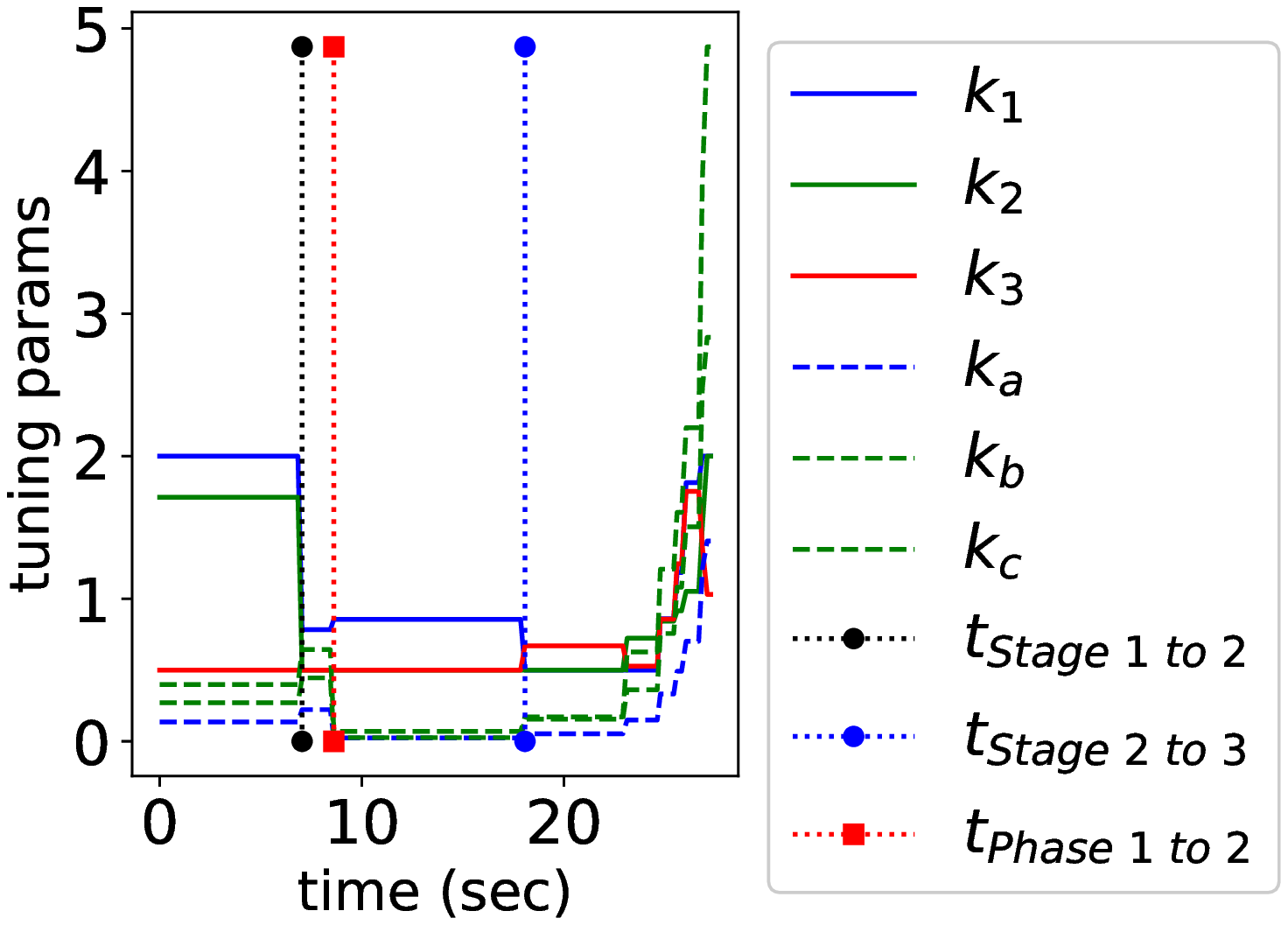}}%
\hfill
\subcaptionbox{Nonmaneuvering Target\label{fig:Gslinetuning}}{\includegraphics[width=.2\textwidth,height=9.5cm,keepaspectratio,trim={1cm 0.3cm 0.7cm .08cm}]{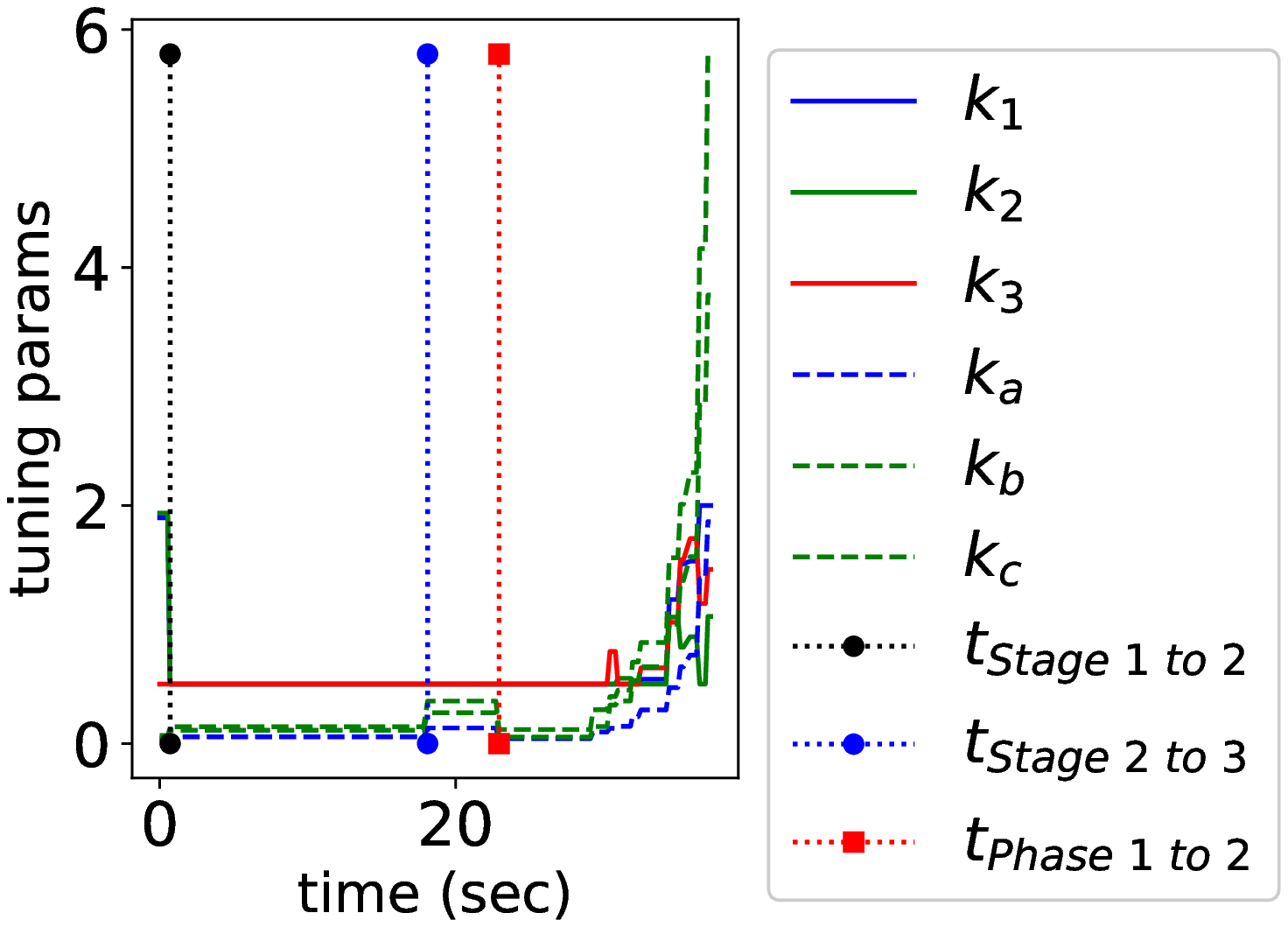}}%
\hfill 
\subcaptionbox{Constant-Maneuvering Target\label{fig:Gcirculartuning}}{\includegraphics[width=.2\textwidth,height=9.5cm,keepaspectratio,trim={1cm 0.3cm 0.7cm .08cm}]{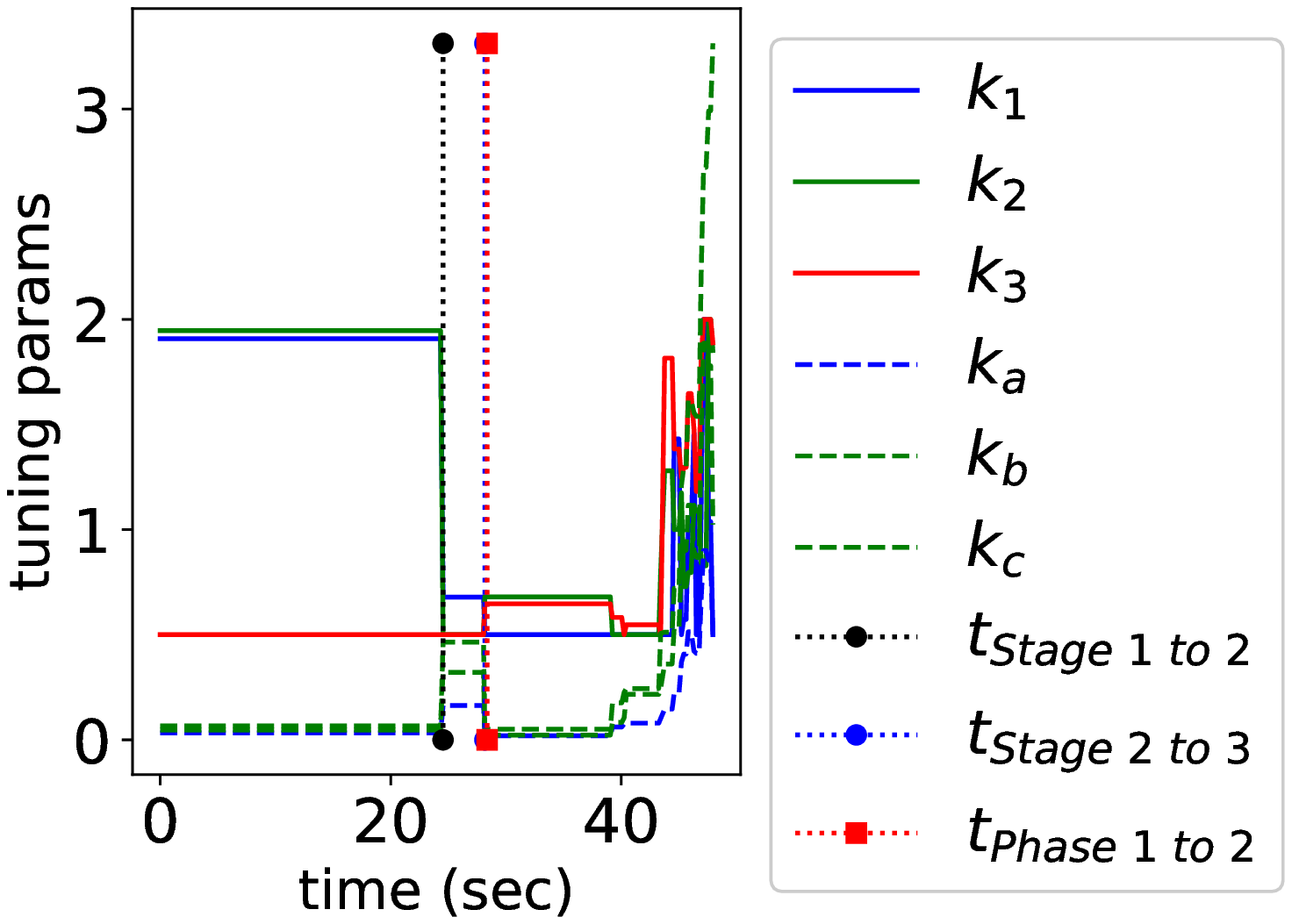}}%
\hfill 
\subcaptionbox{Sinusoidally Maneuvering Target\label{fig:Gsinusoidaltuning}}{\includegraphics[width=.2\textwidth,height=9.5cm,keepaspectratio,trim={1cm 0.3cm 0.7cm .08cm}]{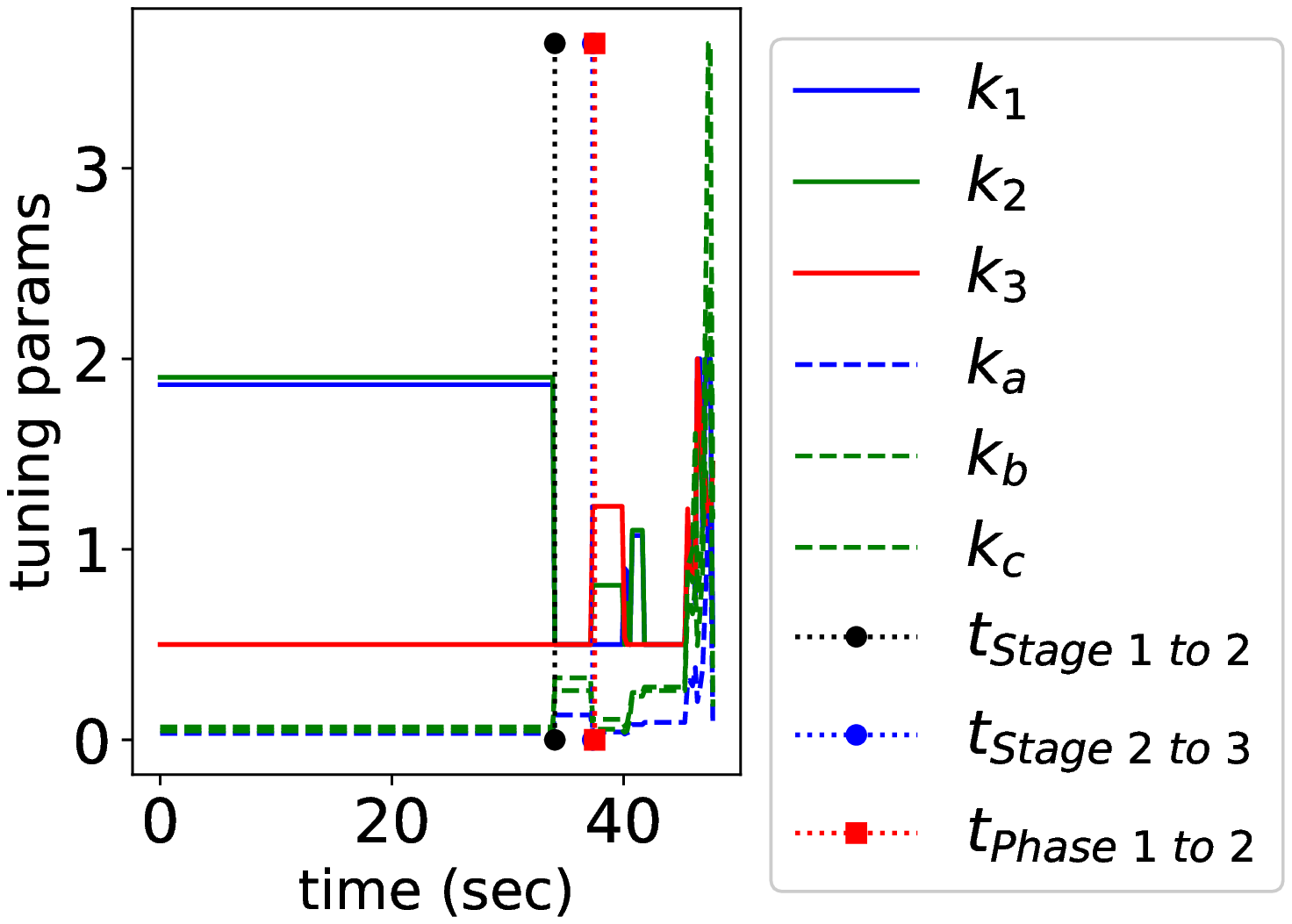}}%
\vspace{-1\baselineskip}
\caption{Guidance parameters}
\label{fig:Gtuningparams}
\end{figure*}

The results of simulations are presented in Figs. \ref{fig:Gtrajs} to \ref{fig:Gtuningparams}. 

Here, the trajectories of the UAV and the target (Fig. \ref{fig:Gtrajs}), the variation of the range between them along with its components and their rates (Fig. \ref{fig:Gdists}), the LOS angles (Fig. \ref{fig:Greferenceangles}), the UAV's velocity commands (Fig. \ref{fig:Gheadspeeds}) with time are depicted.
The guidance parameters tuned over multiple stages are depicted in Fig. \ref{fig:Gtuningparams}. The switching times from Stage 1 to 2 and Stage 2 to 3, and the switching time from Phase 1 to 2 (denoted by $t_{Stage \  1 \  to \  2}$, $t_{Stage \  2 \  to \  3}$, and $t_{Phase \  1 \  to \  2}$, respectively) are depicted by dotted lines plotted along the y-axis in Figs. \ref{fig:Gdists} to \ref{fig:Gtuningparams}.



\subsection{Inferences}
From the results presented above, it can be seen that the algorithm performed with good efficacy. In all cases presented, the UAV's speed ($V_p$), its heading angle ($\alpha_p$), and flight path angle($\gamma$) are not observed to change by a large magnitude in a short time, in any of the four cases presented, as can be seen in the Fig. \ref{fig:Gheadspeeds}. Also, $V_{p}<6$ m/s in all the cases.
Here, it should be noted that, the guidance scheme requires accurate information of $\ddot{\alpha_t}$ at each instant, while in this paper simple euler methods were used for the estimation of $\ddot{\alpha_t}$ from $\dot{\alpha_t}$. This leads to errors in tracking of $\psi_{des}=\alpha_t + \zeta_{des}$, which holds true for any non-constant maneuvering target. As can be seen from Fig. \ref{fig:sinusoidalinput}, $\psi-\psi_{des}$ oscillates at the end of the mission-time, rather than smoothly converging toward zero. However, even in presence of errors in state estimation, along with the inclusion of autopilot and system dynamics, the algorithm is able to nearly achieve desired approach angles ($\psi_{des}$ and $\theta_{des}$) in the four cases presented, as seen in Fig. \ref{fig:Greferenceangles}. Soft-landing on the target is also shown to be achieved in all four cases, since $R_{xy}, R_z,\dot{R}_{xy}, \dot{R}_z$ are close to zero at the end of the simulation as seen from Fig. \ref{fig:Gdists}. This shows that the guidance scheme can be easily implemented in a realistic scenario, that is, on an off-the-shelf autopilot system, without major refactoring.

\section{Conclusions and Future Work}\label{sec:conclusions}
A novel terminal angle-constrained guidance law inspired by the sliding mode control philosophy has been presented in this paper for the autonomous landing of a UAV on stationary, moving, and accelerating ground targets. Stability analysis and a detailed discussion on the selection of guidance parameters have also been presented. By numerical simulation studies conducted in two different ways, i.e., in the absence and the presence of an autopilot system, the guidance law has been shown to effectively achieve soft landing on stationary and maneuvering targets at desired approach angles(both the azimuth and the elevation angles). Due to the portability of the Ardupilot-ROS software, the setup used for implementing the guidance scheme in SITL simulations could be ported to real-world testing platforms with ease, i.e., without any major refactoring. Future works on the presented problem involve improvements in the estimation of the target state vector and experimental validation of the presented guidance algorithm on real test-beds.

\bibliographystyle{IEEEtran}
\bibliography{refs}

\end{document}